%% file: main.tex
\documentclass[11pt,letterpaper]{article}
\usepackage[lmargin=1.0in,rmargin=1.0in,bottom=1.0in,top=1.0in,twoside=False]{geometry}

\usepackage{fullpage,amssymb,amsmath}
\usepackage{graphicx}
\usepackage{tikz}
\usetikzlibrary{shapes}

\usepackage{skull}
\usepackage[T1]{fontenc}

 \usepackage{xcolor}
 \usepackage{mathtools}
\usepackage{microtype}
\usepackage{amsfonts}
\usepackage{comment}
\usepackage[english]{babel}
\usepackage{mathrsfs}
\usepackage[vlined, ruled, linesnumbered]{algorithm2e}
\DontPrintSemicolon

\usepackage{trimspaces}
\usepackage{nccfoots}
\usepackage{setspace}
\usepackage{inconsolata}
\usepackage{libertine}
\usepackage[absolute]{textpos}
\usepackage{thmtools}
\usepackage{thm-restate}

\usepackage{xspace}
\usepackage{enumitem}

\usepackage{todonotes}

\usepackage{longtable}

\definecolor{blue}{rgb}{0.1,0.2,0.5}
\definecolor{brown}{rgb}{0.6,0.6,0.2}
\usepackage[ocgcolorlinks, linkcolor={blue}, citecolor={brown}]{hyperref}

\usepackage{comment}

\usepackage[amsmath,thmmarks,hyperref]{ntheorem}
\usepackage{cleveref}

\crefformat{page}{#2page~#1#3}%
\Crefformat{page}{#2Page~#1#3}%
\crefformat{equation}{#2(#1)#3}%
\Crefformat{equation}{#2(#1)#3}%
\crefformat{figure}{#2Figure~#1#3}%
\Crefformat{figure}{#2Figure~#1#3}%
\crefformat{section}{#2Section~#1#3}
\Crefformat{section}{#2Section~#1#3}
\crefformat{chapter}{#2Chapter~#1#3}
\Crefformat{chapter}{#2Chapter~#1#3}
\crefformat{chapter*}{#2Chapter~#1#3}
\Crefformat{chapter*}{#2Chapter~#1#3}
\crefformat{part}{#2Part~#1#3}
\Crefformat{part}{#2Part~#1#3}
\crefformat{enumi}{#2(#1)#3}
\Crefformat{enumi}{#2(#1)#3}

\usepackage{latexsym}

\theoremnumbering{arabic}
\theoremstyle{plain}
\theoremsymbol{}
\theorembodyfont{\itshape}
\theoremheaderfont{\normalfont\bfseries}
\theoremseparator{.}

\newtheorem{theorem}{Theorem}
\crefformat{theorem}{#2Theorem~#1#3}
\Crefformat{theorem}{#2Theorem~#1#3}

\newcommand{\newtheoremwithcrefformat}[2]{%
  \newtheorem{#1}[theorem]{#2}%
  \crefformat{#1}{##2\MakeUppercase#1~##1##3}%
  \Crefformat{#1}{##2\MakeUppercase#1~##1##3}%
}
\newcommand{\newseptheoremwithcrefformat}[2]{%
  \newtheorem{#1}{#2}%
  \crefformat{#1}{##2\MakeUppercase#1~##1##3}%
  \Crefformat{#1}{##2\MakeUppercase#1~##1##3}%
}

\newtheoremwithcrefformat{lemma}{Lemma}
\newtheoremwithcrefformat{proposition}{Proposition}
\newtheoremwithcrefformat{observation}{Observation}
\newtheoremwithcrefformat{conjecture}{Conjecture}
\newtheoremwithcrefformat{corollary}{Corollary}
\newseptheoremwithcrefformat{claim}{Claim}
\theorembodyfont{\upshape}
\newtheoremwithcrefformat{example}{Example}
\newtheoremwithcrefformat{remark}{Remark}
\newseptheoremwithcrefformat{definition}{Definition}
\newseptheoremwithcrefformat{question}{Question}

\theoremstyle{nonumberplain}
\theoremheaderfont{\scshape}
\theorembodyfont{\normalfont}
\theoremsymbol{\ensuremath{\square}}
\newtheorem{proof}{Proof}

\theoremsymbol{\ensuremath{\lrcorner}}
\newtheorem{claimproof}{Proof of Claim}

\def\cqedsymbol{\ifmmode$\lrcorner$\else{\unskip\nobreak\hfil
\penalty50\hskip1em\null\nobreak\hfil$\lrcorner$
\parfillskip=0pt\finalhyphendemerits=0\endgraf}\fi}

\tikzset{
    position/.style args={#1:#2 from #3}{
        at=(#3.#1), anchor=#1+180, shift=(#1:#2)
    }
}

\newcommand{\Bc}{\mathcal{B}}

\newcommand{\Cc}{\mathcal{C}}

\newcommand{\Lc}{\mathcal{L}}

\newcommand{\Fc}{\mathcal{F}}
\newcommand{\Rc}{\mathcal{R}}
\newcommand{\Wc}{\mathcal{W}}
\newcommand{\dist}{\mathsf{dist}}

\newcommand{\Oh}{\mathcal{O}}
\newcommand{\eps}{\epsilon}
\renewcommand{\epsilon}{\varepsilon}

\let\originalleft\left
\let\originalright\right
\renewcommand{\left}{\mathopen{}\mathclose\bgroup\originalleft}
\renewcommand{\right}{\aftergroup\egroup\originalright}

\renewcommand{\leq}{\leqslant}
\renewcommand{\geq}{\geqslant}

\renewcommand{\setminus}{-}

\newcommand{\MIP}{\textsc{Maximum Induced Packing}\xspace}
\newcommand{\MIF}{\textsc{Max Induced Forest}\xspace}
\newcommand{\MIS}{\textsc{MWIS}\xspace}

\newcommand{\mso}{$\mathsf{MSO}$\xspace}
\newcommand{\msotwo}{$\mathsf{MSO}_2$\xspace}
\newcommand{\cmsotwo}{$\mathsf{CMSO}_2$\xspace}
\newcommand{\cpmsotwo}{$\mathsf{C}_{\leq p}\mathsf{MSO}_2$\xspace}

\newcommand{\Bb}{\mathcal{B}}

\newcommand{\inc}{\mathrm{inc}}
\newcommand{\wh}[1]{\widehat{#1}}
\newcommand{\forked}[1]{\Breve{#1}}

\newcommand{\Blob}[1]{#1^\circ}

\newcommand{\VtxTaken}{A}
\newcommand{\VtxDeleted}{X}
\newcommand{\VtxFree}{R}

\newcommand{\rcall}{\mathcal{R}}
\newcommand{\qcall}{\mathcal{Q}}
\newcommand{\VtxActive}{W}
\newcommand{\rlevel}{\ell}
\newcommand{\pivot}{\nu}
\newcommand{\quota}{\gamma}
\newcommand{\degord}{\eta}
\newcommand{\poslimit}{\zeta}

\newcommand{\wei}{\mathfrak{w}}

\newcommand{\reach}{\mathrm{set}}
\newcommand{\tw}{\operatorname{tw}}

\newcommand{\Sentences}{\mathsf{Sentences}}
\newcommand{\msotypes}{\mathsf{Types}}
\newcommand{\msotype}{\mathsf{type}}

\newcommand{\typetree}{\mathsf{TypeTree}}
\newcommand{\cc}{\mathrm{cc}}
\newcommand{\extS}{\mathsf{extS}}
\newcommand{\branchpart}{\mathcal{B}}
\newcommand{\dom}{\mathrm{dom}}

\newcommand{\Stab}{\mathsf{S}}

\renewcommand{\phi}{\varphi}

\usepackage{todonotes}

\newcommand{\prz}[2][]{}
\newcommand{\map}[2][]{}
\newcommand{\mip}[2][]{}

\usepackage{lineno}

\begin{document}

\title{Finding large induced sparse subgraphs in $C_{>t}$-free graphs in quasipolynomial time}

\author{
Peter Gartland\thanks{University of California, Santa Barbara, USA, \texttt{petergartland@ucsb.edu}.}
\and
Daniel Lokshtanov\thanks{University of California, Santa Barbara, USA, \texttt{daniello@ucsb.edu}.}
\and
Marcin{}~Pilipczuk\thanks{
 Institute of Informatics, University of Warsaw, Poland, \texttt{malcin@mimuw.edu.pl}.
 This work is 
a part of project CUTACOMBS that has received funding from the European Research Council (ERC) 
under the European Union's Horizon 2020 research and innovation programme (grant agreement No.~714704).
}
\and
Micha\l{}~Pilipczuk\thanks{
 Institute of Informatics, University of Warsaw, Poland, \texttt{michal.pilipczuk@mimuw.edu.pl}.
 This work is 
a part of project TOTAL that has received funding from the European Research Council (ERC) 
under the European Union's Horizon 2020 research and innovation programme (grant agreement No.~677651).
}
\and
Pawe\l{} Rz\k{a}\.{z}ewski\thanks{
 Faculty of Mathematics and Information Science, Warsaw University of Technology, Poland, and Institute of
Informatics, University of Warsaw, Poland, \texttt{p.rzazewski@mini.pw.edu.pl}.
Supported by Polish National Science Centre grant no. 2018/31/D/ST6/00062.
}
}

\begin{titlepage}
\def\thepage{}
\thispagestyle{empty}
\maketitle

\begin{textblock}{20}(0, 11.7)
\includegraphics[width=40px]{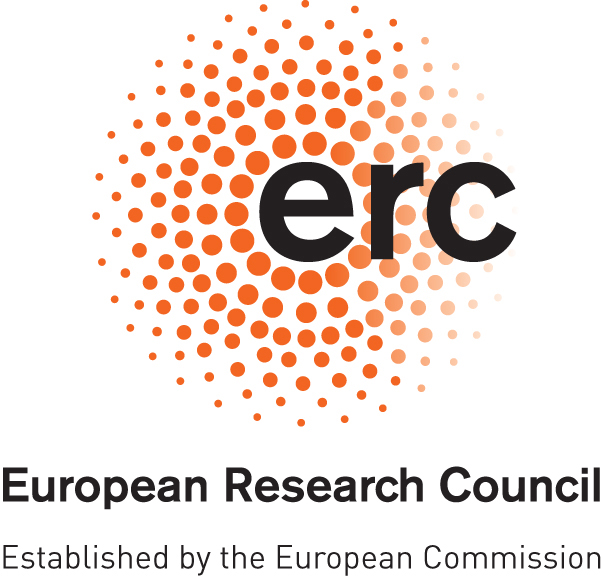}%
\end{textblock}
\begin{textblock}{20}(-0.25, 12.1)
\includegraphics[width=60px]{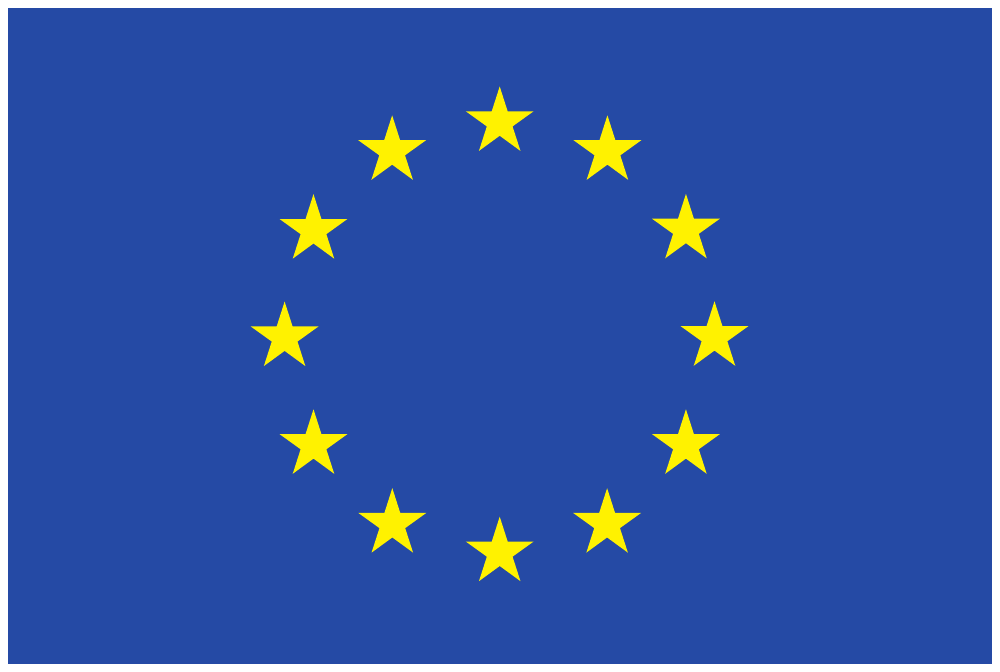}%
\end{textblock}
\input{abstract}
\end{titlepage}

\section{Introduction}\label{sec:intro}
\input{intro}

\section{Overview}\label{sec:overview}
\input{overview}

\section{Preliminaries}\label{sec:prelims}
\input{prelims}

\section{Branching framework}\label{sec:branching}
\input{branching}

\section{Branching strategies: choosing pivots in $P_t$-free and $C_{>t}$-free graphs}\label{sec:pivot}
\input{branching-analysis}

\section{$C_{>t}$-free graphs of bounded degeneracy have bounded treewidth}\label{sec:deg2tw}
\input{deg2tw}


\section{\msotwo and $C_{> t}$-free graphs}\label{sec:longhole}
\input{longhole}

\section{A simple technique for approximation schemes}\label{sec:packing}
\input{packing}

\bibliographystyle{abbrv}
\bibliography{references}

\end{document}

%% file: abstract.tex
\begin{abstract}
For an integer $t$, a graph $G$ is called {\em{$C_{>t}$-free}} if $G$ does not contain any induced cycle on more than~$t$ vertices. We prove the following statement: for every pair of integers $d$ and $t$ and a \cmsotwo statement~$\phi$, there exists an algorithm that, given an $n$-vertex $C_{>t}$-free graph $G$ with weights on vertices, finds in time $n^{\Oh(\log^3 n)}$ a maximum-weight vertex subset $S$ such that $G[S]$ has degeneracy at most $d$ and satisfies $\phi$. The running time can be improved to $n^{\Oh(\log^2 n)}$ assuming $G$ is $P_t$-free, that is, $G$ does not contain an induced path on $t$ vertices. This expands the recent results of the authors [to appear at FOCS 2020 and SOSA 2021] on the {\sc{Maximum Weight Independent Set}} problem on $P_t$-free graphs in two directions: by encompassing the more general setting of $C_{>t}$-free graphs, and by being applicable to a much wider variety of problems, such as {\sc{Maximum Weight Induced Forest}} or {\sc{Maximum Weight Induced Planar Graph}}.
\end{abstract}

%% file: intro.tex
Consider the {\sc{Maximum Weight Independent Set}} ({\sc{MWIS}}) problem: given a vertex-weighted graph $G$, find an independent set in $G$ that has the largest possible weight. While ${\mathsf{NP}}$-hard in general, the problem becomes more tractable when structural restrictions are imposed on the input graph $G$. In this work we consider restricting $G$ to come from a fixed {\em{hereditary}} (closed under taking induced subgraphs) class $\Cc$. The goal is to understand how the complexity of {\sc{MWIS}}, and of related problems, changes with the class $\Cc$. A concrete instance of this question is to consider {\em{$H$-free graphs}} --- graphs that exclude a fixed graph $H$ as an induced subgraph --- and classify for which $H$, {\sc{MWIS}} becomes polynomial-time solvable in $H$-free graphs.

Somewhat surprisingly, we still do not know the complete answer to this question. A classic argument of Alekseev~\cite{Alekseev82} shows that {\sc{MWIS}} is ${\mathsf{NP}}$-hard in $H$-free graphs, unless $H$ is a forest of paths and {\em{subdivided claws}}: graphs obtained from the claw $K_{1,3}$ by subdividing each of its edges an arbitrary number of times. The remaining cases are still open apart from several small ones: of $P_5$-free graphs~\cite{LVV}, $P_6$-free graphs~\cite{GrzesikKPP19}, claw-free graphs~\cite{Sbihi80,Minty80}, and fork-free graphs~\cite{Al04,LozinM08}. Here and further on, $P_t$ denotes a path on~$t$~vertices.

On the other hand, there are multiple indications that {\sc{MWIS}} indeed has a much lower complexity in $H$-free graphs, whenever $H$ is a forest of paths and subdivided claws, than in general graphs. Concretely, in this setting the problem is known to admit both a subexponential-time algorithm~\cite{BacsoLMPTL19,qptas-arxiv} 
and a QPTAS~\cite{ChudnovskyPPT20,qptas-arxiv}; note that the existence of such algorithms for general graphs is excluded under standard complexity assumptions. Very recently, the first two authors gave a {\em{quasipolynomial-time}} algorithm for {\sc{MWIS}} in $P_t$-free graphs, for every fixed $t$~\cite{GL20}. The running time was $n^{\Oh(\log^3 n)}$, which was subsequently improved to $n^{\Oh(\log^2 n)}$ by the last three authors~\cite{PPR20SOSA}.

A key fact that underlies most of the results stated above is that $P_t$-free graphs admit the following balanced separator theorem (see \cref{thm:Pt-free-balanced}): In every $P_t$-free graph, we can find a connected set $X$ consisting of at most $t$ vertices, such that the number of vertices in every connected component of $G-N[X]$ is at most half of the number of vertices of $G$. It has been observed by Chudnovsky et al.~\cite{ChudnovskyPPT20} that the same statement is true also in the class of {\em{$C_{>t}$-free graphs}}: graphs that do not contain an induced cycle on more than $t$ vertices. Note here that, on one hand, every $P_t$-free graph is $C_{>t}$-free as well, and, on the other hand, $C_{>t}$-free graphs generalize the well-studied class of {\em{chordal graphs}}, which are exactly $C_{>3}$-free. Using the separator theorem, Chudnovsky et al.~\cite{ChudnovskyPPT20,qptas-arxiv} gave a subexponential-time algorithm and a QPTAS for {\sc{MWIS}} on $C_{>t}$-free graphs, for every fixed $t$.

The basic toolbox developed for {\sc{MWIS}} can also be  applied to other problems of similar nature. Consider, for instance, the {\sc{Maximum Weight Induced Forest}} problem: in a given vertex-weighted graph $G$, find a maximum-weight vertex subset that induces a forest; note that by duality, this problem is equivalent to {\sc{Feedback Vertex Set}}. By lifting techniques used to solve {\sc{MWIS}} in polynomial time in $P_5$-free and $P_6$-free graphs~\cite{LVV,GrzesikKPP19}, Abrishami et al.~\cite{ACPRzS} showed that {\sc{Maximum Weight Induced Forest}} is polynomial-time solvable both in $P_5$-free and in $C_{>4}$-free graphs. In fact, the result is even more general: it applies to every problem of the form ``find a maximum-weight induced subgraph of treewidth at most $k$''; {\sc{MWIS}} and {\sc{Maximum Weight Induced Forest}} are particular instantiations for $k=0$ and~$k=1$,~respectively. 

As far as subexponential-time algorithms are concerned, Novotn\'a et al.~\cite{NovotnaOPRLW19} showed how to use separator theorems to get subexponential-time algorithms for any problem of the form ``find the largest induced subgraph belonging to $\Cc$'', where $\Cc$ is a fixed hereditary class of graphs that have a linear number of edges. The technique applies both to $P_t$-free and $C_{>t}$-free graphs under the condition that the problem in question
 admits an algorithm which is single-exponential in the treewidth of the instance graph.

\paragraph*{Our results.} We extend the recent results on quasipolynomial-time algorithms for {\sc{MWIS}} in $P_t$-free graphs~\cite{GL20,PPR20SOSA} in two directions:
\begin{itemize}[nosep]
 \item[(a)] We expand the area of applicability of the techniques to $C_{>t}$-free graphs.
 \item[(b)] We show how to solve in quasipolynomial time not only the {\sc{MWIS}} problems, but a whole family of problems that can be, roughly speaking, described as finding a maximum-weight induced subgraph that is sparse and satisfies a prescribed property. 
\end{itemize}
Both of these extensions require a significant number of new ideas.
Formally, we prove the following.

\begin{restatable}{theorem}{mainthm}
\label{thm:main}
Fix a pair of integers $d$ and $t$ and a \cmsotwo sentence $\phi$. Then there exists an algorithm that, given a $C_{>t}$-free $n$-vertex graph $G$
and a weight function $\wei\colon V(G) \to \mathbb{N}$, in time $n^{\Oh(\log^3 n)}$ finds a subset $S$ of vertices such that $G[S]$ is $d$-degenerate, $G[S]$ satisfies $\phi$, and, subject to the above, $\wei(S)$ is maximum~possible; the algorithm may also conclude that no such vertex subset exists. The running time can be improved to $n^{\Oh(\log^2 n)}$ if $G$ is $P_t$-free.
\end{restatable}

Recall here that a graph $G$ is {\em{$d$-degenerate}} if every subgraph of $G$ contains a vertex of degree at most~$d$; for instance, $1$-degenerate graphs are exactly forests and every planar graph is $5$-degenerate. Also, \cmsotwo is the {\em{Monadic Second Order}} logic of graphs with quantification over edge subsets and modular predicates, which is a standard logical language for formulating graph properties. In essence, the logic allows quantification over single vertices and edges as well as over subsets of vertices and of edges. In atomic expressions one can check whether an edge is incident to a vertex, whether a vertex/edge belongs to a vertex/edge subset, and whether the cardinality of some set is divisible by a fixed modulus. We refer to~\cite{CourcelleE12} for a broader introduction.

\paragraph*{Corollaries.}
By applying \cref{thm:main} for different sentences $\phi$, we can model various problems of interest. For instance, as $1$-degenerate graphs are exactly forests, we immediately obtain a quasipolynomial-time algorithm for the {\sc{Maximum Weight Induced Forest}} problem in $C_{>t}$-free graphs. Further, as being planar is expressible in \cmsotwo and planar graphs are $5$-degenerate, we can conclude that the problem of finding a maximum-weight induced planar subgraph can be solved in quasipolynomial time on $C_{>t}$-free graphs. In \cref{sec:generalization} we give a generalization of \cref{thm:main} that allows counting the weights only on a subset of $S$. From this generalization it follows that  for instance the following problem can be solved in quasipolynomial time on $C_{>t}$-free graphs: find the largest collection of pairwise nonadjacent induced~cycles.\prz{perhaphs some explanation why this is possible with the more general statement, but not with the basic one?}\mip{I think it's not worth delving into at this point.}

Let us point out a particular corollary of \cref{thm:main} of a more general nature.
It is known that for every pair of integers $d$ and $t$ there exists $\ell = \ell(d,t)$ such that every graph that contains $P_{\ell}$ \emph{as a subgraph}, contains either $K_{d+2}$, or $K_{d+1,d+1}$, or $P_t$ as an induced subgraph~\cite{AtminasLR12}.
Since the degeneracy of $K_{d+2}$ and $K_{d+1,d+1}$ is larger than $d$, we conclude that every $P_t$-free graph of degeneracy at most $d$ does not contain $P_{\ell}$ as a subgraph.
On the other hand, for every integer $\ell$, the class of graphs that do not contain $P_{\ell}$ as a subgraph is well-quasi-ordered by the induced subgraph relation~\cite{Ding92}. 
It follows that for every pair of integers $t$ and $d$ and every hereditary class $\Cc_d$ such that every graph in $\Cc_d$ has degeneracy at most $d$,
the class $\Cc_d\cap (P_t\textrm{-free})$ of $P_t$-free graphs from $\Cc_d$ is characterized by a finite number of forbidden induced subgraphs: there exists a finite list $\mathcal{F}$ of graphs such that a graph $G$ belongs to $\Cc_d\cap (P_t\textrm{-free})$ if and only if $G$ does not contain any graph from $\mathcal{F}$ as an induced subgraph. As admitting a graph from $\mathcal{F}$ as an induced subgraph can be expressed by a \cmsotwo sentence, from \cref{thm:main} we can conclude the following.

\begin{theorem}
 Let $\Cc$ be a hereditary graph class such that each member of $\Cc$ is $d$-degenerate, for some integer $d$. Then for every integer $t$ there exists algorithm that, given a $P_t$-free $n$-vertex graph $G$ and a weight function $\wei\colon V(G) \to \mathbb{N}$, in time $n^{\Oh(\log^2 n)}$ finds a subset $S$ of vertices such that $G[S]\in \Cc$ and, subject to this, $\wei(S)$ is maximum~possible.
\end{theorem}

\paragraph*{Degeneracy and treewidth.} Readers familiar with the literature on algorithmic results for \cmsotwo logic might be slightly surprised by the statement of \cref{thm:main}. Namely, \cmsotwo is usually associated with graphs of bounded treewidth, where the tractability of problems expressible in this logic is asserted by Courcelle's Theorem~\cite{DBLP:journals/iandc/Courcelle90}. \cref{thm:main}, however, speaks about \cmsotwo-expressible properties of graphs of bounded degeneracy. While degeneracy is upper-bounded by treewidth, in general there are graphs that have bounded degeneracy and arbitrarily high treewidth.
However, we prove that in the case of $C_{>t}$-free graphs, the two notions are functionally equivalent.

\begin{restatable}{theorem}{degistw}
\label{thm:deg2tw}
For every pair of integers $d$ and $t$, there exists an integer $k = (dt)^{\Oh(t)}$
such that every $C_{>t}$-free graph of degeneracy at most $d$
has treewidth at most $k$.
\end{restatable}

As the properties of having treewidth at most $k$ and having degeneracy at most $d$ are expressible in \cmsotwo, from \cref{thm:deg2tw} it follows that in the statement of \cref{thm:main}, assumptions ``$G[S]$ has degeneracy at most $d$'' and ``$G[S]$ has treewidth at most $k$'' could be replaced by one another. Actually, both ways of thinking will become useful in the proof.

\paragraph*{Simple QPTASes.} As an auxiliary result, we also show a simple technique for turning algorithms for {\sc{MWIS}} in $P_t$-free and $C_{>t}$-free graphs into approximation schemes for (unweighted) problems of the following form: in a given graph, find the largest induced subgraph belonging to $\Cc$, where $\Cc$ is a fixed graph class that is closed under taking disjoint unions and induced subgraphs and is {\em{weakly hyperfinite}}~\cite[Section 16.2]{sparsity}. This last property is formally defined as follows: for every $\eps>0$, there exists a constant $c(\eps)$ such that from every graph $G\in \Cc$ one can remove an $\eps$ fraction of vertices so that every connected component of the remaining graph has at most $c(\eps)$ vertices. Weak hyperfiniteness is essentially equivalent to admitting sublinear balanced separators, so all the well-known classes of sparse graphs, e.g. planar graphs or all proper minor-closed classes, are weakly hyperfinite. We present these results in \cref{sec:packing}.

\paragraph*{3-Coloring.} In~\cite{PPR20SOSA}, it is shown how to modify the quasipolynomial-time algorithm for \textsc{MWIS} in $P_t$-free graphs
to obtain an algorithm for \textsc{3-Coloring} with the same asymptotic running time bound in the same graph class. 
We remark here that the same modification can be applied to the algorithm of \cref{thm:main}, obtaining the following:
\begin{theorem}
For every integer $t$ there exists an algorithm that, given an $n$-vertex $C_{>t}$-free graph $G$,
    runs in time $n^{\Oh(\log^3 n)}$ and verifies whether $G$ is 3-colorable.
\end{theorem}

%% file: overview.tex
In this section we present an overview of the proof of our main result, \cref{thm:main}. We try to keep the description non-technical, focusing on explaining the main ideas and intuitions. Complete and formal proofs follow in subsequent sections.

\subsection{Approach for $P_t$-free graphs}\label{sec:Pt-free}

We need to start by recalling the basic idea of the quasipolynomial-time algorithm for {\sc{MWIS}} in $P_t$-free graphs~\cite{GL20,PPR20SOSA}; we choose to follow the exposition of~\cite{PPR20SOSA}. The main idea is to exploit the following balanced separator theorem.

\begin{theorem}[Gy\'arf\'as~\cite{gyarfas}, Bacs\'o et al.~\cite{BacsoLMPTL19}]\label{thm:Pt-free-balanced}
 Let $G$ be an $n$-vertex $P_t$-free graph. Then there exists a set $X$ consisting of at most $t$ vertices of $G$ such that $G[X]$ is connected and every connected component of $G-N[X]$ has at most $n/2$~vertices.
Furthermore, such a set can be found in polynomial time.
\end{theorem}

In the {\sc{MWIS}} problem, there is a natural branching strategy that can be applied on any vertex $u$. Namely, branch into two subproblems: in one subproblem --- {\em{success branch}} --- assume that $u$ is included in an optimal solution, and in the other --- {\em{failure branch}} --- assume it is not. In the success branch we can remove both $u$ and all its neighbors from the consideration, while in the failure branch only $u$ can be removed. Hence, \cref{thm:Pt-free-balanced} suggests the following naive Divide\&Conquer strategy: find a set $X$ as provided by the Theorem and branch on all the vertices of $X$ as above in order to try to disconnect the graph. This strategy does {\em{not}} lead to any reasonable algorithm, because the graph would get shattered only in the subproblem corresponding to success branches for all $x\in X$. However, there is an intuition that elements of $X$ are reasonable candidates for {\em{branching pivots}}: vertices such that branching on them leads to a significant progress of the algorithm.

The main idea presented in~\cite{PPR20SOSA} is to perform branching while measuring the progress in disconnecting the graph in an indirect way. Let $G$ be the currently considered graph. For a pair of vertices $u$ and $v$, let the {\em{bucket}} of $u$ and $v$ be defined as:
$$\Bc^G_{u,v}\coloneqq \{\,P\ \colon\ P \textrm{ is an induced path in }G\textrm{ with endpoints }u\textrm{ and }v\}.$$
Observe that since $G$ is $P_t$-free, every element of $\Bc^G_{u,v}$ is a path on fewer than $t$ vertices, hence $\Bc^G_{u,v}$ has always at most $n^{t-1}$ elements and can be computed in polynomial time (for a fixed $t$). On the other hand, $\Bc^G_{u,v}$ is nonempty if and only if $u$ and $v$ are in the same connected component of $G$.

Let $X$ be a set whose existence is asserted by \cref{thm:Pt-free-balanced}.
Observe that if $u$ and $v$ are in different components of $G-N[X]$, then {\em{all}} the paths of $\Bc^G_{u,v}$ are intersected by $N[X]$. Moreover, as every connected component of $G-N[X]$ has at most $n/2$ elements, this happens for at least half of the pairs $\{u,v\}\in \binom{V(G)}{2}$. Since $X$ has only at most $t$ vertices, by a simple averaging argument we conclude the following.

\begin{claim}\label{cl:heavy-Pt}
 There is a vertex $x$ such that $N[x]$ intersects at least a $\frac{1}{t}$ fraction of paths in at least $\frac{1}{2t}$ fraction of~buckets.
\end{claim}

A vertex $x$ having the property mentioned in \cref{cl:heavy-Pt} shall be called {\em{$\frac{1}{2t}$-heavy}}, or just {\em{heavy}}. Then \cref{cl:heavy-Pt} asserts that there is always a heavy vertex; note that such a vertex can be found in polynomial time by inspecting the vertices of $G$ one by one.

We may now present the algorithm:
\begin{enumerate}[nosep]
 \item If $G$ is disconnected, then apply the algorithm to every connected component of $G$ separately.
 \item Otherwise, find a heavy vertex in $G$ and branch on it.
\end{enumerate}

We now sketch a proof of the following claim: on each root-to-leaf path in the recursion tree, this algorithm may execute only $\Oh(\log^2 n)$ success branches. By \cref{cl:heavy-Pt}, in each success branch a constant fraction of buckets get their sizes reduced by a constant multiplicative factor. Since buckets are of polynomial size in the first place, after $\Omega(\log n)$ success branches a $\frac{1}{10}$ fraction of the initial buckets must become empty. Since in a connected graph all the buckets are nonempty, it follows that after $\Omega(\log n)$ success branches, the vertex count of the connected graph we are working on must have decreased by at least a multiplicative factor of $0.01$ with respect to the initial graph. As this can happen only $\Oh(\log n)$ times, the claim follows.

Now the recursion tree has depth at most $n$ and each root-to-leaf path contains at most $\Oh(\log^2 n)$ success branches.
Therefore, the total size of the recursion tree is $n^{\Oh(\log^2 n)}$, which implies the same bound on the running time. This concludes the description of the algorithm for $P_t$-free~graphs; let us recall that this algorithm was already presented in~\cite{PPR20SOSA}.

\subsection{Lifting the technique to $C_{>t}$-free graphs}

We now explain how to lift the technique presented in the previous section to the setting of $C_{>t}$-free graphs. As we mentioned before, the main ingredient --- the balanced separator theorem --- remains true.

\begin{theorem}[Gy\'arf\'as~\cite{gyarfas}, Chudnovsky et al.~\cite{ChudnovskyPPT20}]\label{lem:separator}
Let $G$ be an $n$-vertex $C_{>t}$-free graph.
Then there is a set $X$ consisting of at most $t$ vertices of $G$ such that $G[X]$ is connected and every connected component of $G - N[X]$ has at most $n/2$ vertices.
Furthermore, such a set can be found in polynomial time.
\end{theorem}

However, in the previous section we used the $P_t$-freeness of the graph in question also in one other place: to argue that the buckets $\Bc^G_{u,v}$ are of polynomial size. This was crucial for the argument that $\Omega(\log n)$ success branches on heavy vertices lead to emptying a significant fraction of the buckets. Solving this issue requires reworking the concept of buckets. 

The idea is that in the $C_{>t}$-free case, the objects placed in buckets will connect triples of vertices, rather than pairs. Formally, a {\em{connector}} is a graph formed from three disjoint paths $Q_1,Q_2,Q_3$ by picking one endpoint $a_i$ of $Q_i$, for each $i=1,2,3$, and either identifying vertices $a_1,a_2,a_3$ into one vertex, or turning $a_1,a_2,a_3$ into a triangle; see \cref{fig:tripod}. The paths $Q_i$ are the {\em{legs}} of the connector, the other endpoints of the legs are the {\em{tips}}, and the (identified or not) vertices $a_1,a_2,a_3$ are the {\em{center}} of the connector. We remark that we allow the degenerate case when one or more paths $Q_1,Q_2,Q_3$ has only one vertex, but we require the tips to be pairwise distinct.

The following claim is easy to prove by considering any inclusion-wise minimal connected induced subgraph containing $u,v,w$.

\begin{claim}\label{cl:connector-connected}
 If vertices $u,v,w$ belong to the same connected component of a graph $G$, then in $G$ there is an induced connector with tips $u,v,w$.
\end{claim}

A {\em{tripod}} is a connector in which every leg has length at most $t/2+1$ (w.l.o.g. $t$ is even). Every connector contains a {\em{core}}: the tripod induced by the vertices at distance at most $t/2$ from the center. The next claim is the key observation that justifies looking at connectors and tripods.  

\begin{claim}\label{cl:tripod-hit}
 Let $G$ be a $C_{>t}$-free graph, let $T$ be an induced connector in $G$, and let $X$ be a subset of vertices such that $G[X]$ is connected and no two tips of $T$ are in the same connected component of $G-N[X]$. Then $N[X]$ intersects the core of $T$. 
\end{claim}
\begin{claimproof}
 Since no two tips of $T$ lie in the same component of $G-N[X]$, it follows that $N[X]$ intersects at least two legs of $T$, say $Q_1$ and $Q_2$ at vertices $q_1$ and $q_2$, respectively. We may choose $q_1$ and $q_2$ among $N[X]\cap V(Q_1)$ and $N[X]\cap V(Q_2)$ so that they are as close in $T$ as possible to the center of $T$. Since $G[X]$ is connected, there exists a path $P$ with endpoints $q_1$ and~$q_2$ such that all the internal vertices of $P$ belong to $X$. Now $P$ together with the shortest $q_1$-$q_2$ path within $T$ form an induced cycle in~$G$. As this cycle must have at most $t$ vertices, we conclude that $q_1$ or $q_2$ belongs to the core of $T$.
\end{claimproof}

\cref{cl:tripod-hit} suggests that in $C_{>t}$-free graphs, cores of connectors are objects likely to be hit by balanced separators provided by \cref{lem:separator}, similarly as in $P_t$-free graphs, induced paths were likely to be hit by balanced separators given by \cref{thm:Pt-free-balanced}. Let us then use cores as objects for defining buckets.

Let $G$ be a $C_{>t}$-free graph.
For an unordered triple $\{u,v,w\}\in \binom{V(G)}{3}$ of distinct vertices, we define the bucket $\Bc^G_{u,v,w}$ as the set of all cores of all induced connectors with tips $u,v,w$. Let us stress here that $\Bc^G_{u,v,w}$ is a set, not a multiset, of tripods: even if some tripod is the core of multiple connectors with tips $u,v,w$, it is included in $\Bc^G_{u,v,w}$ only once. Therefore, as each tripod has $\Oh(t)$ vertices, the buckets are again of size $n^{\Oh(t)}$ and can be enumerated in polynomial time. By \cref{cl:connector-connected}, the bucket $\Bc^G_{u,v,w}$ is nonempty if and only if $u,v,w$ are in the same connected component of $G$. Moreover, from \cref{cl:tripod-hit} we  infer the following.

\begin{claim}
 Let $\{u,v,w\}\in \binom{V(G)}{3}$ be a triple of vertices of $G$ and let $X$ be a vertex subset such that $G[X]$ is connected and no two vertices out of $u,v,w$ belong to the same connected component of $G-N[X]$. Then $N[X]$ intersects all the tripods in the bucket $\Bc^G_{u,v,w}$.
\end{claim}

Now we would like to obtain an analogue of \cref{cl:heavy-Pt}, that is, find a vertex $x$ such that $N[x]$ intersects a significant fraction of tripods in a significant fraction of buckets. Let then $X$ be a set provided by \cref{lem:separator} for $G$. For a moment, let us assume optimistically that each connected component of $G-N[X]$ contains at most $n/10$ vertices, instead of $n/2$ as promised by \cref{lem:separator}. Observe that if we choose a triple of distinct vertices uniformly at random, then with probability at least $\frac{1}{2}$ no two of these vertices will lie in the same connected component of $G-N[X]$. By \cref{cl:tripod-hit}, this implies that $N[X]$ intersects all the tripods in at least half of the buckets. By the same averaging argument as before, we get the following.

\begin{claim}\label{cl:very-shattered}
 Suppose that in $G$ there is a set $X$ consisting of at most $t$ vertices such that $G[X]$ is connected and every connected component of $G-N[X]$ has at most $n/10$ vertices. Then there is a heavy vertex in $G$. 
\end{claim}

Here, we define a heavy vertex as before: it is a vertex $x$ such that $N[x]$ intersects at least a $\frac{1}{t}$ fraction of tripods in at least a $\frac{1}{2t}$ fraction of buckets.

Unfortunately, our assumption that every component of $G-N[X]$ contains at most $n/10$ vertices, instead of at most $n/2$ vertices, is too optimistic. Consider the following example: $G$ is a path on $n$ vertices. The cores of connectors degenerate to subpaths consisting of at most $t$ consecutive vertices of the path, and for every vertex $x$, the set $N[x]$ intersects any tripod in only an $\Oh(t/n)$ fraction of the buckets. Therefore, in this example there is no heavy vertex at all. We need to resort to a different strategy.

\paragraph*{Secondary branching.}
So let us assume that the currently considered graph $G$ is connected and has no heavy vertex --- otherwise we may either recurse into connected components or branch on the heavy vertex (detectable in polynomial time). We may even assume that there is no {\em{$(10^{-8}/t)$-heavy}} vertex: a vertex $x$ such that $N[x]$ intersects at least a $(10^{-8}/t)$ fraction of tripods in at least a $(10^{-8}/t)$ fraction of buckets buckets. Indeed, branching on such vertices also leads to quasipolynomial running time (with all factors in the analysis appropriately~scaled). 

Let us fix a set $X$ provided by \cref{lem:separator} for $G$; then $G[X]$ is connected and each connected component of $G-N[X]$ has at most $n/2$ vertices. By \cref{cl:very-shattered}, there must be some components of $G-N[X]$ that have more than $n/10$ vertices, for otherwise there would be a heavy vertex.
Let $C$ be such a component and let us apply \cref{lem:separator} again, this time to $G[C]$, obtaining a connected set $Y$ of size at most $t$ such that every connected component of
$G[C] \setminus N[Y]$ has at most $|C|/2$ vertices. If the distance between $X$ and $Y$ is small, say at most $10t$, then one can replace $X$ with the union of $X$, $Y$, and a shortest
path between $X$ and $Y$, and repeat the argument. The new set $X$ is still of size $\Oh(t)$, so the argument of \cref{cl:very-shattered} applies with adjusted constants, 
and the absence of a heavy vertex gives another component $C'$ with more than $n/10$ vertices. 
This process can continue only for a constant number of steps. Hence, at some moment we end up with a connected set $X$ of size $\Oh(t)$ such that every connected component of $G\setminus N[X]$
has at most $n/2$ vertices, a connected component $C$ of $G\setminus N[X]$ with more than $n/10$ vertices, a connected set $Y \subseteq C$ of size at most $t$ such that every connected
component of $G[C] \setminus N[Y]$ has at most $|C|/2$ vertices and the distance between $X$ and $Y$ is more than $10t$. 

The crucial observation now is as follows: there exists exactly one connected component of $G[C] \setminus N[Y]$, call it $D_0$, that is adjacent to a vertex of $N[X]$. 
The existence of at least one such component follows from the connectivity of $G$. If there were two such components, say $D_0$ and $D_1$, then one can construct an induced cycle in $G$
by going from $X$ via $D_0$ to $Y$ and back to $X$ via $D_1$. This cycle is long since the distance between $X$ and $Y$ is more than $10t$, which contradicts $G$ being $C_{\geq t}$-free.
Denote $B := C \setminus D_0$. Note that $B$ is connected and $|B| =  |C|-|D_0| \geq |C|-|C|/2 = |C|/2 \geq n/20$.%
\footnote{We refer to Figure~\ref{fig:secondary} in Section~\ref{sec:secondary} for an illustration.
 Note that in the formal argument of Section~\ref{sec:secondary} the component 
 $C$ is called $C_2$ and the distance between $X$ and $Y$ is lower bounded by $8t$, not $10t$.}

Repeating the same proof as in the previous observation, note that for every induced subgraph $G'$ of $G$, there is at most one component of $G'[V(G') \cap C]$ that contains both a vertex
of $B$ and a neighbor of $N[X]$: If there were two such components, one could construct a long induced cycle by going from $X$ via the first component to $B$ and back to $X$ via the second one. 
If such a component exist, we call it \emph{the chip} of $G'$. 

Note that if $G'$ has no chip, then every connected component of $G'$ contains at most $0.95n$ vertices as $n/20 \leq |B| \leq n/2$.
Thus, the goal of the secondary branching is to get to an induced subgraph that contains no chip, that is, to separate $B$ from $N[X]$. 
The crucial combinatorial insight that we discuss in the next paragraph is that the area of the graph between $N[X]$ and $B$ behaves like a $P_t$-free graph and 
is amenable to the branching strategy for $P_t$-free graphs.

Consider the chip $C'$ in an induced subgraph $G'$ of $G$. A {\em{$C'$-link}} is a path in $G'$ with endpoints in $N[X]\cap N_{G'}(C')$ and all internal vertices in $C'$; this path should be induced, except that we allow the existence of an edge between the endpoints. Observe the following:

\begin{claim}\label{cl:link-short}
 Every $C'$-link has at most $t$ vertices.
\end{claim}
\begin{claimproof}
 Let $P$ be a $C'$-link. Since the endpoints of $P$ are in $N[X]$ and $G[X]$ is connected, there exists an induced path $Q$ in $G[N[X]]$ with same endpoints as $P$ such that all the internal vertices of $P$ are in $X$. Then $P\cup Q$ is an induced cycle in $G$, hence both $P$ and $Q$ must have at most $t$ vertices.
\end{claimproof}

The idea is that in order to cut the chip away, we perform a secondary branching procedure, but this time we use $C'$-links as objects that are hit by neighborhoods of vertices. Formally, for a pair
$\{u,v\}\in \binom{N[X] \cap N_{G'}(C')}{2}$, we consider the {\em{secondary bucket}} $\Lc^{G'}_{u,v}$ consisting of all $C'$-links with endpoints $u$ and $v$.
Again, by \cref{cl:link-short}, each secondary bucket is of size at most $n^t$ and can be enumerated in polynomial time. Note that $\Lc^{G'}_{u,v}$ is nonempty for every distinct vertices $u,v \in N_{G'}(C')$.

We shall say that a vertex $z$ of $G$ is {\em{secondary-heavy}} if $N[z]$ intersects at least a $\frac{1}{t}$ fraction of links in at least a $\frac{1}{2t}$ fraction of nonempty secondary buckets. 

\begin{claim}\label{cl:secondary-heavy}
 If $|N_{G'}(C')| \geq 2$, then there is a secondary-heavy vertex.
\end{claim}
\begin{claimproof}[Sketch]
We apply a weighted variant of \cref{lem:separator} to the graph $G'[N_{G'}[C']]$ in order to find a set $Z\subseteq N_{G'}[C']$ of size at most $t$ such that every connected component of $G'[N_{G'}[C']]-N[Z]$ contains at most half of the vertices of $N_{G'}(C')$. Then $N[Z]$ intersects all the links in at least half of the buckets. The same averaging argument as used before shows that one of vertices of $Z$ is secondary-heavy.
\end{claimproof}

The secondary branching procedure now branches on a secondary-heavy vertex (detectable in polynomial time).
This is always possible by \cref{cl:secondary-heavy} as long as $N_{G'}(C')$ contains at least two vertices. 
If $N_{G'}(C') = \{v\}$ for some vertex $v$, we choose $v$ as the branching pivot and observe that both in the success and the failure branch there is no chip.

The same analysis as in \cref{sec:Pt-free} shows that branching on secondary-heavy vertices results in a recursion tree with $n^{\Oh(\log^2 n)}$ leaves. In each of these leaves 
there is no chip, so every connected component of $G'$ contains at most $0.95n$ vertices.

To summarize, we perform branching on $(10^{-8}/t)$-heavy vertices and recursing on connected components as long as a $(10^{-8}/t)$-heavy vertex can be found.
When this ceases to be the case, we resort to the secondary branching.
Such an application of secondary branching results in producing $n^{\Oh(\log^2 n)}$ subinstances to solve,
and in each of these subinstances the size of the largest connected component is at most $95\%$ of the vertex count of the graph for which the secondary branching was initiated.
We infer that the running time is $n^{\Oh(\log^3 n)}$. This concludes the description of an $n^{\Oh(\log^3 n)}$-time algorithm for \MIS on $C_{>t}$-free graphs.

\subsection{Degeneracy branching}\label{ov:degeneracy}

Our goal in this section is to generalize the approach presented in the previous section to an algorithm solving the following problem: given a vertex-weighted $C_{>t}$-free graph $G$, find a maximum-weight subset of vertices $S$ such that $G[S]$ is $d$-degenerate. Here $d$ and $t$ are considered fixed constants. Thus we allow the solution to be just sparse instead of independent, but, compared to \cref{thm:main}, so far we do not introduce \cmsotwo-expressible properties. Let us call the considered problem {\sc{Maximum Weight Induced $d$-Degenerate Graph}} ({\sc{MWID}}). The algorithm for {\sc{MWID}} that we are going to sketch is formally presented in \cref{sec:branching} and~\cref{sec:pivot}, and is the subject of \cref{thm:deg-branching} there.

Recall  that a graph $G$ is $d$-degenerate if every subgraph of $G$ has a vertex of degree at most $d$. We will rely on the following characterization of degeneracy, which is easy to prove.

\begin{claim}\label{cl:degenerate}
 A graph $G$ is $d$-degenerate if and only if there exists a function $\degord\colon V(G)\to \mathbb{N}$ such that for every $uv\in E(G)$ we have $\degord(u)\neq \degord(v)$ and for each $u\in V(G)$, $u$ has at most $d$ neighbors $v$ with $\degord(v)<\degord(u)$.
\end{claim}

A function $\degord(\cdot)$ satisfying the premise of \cref{cl:degenerate} shall be called a {\em{degeneracy ordering}}. Note that we only require that a degeneracy ordering is injective on every edge of the graph, and not necessarily on the whole vertex set. For a vertex $u$, the value $\degord(u)$ is the {\em{position}} of $u$ and the set neighbors of $u$ with smaller positions is the {\em{left neighborhood}} of $u$.

\medskip

We shall now present a branching algorithm for the {\sc{MWID}} problem. For convenience of exposition, let us fix the given $C_{>t}$-free graph $G$, an optimum solution $S^\star$ in $G$, and a degeneracy ordering $\degord^\star$ of $G[S^\star]$. We may assume that the co-domain of $\degord^\star$ is $[n]\coloneqq \{1,\ldots,n\}$.

Recall that when performing branching for the {\sc{MWIS}} problem, say on a vertex $x$, in the failure branch we were removing $x$ from the graph, while in the success branch we were removing both $x$ and its neighbors. When working with {\sc{MWID}}, we cannot proceed in the same way in the second case, because the neighbors of $x$ can  be still included in the solution. Therefore, instead of modifying the graph $G$ along the recursion, we keep track of two disjoint sets $\VtxTaken$ and $\VtxActive$: $\VtxTaken$ consists of vertices already decided to be included in the solution, while $\VtxActive$ is the set of vertices that are still allowed to be taken to the solution in further steps. Initially, $\VtxTaken=\emptyset$ and $\VtxActive=V(G)$. We shall always branch on a vertex $x\in \VtxActive$: in the failure branch we remove $x$ from $\VtxActive$, while in the success branch we move $x$ from $\VtxActive$ to $\VtxTaken$. The intuition is that moving $x$ to $\VtxTaken$ puts more restrictions on the neighbors of $x$ that are still in $\VtxActive$. This is because they are now adjacent to one more vertex in $\VtxTaken$, and they cannot be adjacent to too many, at least as far as vertices with smaller positions are concerned.

For the positions, during branching we will maintain the following two pieces of information:
\begin{itemize}[nosep]
 \item a function $\degord\colon \VtxTaken\to [n]$  that is our guess on the restriction of $\degord^\star$ to $\VtxTaken$; and
 \item a function $\poslimit\colon \VtxActive\to [n]$ which signifies a {\em{lower bound}} on the position of each vertex of $\VtxActive$, assuming it is to be included in the solution. 
\end{itemize}
Initially, we set $\poslimit(v)=1$ for each $v\in V(G)$.
The quadruple $(\VtxTaken,\VtxActive,\degord,\poslimit)$ as above describes a subproblem solved during the recursion. We will say that such a subproblem is {\em{lucky}} if all the choices made so far are compliant with $S^\star$ and $\degord^\star$, that is,
$$\VtxTaken\subseteq S^\star\subseteq \VtxTaken \cup \VtxActive,\qquad \degord = \degord^\star|_{\VtxTaken},\qquad\textrm{and}\qquad \degord^\star(u)\geq \poslimit(u)\textrm{ for each }u\in S^\star\cap \VtxActive.$$
Additionally to the above, from a lucky subproblem we also require the following property:
\begin{equation}\label{eq:guessed-left}
\textrm{for each }v\in \VtxTaken\textrm{ and }u\in N(v)\cap \VtxActive\textrm{ such that }\poslimit(u)\leq \degord^\star(v)\textrm{, we have }u\in S^\star\textrm{ and }\degord^\star(u)<\degord^\star(v).
\end{equation}
In other words, all the neighbors of a vertex $v\in\VtxTaken$ should have their lower bounds larger than the guessed position of $v$, unless they will be actually included in the solution at positions smaller than that of $v$. The significance of this property will become clear in a moment.

First, observe that if $G[\VtxActive]$ is disconnected, then we can treat the different connected components of $G[\VtxActive]$ separately: for each component $D$ of $G[\VtxActive]$ we solve the subproblem $(\VtxTaken,D,\degord,\poslimit|_D)$ obtaining a solution $S_D$, and we return $\bigcup_D S_D$ as the solution to $(\VtxTaken,\VtxActive,\degord,\poslimit)$. Property~\eqref{eq:guessed-left} is used to guarantee the correctness of this step: it implies that when taking the union of solutions~$S_D$, the vertices of $\VtxTaken$ do not end up with too many left neighbors.

Thus, we may assume that $G[\VtxActive]$ is connected. In this case we execute branching on a vertex of $\VtxActive$. For the choice of the branching pivot $x$ we use exactly the same strategy as described in the previous section: having defined the buckets in exactly the same way, we always pick $x$ to be a heavy vertex in $G[\VtxActive]$, or resort to secondary branching in $G[\VtxActive]$ (which picks secondary-heavy pivots) in the absence of heavy~vertices.

An important observation is that in the success branch --- when the vertex $x\in \VtxActive$ is moved to $\VtxTaken$ --- the algorithm notes a significant progress that allows room for additional guessing (by branching). More precisely, on every root-to-leaf path in the recursion tree there are only $\Oh(\log^3 n)$ success branches, which means that following each success branch we can branch further into $n^{\Oh(1)}$ options, and the size of the recursion tree will be still $n^{\Oh(\log^3 n)}$. We use this power to guess (by branching) the following objects when deciding that $x$ should be included in the solution $S^\star$ (here, we assume that the current subproblem~is~lucky):
\begin{itemize}[nosep]
 \item the position $\degord^\star(x)$;
 \item the set of left neighbors $L=\{v\in \VtxActive\cap N(x)~|~\degord^\star(v)<\degord^\star(x)\}$;\map{Later it's $D$. Shouldn't we make it consistent?}\mip{I stand by my $L$, I want $D$ for components.}
 \item the positions $(\degord^\star(v)\colon v\in L)$; and
 \item for each $v\in L$, its left neighbors $L_v\coloneqq \{u\in \VtxActive\cap N(v)~|~\degord^\star(u)<\degord^\star(v)\}$.
\end{itemize}
This guess is reflected by the following clean-up operations in the subproblem:
\begin{itemize}[nosep]
 \item Move $\{x\}\cup L$ from $\VtxActive$ to $\VtxTaken$ and set their positions in $\degord(\cdot)$ as the guess prescribes. Note that the vertices of $\bigcup_{v \in L} L_v$ are {\em{not}} being moved to $\VtxTaken$.
 \item For each $w\in (N(x)\cap \VtxActive)\setminus L$, increase $\poslimit(w)$ to $\max(\poslimit(w),\degord(x)+1)$.
 \item For each $v\in L$ and $w\in (N(v)\cap \VtxActive)\setminus L_v$, increase $\poslimit(w)$ to $\max(\poslimit(w),\degord(v)+1)$.
\end{itemize}
It is easy to see that if $(\VtxTaken,\VtxActive,\degord,\poslimit)$ was lucky, then at least one of the guesses leads to considering a lucky subproblem. In particular, property~\eqref{eq:guessed-left} is satisfied in this subproblem. This completes the description of a branching step.

It remains to argue why it is still true that on every root-to-leaf path in the recursion tree there~are at most $\Oh(\log^3 n)$ success branches. Before, the key argument was that when a success branch is executed, a constant fraction of buckets (either primary or secondary) loses a constant fraction of elements. Now, the progress is explained by the following claim, which follows easily from the way we perform~branching.

\begin{claim}\label{cl:degenerate-progress}
 Suppose $(\VtxTaken,\VtxActive,\degord,\poslimit)$ is a lucky subproblem in which branching on $x$ is executed, and let $(\VtxTaken',\VtxActive',\degord',\poslimit')$ be any of the obtained child subproblems. Then for every $y\in N(x)\cap \VtxActive$, we either have
 $$y\notin \VtxActive'\qquad\textrm{or}\qquad |\{z\in \VtxTaken\cap N(y)~|~\degord(z)<\poslimit(y)\}| < |\{z\in \VtxTaken'\cap N(y)~|~\degord'(z)<\poslimit'(y)\}|.$$
\end{claim}

Note that for $y\in \VtxActive$, if $y$ gets included in the solution, then the whole set $$M_y\coloneqq \{z\in \VtxTaken\cap N(y)~|~\degord(z)<\poslimit(y)\}$$ must become the left neighbors of $y$. So if the size of $M_y$ exceeds $d$, then we can conclude that $y$ cannot be included in the solution and we can safely remove $y$ from $\VtxActive$. Thus, the increase of the cardinality of $M_y$ for all neighbors $y$ of $x$ that do not get excluded from consideration is the progress achieved by the~algorithm.

Formally, we do as follows. Recall that before, we measured the progress in emptying a bucket $\Bc^G_{u,v,w}$ by monitoring its size. Now, we monitor the {\em{potential}} of $\Bc^G_{u,v,w}$ defined as
$$\Phi(\Bc^G_{u,v,w}) \coloneqq \sum_{T\in \Bc^G(u,v,w)}\  \sum_{y\in V(T)} (d-|M_y|).$$
Thus, $\Phi(\Bc^G(u,v,w))$ measures how much the vertices of tripods of $\Bc^G_{u,v,w}$ have left till saturating their ``quotas'' for the number of left neighbors. From \cref{cl:degenerate-progress} it can be easily inferred that when branching on a heavy vertex, a constant fraction of buckets lose a constant fraction of their potential, and the same complexity analysis as before goes through.

\subsection{\cmsotwo properties}

We now extend the approach presented in the previous section to a sketch of a proof of \cref{thm:main} in full generality. That is, we also take into account \cmsotwo-expressible properties. 

\paragraph*{Degeneracy and treewidth.}
The first step is to argue that degeneracy and treewidth are functionally equivalent in $C_{>t}$-free graphs, i.e., to prove \cref{thm:deg2tw}. This part of the reasoning is presented in \cref{sec:deg2tw}.

The argument goes roughly as follows. Suppose, for contradiction, that $G$ is a $C_{>t}$-free $d$-degenerate graph that has huge treewidth (in terms of $d$ and $t$). Using known results~\cite{DBLP:journals/corr/abs-2008-02133}, in $G$ we can find a huge {\em{bramble}}~$\Bc$ --- a family of connected subgraphs that pairwise either intersect or are adjacent --- such that every vertex of $G$ is in at most two elements of $\Bc$. This property 
means that $\Bc$ gives rise to a huge clique minor in $G'$, the graph obtained from $G$ by adding a copy of every vertex \prz{and making it adjacent to the original vertex?}\mip{This is said by saying true twin.} (the copy is a true twin of the original). Note that $G'$ is still $C_{>t}$-free and is $2d+1$-degenerate. Now, we can easily prove that the obtained clique minor in $G'$ can be assumed to have {\em{depth}} at most $t$: every branch set induces a subgraph of radius at most~$t$. Using known facts about bounded-depth minors~\cite[Lemma 2.19 and Corollary 2.20]{sparsity-notes}, it follows that $G'$ contains a topological minor model of a large clique that has depth at most $3t+1$: every path representing an edge has length at most~$6t+3$. Finally, we show that if we pick at random $t+1$ roots $v_0,\ldots,v_t$ of this topological minor model, and we connect them in order into a cycle in $G'$ using the paths from the model, then with high probability\prz{actually, this high is 0.5, I guess it should be 'positive'}\mip{We could make it as close to $1$ as we like} this cycle will be induced. This is because $G'$ is $(2d+1)$-degenerate, so two paths of the model chosen uniformly at random are with high probability nonadjacent, due to their shortness. Thus, we uncovered an induced cycle on more than $t$ vertices in $G'$, a contradiction.

\paragraph*{Boundaried graphs and types.} We proceed to the proof of \cref{thm:main}. By \cref{thm:deg2tw}, the subgraph $G[S]$ induced by the solution has treewidth smaller than $k$, where $k$ is a constant that depends only on $d$ and $t$. Therefore, we will use known compositionality properties of \cmsotwo logic on graphs of bounded~treewidth.

For an integer $\ell$, an {\em{$\ell$-boundaried graph}} is a pair $(H,\iota)$, where $H$ is a graph and $\iota$ is an injective partial function from $V(H)$ to $[\ell]$, called the {\em{labelling}}. The domain of $\iota$ is the {\em{boundary}} of $(H,\iota)$ and if $\iota(u)=\alpha$, then $u$ is a boundary vertex with label $\alpha$. On $\ell$-boundaried graphs we have two natural operations: {\em{forgetting}} a label --- removing a vertex with this label from the domain of $\iota$ --- and {\em{gluing}} two boundaried graphs --- taking their disjoint union and fusing boundary vertices with the same labels. It is not hard to see that a graph has treewidth less than $\ell$ if and only if it can be constructed from two-vertex $\ell$-boundaried graphs by means of these operations.

The crucial, well-known fact about \cmsotwo is that this logic behaves in a compositional way under the operations on boundaried graphs. Precisely, for each fixed $\ell$ and \cmsotwo sentence $\phi$ there is a finite set $\msotypes$ of {\em{types}}  such that to every $\ell$-boundaried graph $(H,\iota)$ we can assign $\msotype(H,\iota)\in \msotypes$ so that:
\begin{itemize}[nosep]
 \item Whether $H\models \phi$ can be uniquely determined by examining $\msotype(H,\iota)$.
 \item The type of the result of gluing two $\ell$-boundaried graphs depends only on the types of those graphs.
 \item The type of the result of forgetting a label in an $\ell$-boundaried graph depends only on the label in question and the type of this graph.
\end{itemize}
See \cref{lem:mso-type} for a formal statement. In our proof we will use $\ell\coloneqq 6k$, that is, the boundaries will by a bit larger than the promised bound on the treewidth.

\paragraph*{Enriching branching with types.} We now sketch how to enrich the algorithm from the previous section to the final branching procedure; this part of the reasoning is presented in \cref{sec:longhole}. The idea is that we perform branching as in the previous section (with significant augmentations, as will be described in a moment), but in order to make sure that the constructed induced subgraph $G[S]$ satisfies $\phi$, we enrich each subproblem with the following information:
\begin{itemize}[nosep]
 \item A rooted tree decomposition $(T,\beta)$ of $G[\VtxTaken]$ of width at most $\ell$ ($\beta\colon V(T)\to \VtxTaken$ is the bag function).
 \item For each node $a$ of $T$, a {\em{projected type}} $\msotype_a\in \msotypes$.
\end{itemize}
Again, we fix some optimum solution $S^\star$ together with a $d$-degeneracy ordering $\degord^\star$ of $G[S^\star]$.
Compared to the approach of the previous section, we extend the definition of a subproblem being lucky as follows:
\begin{itemize}[nosep]
 \item For each connected component $D$ of $G[\VtxActive\cap S^\star]$, we require that $N(D)\cap \VtxTaken$ is a set of size at most~$4k$ such that there exists a bag of $(T,\beta)$ that entirely contains it. For such a component $D$, let $a(D)$ be the topmost node of $T$ satisfying $N(D)\cap \VtxTaken\subseteq \beta(a(D))$.
 \item For each node $a$ of $T$, consider the graph $H_a$ induced by $\beta(a)$ and the union of all those components $D$ of $G[\VtxActive\cap S^\star]$ for which $a(D)=a$. Then the type of $H_a$ with $\beta(a)$ as the boundary is equal~to~$\msotype_a$.
\end{itemize}
Thus, one can imagine the solution $S^\star$ as $\VtxTaken$ plus several extensions into $\VtxActive$ --- vertex sets of components of $G[\VtxActive\cap S^\star]$. Each of such extensions $D$ is hanged under a single node $a(D)$ of $(T,\beta)$ and is attached to it through a neighborhood of size at most $4k$. For each node $a$ of $T$, we store the projected combined type $\msotype_a$ of all the extensions $D$ for which $a$ is the topmost node to which $D$ attaches. Note that as since the algorithm will always make only $\log^{\Oh(1)} n$ success branches along each root-to-leaf path, we maintain the invariant that $|\VtxTaken|\leq \log^{\Oh(1)} n$, which implies the same bound on the number of nodes of $T$. 

Recall that in the algorithm presented in the previous section, two basic operations were performed: (a) recursing on connected components of $G[\VtxActive]$ once this graph becomes disconnected; and (b) branching on a node~$x\in \VtxActive$. 

Lifting (a) to the new setting is conceptually simple. Namely, each graph $H_a$ is correspondingly split among the components of $G[\VtxActive]$, so we guess the projected types of those parts of $H_a$ so that they compose to $\msotype_a$. The caveat is that in order to make the time complexity analysis sound, we can perform such guessing only when a significant progress is achieved by the algorithm. This is done by performing (a) {\em{only}} when each connected component of $G[\VtxActive]$ contains at most $99\%$ of all the vertices of $\VtxActive$, which means that the number of active vertices after this step will drop by $1\%$ in each branch. This requires technical care.

More substantial changes have to be applied to lift operation (b), branching on a node $x\in \VtxActive$. Failure branch works the same way as before: we just remove $x$ from $\VtxActive$. As for success branches, recall that in each of them together with $x$ we move to $\VtxTaken$ the whole set $L$ of left neighbors of $x$ in $\VtxActive$. Clearly, the vertices of $L\cup \{x\}$ belong to the same component of $G[\VtxTaken\cap S^\star]$, say $D$. It would be natural to reflect the move of $L\cup \{x\}$ to $\VtxTaken$ in the decomposition $(T,\beta)$ as follows: create a new node $b$ with $\beta(b)=(L\cup \{x\})\cup (N(D)\cap \VtxTaken)$ and make it a child of $a(D)$ in $T$. Note that thus, $|\beta(b)|\leq d+1+4k\leq 5k$, so the bound on the width of $(T,\beta)$ is maintained (even with a margin). However, there is a problem: if by $\VtxTaken'$ we denote the updated~$\VtxTaken$, i.e., $\VtxTaken'=\VtxTaken\cup (L\cup \{x\})$, then the removal of $L\cup \{x\}$ breaks $D$ into several connected components whose neighborhoods in $\VtxTaken'$ are contained in $(N(D)\cap \VtxTaken)\cup (L\cup \{x\})$. This set, however, may have size as large as $4k+d+1$, so we do not maintain the invariant that every connected component of $G[\VtxActive\cap S^\star]$ has at most $4k$ neighbors in $\VtxTaken$.

We remedy this issue using a trick that is standard in the analysis of bounded-treewidth graphs. Since the graph $G[D\cup (N(D)\cap \VtxActive)]$ has treewidth smaller than $k$, there exists a set $K$ consisting of at most $k$ vertices of $D\cup (N(D)\cap \VtxActive)$ such that every connected component of $D-K$ contains at most $|N(D)\cap \VtxActive|/2$ vertices of $N(D)\cap \VtxActive$ (see \cref{lem:tw-balanced-sep}). The algorithm guesses $K$ along with $L$, moves $K$ to $\VtxActive$ along with $x$ and $L$, and and sets $\beta(b)\coloneqq (N(D)\cap \VtxTaken)\cup (L\cup \{x\})\cup K$; thus $|\beta(b)|\leq 4k+(d+1)+k\leq 6k$. Now it is easy to see that due to the inclusion of $K$, every connected component of $D-(K\cup L\cup \{x\})$ has only at most $2k+k+(d+1)\leq 4k$ neighbors in $\beta(b)$, and the problematic invariant is maintained.

This concludes the overview of the proof of \cref{thm:main}.

%% file: prelims.tex
We use standard graph notation. For an positive integer $p$, we write $[p]\coloneqq \{1,\ldots,p\}$. For a set $A$, $\binom{A}{p}$ denotes the set of all $p$-element subsets of $A$.

We say that two vertex subsets $X_1,X_2 \subseteq V(G)$ are \emph{adjacent} if either $X_1 \cap X_2 \neq \emptyset$ or there is an edge $x_1x_2 \in E(G)$, such that $x_1 \in X_1$ and $x_2 \in X_2$.
Otherwise, the sets are \emph{nonadjacent}.

For a path $P$, the \emph{length} of $P$ is the number of edges of $P$.
For a graph $G$, the {\em{radius}} of $G$ is $\min_{v \in V(G)} \max_{u \in V(G)} \dist(u,v)$, where $\dist(u,v)$ denotes the \emph{distance} between $u$ and $v$, i.e., the length of a shortest $u$-$v$-path in $G$. 

\subsection{Graph minors}

Let $H$ be a graph.\mip{I rewrote this part to have more concise definitions. I think it can stay here.}
A {\em{minor model}} of $H$ in a graph $G$ is a mapping $\eta$ that assigns to each $v\in V(H)$ a connected subgraph $\eta(v)$ of $G$ so that:
\begin{itemize}
 \item the subgraphs $\{\eta(v)\colon v\in V(H)\}$ are pairwise disjoint; and
 \item for every edge $v_1v_2\in E(H)$, there is an edge in $G$ with one endpoint in $\eta(v_1)$ and the other in $\eta(v_2)$.
\end{itemize}
Such a minor model $\eta$ has {\em{depth}} $t$ if every subgraph $\eta(v)$ has radius at most $t$. We say that $G$ contains $H$ as a (depth-$t$) minor if there is a (depth-$t$) minor model of $H$ in $G$.

%


A {\em{topological minor model}} of $H$ in $G$ is a mapping $\psi$ that assigns to each vertex $v\in V(H)$ a vertex $\psi(v)$ in $G$ and to each edge $e\in E(H)$ a path $\psi(e)$ in $G$ so that
\begin{itemize}
 \item vertices $\psi(v)$ are pairwise different for different $v\in V(H)$;
 \item for each edge $v_1v_2\in E(H)$, the path $\psi(v_1v_2)$ has endpoints $\psi(v_1)$ and $\psi(v_2)$ and does not pass through any of the vertices of $\{\psi(v)\colon v\in V(H)\}$ other than $\psi(v_1)$ and $\psi(v_2)$; and
 \item paths $\{\psi(e)\colon e\in E(H)\}$ are pairwise disjoint apart from possibly sharing endpoints.
\end{itemize}
 The vertices $\{\psi(v)\colon v\in V(H)\}$ are the {\em{roots}} of the topological minor model $\psi$. We say that $\psi$ has {\em{depth}} $t$ if each path $\psi(e)$ for $e\in E(H)$ has length at most $2t+1$. We say that $G$ contains $H$ as a (depth-$t$) topological minor if there is a depth-$t$ topological minor model of $H$ in $G$. It is easy to see that if $G$ contains $H$ as a depth-$t$ topological minor, then it also contains $H$ as a depth-$t$ minor.
%
%

\subsection{Treewidth and tree decompositions}

A \emph{tree decomposition} of a graph $G$ is a pair $(T, \beta)$, where $T$ is a tree and $\beta$ is a function that maps every vertex of $T$ to a subset of $V(G)$, such that the following properties hold:
\begin{itemize}
\item every edge of $G$ is contained in $\beta(a)$ for some $a \in V(T)$, and
\item for every $v \in V(G)$, the set $\{a \in V(T)~|~v \in \beta(a)\}$ is nonempty and induces a subtree of $T$.
\end{itemize} 
The sets $\beta(a)$ for $a \in V(T)$ are called the \emph{bags} of the decomposition $(T,\beta)$.
The \emph{width} of the decomposition $(T,\beta)$ is $\max_{a \in V(T)} |\beta(a)| - 1$ and the \emph{treewidth} of a graph is the minimum
possible width of its decomposition.

We will need the following well-known observation about graphs of bounded treewidth.
\begin{lemma}[see Lemma 7.19 of~\cite{platypus}]\label{lem:tw-balanced-sep}
Let $H$ be a graph of treewidth less than $k$ and let $A \subseteq V(H)$. 
Then there exists a set $X \subseteq V(H)$ of size at most $k$
such that every connected component of $H\setminus X$ has at most $|A|/2$ vertices of $A$.
\end{lemma}

\subsection{\msotwo and \msotwo types}
We assume the incidence encoding of graphs as relational structures: a graph $G$ is encoded as a relational structure whose universe consists of vertices and edges of $G$ (each distinguishable by a unary predicate), and there is one binary incidence relation binding each edge with its two endpoints.
With this representation, the \msotwo logic on graphs is the standard \mso logic on relational structures as above, which boils down to allowing the formulas to use vertex variables, edge variables,
vertex set variables, and edge set variables, together with the ability of quantifying over them. For a positive integer $p$, the \cpmsotwo logic extends \msotwo by allowing atomic formulas of the form $|X|\equiv a\bmod m$, where $X$ is a set variable and $0\leq a<m\leq p$ are  integers; denote $\mathsf{CMSO}_2\coloneqq \bigcup_{p>0} \mathsf{C}_{\leq p}\mathsf{MSO}_2$. The {\em{quantifier rank}} of a formula is the maximum number of nested quantifiers in it.

For a finite set $\Lambda$ of labels, a \emph{$\Lambda$-boundaried graph}
is a pair $(G,\iota)$ where $G$ is a graph and
$\iota \colon B \to \Lambda$ is an injective function for some $B \subseteq V(G)$; then $B=\dom \iota$ is called the {\em{boundary}} of $(G,\iota)$.
A \emph{$k$-boundaried graph} is a shorthand for a $[k]$-boundaried graph, where we denote $[k]=\{1,\ldots,k\}$.
For $v \in B$, $\iota(v)$ is the \emph{label} of $v$.
We define two operations on $\Lambda$-boundaried graphs.

If $(G_1,\iota_1)$ and $(G_2,\iota_2)$ are two $\Lambda$-boundaried graphs,
   then the result of \emph{gluing} them is the boundaried graph
   $(G_1,\iota_1) \oplus_\Lambda (G_2,\iota_2)$ that is obtained
   from the disjoint union of $(G_1, \iota_1) $ and $(G_2, \iota_2)$ by identifying vertices of the same label (so that the resulting labelling is again injective). 

If $(G,\iota)$ is a $\Lambda$-boundaried graph and $l \in \Lambda$, then
the result of \emph{forgetting $l$} is the $\Lambda$-boundaried graph
$(G,\iota|_{\iota^{-1}(\Lambda \setminus \{l\})})$. That is, if $l$ belongs
to the range of $\iota$, we remove the label from the corresponding vertex
(but we keep this vertex in $G$). 

By \cpmsotwo on $\Lambda$-boundaried graphs we mean the \cpmsotwo logic over graphs enriched with $|\Lambda|$ unary predicates: for each label $l\in \Lambda$ we have a unary predicate that selects the only vertex with label $l$, if existent. Effectively, this boils down to the possibility of using elements of the boundary as constants in \cpmsotwo formulas.

The following folklore proposition explains the compositionality properties of \cmsotwo on boundaried graphs. The statement and the proof is standard, see e.g.~\cite[Lemma~6.1]{GroheK09}, so we only sketch it.

\begin{proposition}\label{lem:mso-type}
For every triple of integers $k,p,q$, there exists a finite set $\msotypes^{k,p,q}$ and a function that assigns to every $k$-boundaried graph $(G,\iota)$ a \emph{type} $\msotype^{k,p,q}(G,\iota) \in \msotypes^{k,p,q}$ such that the following holds:
\begin{enumerate}
\item The types of isomorphic graphs are the same.
\item 
  For every \cpmsotwo sentence $\phi$ on $k$-boundaried graphs, whether $(G,\iota)$ satisfies $\phi$ depends only on the type $\msotype^{k,p,q}(G,\iota)$, where $q$ is the quantifier rank of $\phi$.
  More precisely,
  there exists a subset $\msotypes^{k,p,q}[\phi] \subseteq \msotypes^{k,p,q}$
  such that for every $k$-boundaried graph $(G,\iota)$ we have
  $$(G,\iota) \models \phi\qquad\textrm{if and only if}\qquad \msotype^{k,p,q}(G,\iota) \in \msotypes^{k,p,q}[\phi].$$
\item 
 The types of ingredients determine the type of the result of the gluing operation.
More precisely, for every two types $\tau_1,\tau_2 \in \msotypes^{k,p,q}$
there exists a type $\tau_1 \oplus_{k,p,q} \tau_2$ such that
for every two $k$-boundaried graphs $(G_1,\iota_1)$, $(G_2,\iota_2)$,
if $\msotype^{k,p,q}(G_i,\iota_i) = \tau_i$ for $i=1,2$, then
$$\msotype^{k,p,q}((G_1,\iota_1) \oplus_{[k]} (G_2,\iota_2)) = \tau_1 \oplus_{k,p,q} \tau_2.$$ 
Also, the operation $\oplus_{k,p,q}$ is associative and commutative.
\item 
 The type of the ingredient determines the type of the result of the forget label operation.
 More precisely, 
 for every type $\tau \in \msotypes^{k,p,q}$ and $l \in [k]$
there exists a type $\tau_{\neg l}$ such that
for every $k$-boundaried graph $(G,\iota)$,
if $\msotype^{k,p,q}(G,\iota) = \tau$
and $(G',\iota')$ is the result of forgetting $l$ in $(G,\iota)$, then
$$\msotype^{k,p,q}(G',\iota') = \tau_{\neg l}.$$ 
\end{enumerate}
\end{proposition}
\begin{proof}[sketch]
 It is well-known that there are only finitely many syntactically non-equivalent \cpmsotwo sentence over $k$-boundaried graphs and of quantifier rank at most $q$, so let $\Sentences^{k,p,q}$ be a set containing one such sentence from each equivalence class. We set $\msotypes^{k,p,q}$ as the power set (set of all subsets) of $\Sentences^{k,p,q}$. To each $k$-boundaried graph $(G,\iota)$ we define $\msotype^{k,p,q}(G,\iota)\subseteq \Sentences^{k,p,q}$ as the set of all sentences $\psi\in \Sentences^{k,p,q}$ satisfied in $(G,\iota)$. Thus, whether $(G,\iota)$ satisfies $\phi$ can be decided by verifying whether $\msotype^{k,p,q}(G,\iota)$ contains a sentence that is syntactically equivalent to $\phi$. The remaining two assertions --- about compositionality of the gluing and the forget label operations --- follow from a standard argument using Ehrenfeucht-Fra\"isse games.
\end{proof}

In our algorithms, we shall work with a fixed \cmsotwo sentence $\phi$ over graphs with treewidth upper-bounded by a fixed constant $k$. Note that $\phi$ belongs to \cpmsotwo for a fixed constant $p$, and the quantifier rank of $\phi$ is a fixed constant $q$. Hence, when working in $k$-boundaried graphs, whether $\phi$ is satisfied in a $k$-boundaried graph $(G,\iota)$ can be read from its type $\msotype^{k,p,q}(G,\iota)$. To facilitate the computation of types, we shall assume that our algorithms have a hard-coded set of types $\msotypes^{k,p,q}$, together with the subset $\msotypes^{k,p,q}[\phi]\subseteq \msotypes^{k,p,q}$ and functions
$$(\tau_1,\tau_2) \mapsto \tau_1 \oplus_{k,p,q} \tau_2\qquad\textrm{and}\qquad (\tau,l) \mapsto \tau_{\neg l},$$
as described in \cref{lem:mso-type}. Also, the algorithms have hard-coded the types of all $k$-boundaried graphs with $\Oh(k)$ vertices.

We will need also a simple observation that \cmsotwo types preserve connectivity properties,
as being in the same connected component can be easily expressed in \cmsotwo.
\begin{lemma}\label{lem:msotwo-connected}
Fix integers $k \geq 0$, $p \geq 0$, and
$q \geq 4$, and suppose that
$(G_1,\iota_1)$ and $(G_2,\iota_2)$ are two $k$-boundaried graphs
with $\msotype^{k,p,q}(G_1,\iota_1) = \msotype^{k,p,q}(G_2,\iota_2)$.
Then the ranges of $\iota_1$ and $\iota_2$ are equal. Furthermore, for every
two pairs $(u_1,u_2), (v_1,v_2) \in V(G_1) \times V(G_2)$ with
$\iota_1(u_1) = \iota_2(u_2)$ and $\iota_1(v_1) = \iota_2(v_2)$,
   vertices $u_1$ and $v_1$ are in the same connected component of $G_1$
  if and only if 
   vertices $u_2$ and $v_2$ are in the same connected component of $G_2$.
\end{lemma}
\begin{proof}
That the ranges of $\iota_1$ and $\iota_2$ are equal is clear: $G_1$ and $G_2$ satisfy the same sentences of the form ``there exists a vertex with label $i\in [k]$''.
The assertion about having the same connectivity between boundary vertices follows from an analogous argument, supplied with the observation that being in the same connected component can be expressed by an
\msotwo formula of quantifier rank four:
\begin{eqnarray*}
\mathsf{Connected}(x,y) & = & \neg \exists_{A \subseteq V(G)} (x \in A) \wedge (y \notin A) \wedge \\
& & \left(\forall_{e \in E(G)}
\forall_{u\in V(G)}\forall_{v\in V(G)} (\inc(u,e)\wedge \inc(v,e))\Rightarrow ((u \in A) \Leftrightarrow (v \in A))\right).
\end{eqnarray*}
\end{proof}

\medskip

In our proofs, we will need to keep track of the exact value of the treewidth of a constructed induced subgraph of the given graph.
For this, we use the following observation.

\begin{lemma}\label{lem:mso-tw-formula}
For every integer $k$ there exists an \msotwo sentence $\phi_{\tw < k}$
  such that $G \models \phi_{\tw < k}$ if and only if the treewidth of $G$ is less than $k$.
\end{lemma}
\begin{proof}
Let $\mathcal{F}_{\tw < k}$ be the set of all minor-minimal graphs of treewidth at least $k$. By the Robertson-Seymour Theorem~\cite{RobertsonS04}, $\mathcal{F}_{\tw < k}$ is finite.
Define $\phi_{\tw < k}$ to be the conjunction, over all $H \in \mathcal{F}_{\tw < k}$,
of sentences asserting that $H$ is not a minor of $G$. Note here that it is straightforward to express in \msotwo
    that a fixed graph $H$ is a minor of a given graph $G$.
\end{proof}

%% file: branching.tex
In this section we prove \cref{thm:deg-branching} stated below, which is a weaker variant of \cref{thm:main} that does not speak about \cmsotwo properties.

\begin{theorem}
\label{thm:deg-branching}
Fix a pair of integers $d$ and $t$. Then there exists an algorithm that, given a $C_{>t}$-free $n$-vertex graph $G$
and a weight function $\wei\colon V(G) \to \mathbb{N}$, in time $n^{\Oh(\log^3 n)}$ finds a subset of vertices~$S$ such that $G[S]$ is $d$-degenerate and, subject to the above, $\wei(S)$ is maximum~possible. The running time can be improved to $n^{\Oh(\log^2 n)}$ if $G$ is $P_t$-free.
\end{theorem}

We present the strategy for our branching algorithm in a quite general fashion,
so that it can be reused later for the proof of \cref{thm:main}. For the description, we fix a positive integer $d$ that is the bound on the degeneracy of the sought induced subgraph.

We will rely on the following characterization of graphs of bounded degeneracy.
As the statement of \cref{lem:degeneracy} is a bit non-standard, we include a sketch of a proof.
\begin{lemma}\label{lem:degeneracy}
A graph $G$ has degeneracy at most $d$ if and only if
there exists a function $\degord \colon V(G) \to [|V(G)|]$
such that for every $uv \in E(G)$ we have $\degord(u) \neq \degord(v)$,
and for every $v \in V(G)$,
the set $\{ u \in N_G(v)~|~\degord(u) < \degord(v)\}$
has size at most $d$.
\end{lemma}
\begin{proof}[sketch]
If $G$ is $d$-degenerate, then we can construct an ordering $\degord$ by removing a vertex $v$ with minimum degree, 
inductively ordering the vertices of $G - v$, and appending $v$ at the last position.

On the other hand, consider any subgraph $G'$ of $G$ and let $\degord'$ be the ordering $\degord$ restricted to the vertices of $G'$.
Let $v \in V(G')$ be a vertex at the last position in $\degord'$; if there is more than one such vertex, we choose one arbitrarily.
Note that all neighbors of $v$ in $G'$ precede it in $\degord'$, so there are at most $d$ of them.
\end{proof}
Function $\degord$ as in \cref{lem:degeneracy} shall be called a {\em{$d$-degeneracy ordering}} of $G$, and the value $\degord(v)$ is the {\em{position}} of~$v$.
We remark here that, contrary to the usual definition, we do not require $\degord$
to be injective, but only to give different positions to endpoints of a single edge.
Henceforth, we say that a function $\degord \colon Z \to [|V(G)|]$ for some $Z \subseteq V(G)$
is \emph{edge-injective} if $\degord(u) \neq \degord(v)$ for every $uv \in E(G[Z])$. 

Let $G$ be an $n$-vertex $C_{>t}$-free graph and $\wei\colon V(G) \to \mathbb{N}$ be a weight function.
Without loss of generality, assume $t \geq 6$.
Fix a subset $S^\ast \subseteq V(G)$ such that $G[S^\ast]$ has degeneracy at most $d$
and fix a $d$-degeneracy ordering $\degord^\ast \colon S^\ast \to [|S^\ast|]$ of $G[S^\ast]$. We think of $S^\ast$
as of the intended optimum solution.

\subsection{Recursion structure}
A \emph{subproblem} $\rcall$ consists of 
\begin{itemize}
\item Two disjoint vertex sets $\VtxTaken^\rcall, \VtxDeleted^\rcall \subseteq V(G)$.
 We additionally denote
 $\VtxFree^\rcall \coloneqq V(G) \setminus (\VtxTaken^\rcall \cup \VtxDeleted^\rcall)$
and call $\VtxFree^\rcall$ the \emph{free vertices}.
\item An integer $\rlevel^\rcall \geq 0$, called the \emph{level} of the subproblem.
\item A set $\VtxActive^\rcall \subseteq \VtxFree^\rcall$ that is nonadjacent
to $\VtxFree^\rcall \setminus \VtxActive^\rcall$ 
and is of size less than $0.99^{-\rlevel^\rcall}$. The vertices of $\VtxActive^\rcall$
are called the \emph{active vertices};
\item An edge-injective function $\degord^\rcall\colon \VtxTaken^\rcall \to [n]$.
\item A function $\poslimit^\rcall\colon \VtxActive^\rcall \to [n+1]$. 
\end{itemize}
The superscript can be omitted if it is clear from the context.

In our recursive branching algorithms, one recursive call will treat one subproblem. The intended meaning of the components of the subproblem is as follows:
\begin{itemize}
\item $\VtxTaken^\rcall$ corresponds to the vertices already decided to be in the partial solution.
The value $\degord^\rcall(v)$ for $v \in \VtxTaken^\rcall$ indicates the final position
 in the degeneracy ordering.
\item $\VtxDeleted^\rcall$ corresponds to the vertices already decided to be not in the partial solution.
\item $\VtxFree^\rcall$ are vertices yet to be decided.
\end{itemize}
Relying on some limited dependence between the connected components of $G[\VtxFree^\rcall]$
in the studied problems, one recursive call focuses only on a group
of these components, whose union of vertex sets is denoted by $\VtxActive^\rcall$, and seeks for a best
way to extend the current partial solution into $\VtxActive^\rcall$. 
The integer $\poslimit^\rcall(v)$ for $v \in \VtxActive^\rcall$ indicates the minimum
position at which the vertex $v$ can be placed in the final degeneracy ordering.

This intuition motivates the following definition. A subproblem $\rcall$
is \emph{lucky} if 
\begin{itemize}
\item $\VtxTaken^\rcall \subseteq S^\ast$, $\VtxDeleted^\rcall \cap S^\ast = \emptyset$;
\item $\degord^\rcall = \degord^\ast|_{\VtxTaken^\rcall}$;
\item for every $u \in S^\ast \cap \VtxActive^\rcall$ we have
   $\poslimit^\rcall(u) \leq \degord^\ast(u)$; and
\item for every $u \in \VtxActive^\rcall$
   and every $v \in N_G(u) \cap \VtxTaken^\rcall$,
   if $\poslimit^\rcall(u) \leq \degord^\ast(v)$, then $u\in S^\ast$ and $\degord^\ast(u) < \degord^\ast(v)$.
\end{itemize}

We will also need the following notion: for an edge-injective function
$\degord\colon Z \to [n]$ for some $Z \subseteq V(G)$, $v \in V(G)$, and an integer $p$,
the \emph{quota} of $v$ in $\degord$ at position $p$ is defined
as 
\[ \quota(\degord, v, p) = d - |\{u \in N_G(v) \cap Z~|~\degord(u) < p\}|. \]
If $v \in Z$, then $\quota(\degord, v)$ is a shorthand for $\quota(\degord, v, \degord(v))$.
Intuitively, $\quota(\degord, v)$ measures the number of ``free slots'' for the neighbors of $v$ that precede it in $\degord$.
Note that an edge-injective function $\degord:V(G) \to [n]$ is a $d$-degeneracy ordering of $G$
if and only if $\quota(\degord,v) \geq 0$ for every $v \in V(G)$. 
Furthermore, if $\degord' : Z' \to [n]$ is an extension
of $\degord\colon Z \to [n]$ to some $Z' \supseteq Z$, 
   then for every $v \in V(G)$ and $p \in [n]$, it holds that 
   $\quota(\degord, v, p) \geq \quota(\degord', v, p)$. 

\paragraph{Extending.}
For subproblems $\rcall_1$ and $\rcall_2$, we say that
\begin{itemize}
\item $\rcall_2$ \emph{extends} $\rcall_1$
if \begin{itemize}
\item $\VtxTaken^{\rcall_1} \subseteq \VtxTaken^{\rcall_2}$,
$\VtxDeleted^{\rcall_1} \subseteq \VtxDeleted^{\rcall_2}$
  (and hence
 $\VtxFree^{\rcall_1} \supseteq \VtxFree^{\rcall_2}$); 
\item $\VtxTaken^{\rcall_2} \setminus \VtxTaken^{\rcall_1} \subseteq \VtxActive^{\rcall_1}$ and
$\VtxDeleted^{\rcall_2} \setminus \VtxDeleted^{\rcall_1} \subseteq \VtxActive^{\rcall_1}$;
\item 
 $\VtxActive^{\rcall_1} \supseteq \VtxActive^{\rcall_2}$;
 \item 
 $\degord^{\rcall_1} = \degord^{\rcall_2}|_{\VtxTaken^{\rcall_1}}$;
\item  for every $v \in \VtxActive^{\rcall_2}$ we have
 $\poslimit^{\rcall_2}(v) \geq \poslimit^{\rcall_1}(v)$; and
\item  for every $v \in \VtxTaken^{\rcall_2} \cap \VtxActive^{\rcall_1}$ we have
 $\degord^{\rcall_2}(v) \geq \poslimit^{\rcall_1}(v)$.
 \end{itemize}
\item $\rcall_2$ \emph{extends} $\rcall_1$ \emph{completely}
if additionally $(\VtxTaken^{\rcall_2} \setminus \VtxTaken^{\rcall_1}) \cup (\VtxDeleted^{\rcall_2} \setminus \VtxDeleted^{\rcall_1}) = \VtxActive^{\rcall_1}$ and $\rlevel^{\rcall_2} = 0$
(in particular, $\VtxActive^{\rcall_2} = \emptyset$).
\end{itemize}
One easy way to extend a subproblem $\rcall$ is to select a set $Z \subseteq \VtxActive^\rcall$
and move it to $\VtxDeleted$; formally, the operation of \emph{deleting $Z$}
creates a new subproblem $\rcall'$ extending $\rcall$ by keeping all the ingredients
the same, except for $\VtxDeleted^{\rcall'} = \VtxDeleted^\rcall \cup Z$,
$\VtxActive^{\rcall'} = \VtxActive^\rcall \setminus Z$, and
$\poslimit^{\rcall'} = \poslimit^{\rcall}|_{\VtxActive^{\rcall'}}$.
Clearly, if $\rcall$ is lucky and $Z \cap S^\ast = \emptyset$, then $\rcall'$ is lucky as well.

A second (a bit more complicated way) to extend a subproblem $\rcall$
is the following. First, select a set $Z \subseteq \VtxActive^\rcall$.
Second, select an edge-injective function $\degord\colon \VtxTaken^\rcall \cup Z \to [n]$
extending $\degord^\rcall$
such that $\degord(u) \geq \poslimit^\rcall(u)$ for every $u \in Z$;
we henceforth call such a function a \emph{position guess} for $Z$ and $\rcall$.
Third, for every $u \in Z$ select a set $D_u \subseteq \VtxActive^\rcall \setminus Z$
of size at most $\quota(\degord, u)$. The family $(D_u)_{u \in Z}$
is called the \emph{left neighbor guess} for $Z$, $\rcall$, and $\degord$. 
Finally, define the operation of \emph{taking $Z$ at positions $\degord$ with left neighbors
$(D_u)_{u \in Z}$}
as constructing a new subproblem $\rcall'$ created from $\rcall$ by keeping all the ingredients
the same, except for
$\VtxTaken^{\rcall'} = \VtxTaken^\rcall \cup Z$,
$\VtxActive^{\rcall'} = \VtxActive^\rcall \setminus Z$, 
$\degord^{\rcall'} = \degord$, and, for every $w \in \VtxActive^{\rcall'}$, taking
$$\poslimit^{\rcall'}(w) = 
 \max(\poslimit^\rcall(w), \max \{1 + \degord(u)\colon u \in Z \wedge w \in N_G(u) \setminus D_u\}).$$
We have the following immediate observation:
\begin{lemma}\label{lem:lucky-take}
Let $\rcall$ be lucky and $Z \subseteq S^\ast \cap \VtxActive^\rcall$.
Then $\degord|_{\VtxTaken^\rcall\cup Z} \coloneqq \degord^\ast|_{\VtxTaken^\rcall\cup Z}$ is a position guess for $\rcall$ and $Z$
and $(D_u)_{u \in Z}$ defined as 
$$D_u = \{w \in N_G(u) \cap S^\ast \cap \VtxActive^\rcall~|~\degord^\ast(w) < \degord^\ast(u)\}$$
is a left neighbor guess for $\rcall$, $Z$, and $\degord$. 
Furthermore, the result of taking $Z$ at positions $\degord$ and left neighbors $(D_u)_{u \in Z}$
is lucky as well.
\end{lemma}

\paragraph{Filtering.}
We will need a simple filtering step. Consider a subproblem $\rcall$.
\begin{itemize}
\item A vertex $v \in \VtxTaken^\rcall$ is \emph{offending} if there are more than $d$ vertices $u \in N_G(v)$
such that either $u \in \VtxTaken^\rcall$ and $\degord^\rcall(u) < \degord^\rcall(v)$,
or $u \in \VtxActive^\rcall$ and $\poslimit^\rcall(u) \leq \degord^\rcall(v)$. 
\item A vertex $v \in \VtxActive^\rcall$ is \emph{offending} if $\poslimit^\rcall(v) > n$ or
$\quota(\degord^\rcall, v, \poslimit^\rcall(v)) < 0$, where $\quota$ is defined w.r.t. $Z=\VtxTaken^\rcall$.
\item The subproblem $\rcall$ is \emph{clean} if there are no offending vertices.
\end{itemize}
We observe the following.

\begin{lemma}\label{lem:filtering}
 If in a subproblem $\rcall$ there is an offending vertex $v\in \VtxTaken^\rcall$
or $v \in \VtxActive^\rcall \cap S^\ast$, then $\rcall$ is not lucky.
\end{lemma}
\begin{proof}
 For contradiction, suppose $\rcall$ is lucky. Assume first there is an offending vertex $v \in \VtxActive^\rcall \cap S^\ast$. Since $\rcall$ is lucky, we have $\poslimit^\rcall(v)\leq \degord^\ast(v)\leq n$. Furthermore, 
 \begin{eqnarray*}
d & \geq & |\{u\in N_G(v)\cap S^\ast~|~\degord^\ast(u)<\degord^\ast(v)\}|\\ & \geq & |\{u\in N_G(v)\cap \VtxTaken^\rcall~|~\degord^\ast(u)<\degord^\ast(v)\}|\\
& \geq & |\{u\in N_G(v)\cap \VtxTaken^\rcall~|~\degord^\ast(u)<\poslimit^\rcall(v)\}|\\
& = &d-\gamma(\degord^\rcall,v,\poslimit^\rcall(v)),  
 \end{eqnarray*}
 which implies that $\gamma(\degord^\rcall,v,\poslimit^\rcall(v))\geq 0$. This contradicts the assumption that $v$ is offending.
 
 Assume now that there is an offending vertex $v\in \VtxTaken^\rcall$. Note that if $u\in N_G(v)\cap \VtxTaken^\rcall$ satisfies $\degord^\rcall(u)<\degord^\rcall(v)$, then $u\in S^\ast$ and $\degord^\ast(u)<\degord^\ast(v)$.
 Further, if $u\in N_G(v)\cap \VtxActive^\rcall$ satisfies $\poslimit^\rcall(u)\leq \degord^\rcall(v)$, then by the last item of the definition of being lucky we infer that $u\in S^\ast$ and $\degord^\ast(u)<\degord^\ast(v)$. Consequently, if $v$ is offending, then $|\{u\in N_G(v)\cap S^\ast~|~\degord^\ast(u)<\degord^\ast(v)\}|>d$, a contradiction.
\end{proof}

A \emph{filtering step} applied to a subproblem $\rcall$ creates a subproblem
$\rcall'$ that is the result of deleting all offending vertices $v \in \VtxActive^\rcall$
from $\rcall$.
By \cref{lem:filtering} we infer that if $\rcall$ is lucky, then $\rcall'$ is lucky, too. 

\subsection{Subproblem tree}
Note that subproblems of level $0$ have necessarily $\VtxActive^\rcall = \emptyset$. These subproblems shall correspond to the leaves of the recursion. 
Let us now proceed to the description of the recursion tree. 

A \emph{subproblem tree} is a rooted tree where every node $x$ 
is labelled with a subproblem $\rcall(x)$
and is one of the following five types: leaf node, filter node, split node, branch node, and free node.
We require that the root of the tree is labelled with a clean subproblem and that for every $x$ and its child $y$, $\rcall(y)$ extends $\rcall(x)$.

For brevity, we say that $x$ is clean if $\rcall(x)$ is clean.
Similarly, we say that $y$ (completely) extends $x$ if
$\rcall(y)$ (completely) extends $\rcall(x)$.
Also, the level of $x$ is the level of $\rcall(x)$. 

\paragraph{Leaf node.}
A leaf node $x$ has no children, is of level $\rlevel^{\rcall(x)} = 0$, and
  hence has $\VtxActive^{\rcall(x)} = \emptyset$.

\paragraph{Filter node.}
A filter node $x$ has no or one child. 
If $\rcall(x)$ contains an offending vertex $v \in \VtxTaken^{\rcall(x)}$, 
   then $x$ has no children.
Otherwise, $x$ contains a single child $y$ with $\rcall(y)$ being the result
of the filtering step applied to $\rcall(x)$. Note that the child of a filter node
is always clean.

We additionally require that a parent of a filter node is not a filter node.

\paragraph{Split node.}
We say that a subproblem $\rcall$ is \emph{splittable}
if $\rlevel^\rcall \geq 1$ and every connected component of
$G[\VtxActive^\rcall]$ has size less than $0.99^{-\rlevel^\rcall+1}$.
We observe that it is straightforward to split a splittable subproblem into 
a constant number of subproblems of smaller level.
\begin{lemma}\label{lem:splittable}
If $\rcall$ is a splittable subproblem, then there exists a family $\mathcal{F}$ of one or two
subproblems of level $\rlevel^\rcall-1$ that all extend $\rcall$, so that for every $\qcall \in \mathcal{F}$ we have $\VtxTaken^{\qcall} = \VtxTaken^\rcall$, $\VtxDeleted^{\qcall} = \VtxDeleted^{\rcall}$, and $\degord^{\qcall} = \degord^{\rcall}$,
and furthermore $\{\VtxActive^{\qcall}\colon \qcall \in \mathcal{F}\}$ is a partition of $\VtxActive^{\rcall}$
with $\poslimit^{\qcall} = \poslimit^\rcall|_{\VtxActive^{\qcall}}$ for every $\qcall \in \mathcal{F}$. 
\end{lemma}
\begin{proof}
Let $C_1,\ldots,C_p$ be connected components of $G[\VtxActive^{\rcall}]$ in a non-increasing
order of their number of vertices.
Since $\rcall$ is of level $\rlevel^\rcall$, it holds that $\sum_{i=1}^p |C_i| < 0.99^{-\rlevel^\rcall}$.
Let $j$ be the maximum index for which it holds that $\sum_{i=1}^j |C_i| <0.99^{-\rlevel^\rcall+1}$.
Then the assumption that $\rcall$ is splittable implies $j \geq 1$.

If $j=p$, then we can set $\Fc$ to be a singleton that contains a copy of $\Rc$. In particular, the set of active vertices remains unchanged.
So now assume $j < p$. Since we ordered $C_i$s in  a non-increasing order of the number of vertices, the maximality of $j$ implies that 
$\sum_{i=1}^j |C_i| \geq \frac{1}{2} \cdot 0.99^{-\rlevel^\rcall+1}$. 
Consequently,
$$\sum_{i=j+1}^p |C_i| < (1-0.495) \cdot 0.99^{-\rlevel^{\rcall}} < 0.99^{-\rlevel^{\rcall}+1}.$$
Hence, we can split $\VtxActive^\rcall$ into $\bigcup_{i=1}^j C_i$
and $\bigcup_{i=j+1}^p C_i$ for the active sets of the subproblems of $\mathcal{F}$.
Note that in this case $\mathcal{F}$ is of size two.
\end{proof}
A split node $x$ satisfies the following properties:
$\rcall(x)$ is a splittable subproblem and the family of 
subproblems of children of $x$ satisfy the properties promised by \cref{lem:splittable} for the family $\mathcal{F}$. Moreover, we require that all the children of a split node are free nodes.

We remark that if a split node $x$ is clean, then all the children of $x$ are clean as well.

We also make the following observation.
\begin{lemma}\label{lem:split-compose} 
Let $x$ be a clean split node with children $\mathcal{Y}$. 
For every $y \in \mathcal{Y}$, let $\rcall'_y$ be
a clean subproblem extending $\rcall(y)$ completely.
Define a subproblem $\rcall'$ as 
\begin{align*}
\VtxTaken^{\rcall'} &= \bigcup_{y \in \mathcal{Y}} \VtxTaken^{\rcall'_y}, & 
\VtxDeleted^{\rcall'} &= \bigcup_{y \in \mathcal{Y}} \VtxDeleted^{\rcall'_y}, & 
\degord^{\rcall'} &= \bigcup_{y \in \mathcal{Y}} \degord^{\rcall'_y},\\
\VtxActive^{\rcall'} &= \emptyset, & \poslimit^{\rcall'} & = \emptyset, & 
\rlevel^{\rcall'} &= 0. 
\end{align*}
Then, $\rcall'$ is clean and extends completely $\rcall(x)$. 
\end{lemma}
\begin{proof}
Most of the asserted properties follow from the definitions in a straightforward manner; here we only
discuss the nontrivial ones.

First, note that $\degord^{\rcall'}$ is a well-defined function with domain
$\VtxTaken^{\rcall'}$. This is because each $\degord^{\rcall'_y}$ extends $\degord^{\rcall(x)}$,
and the sets $\{\VtxTaken^{\rcall'_y} \setminus \VtxTaken^{\rcall(x)}\colon y \in \mathcal{Y}\}$
are pairwise disjoint, because
$\{\VtxActive^{\rcall(y)}\colon y \in \mathcal{Y}\}$ is a partition of $\VtxActive^{\rcall(x)}$. 
Thus, $\rcall'$ is indeed a subproblem extending completely $\rcall(x)$. 
It remains to show that it is clean, that is, there are no offending vertices in 
$\VtxTaken^{\rcall'}$ (note here that $\VtxActive^{\rcall'}$ is empty).

Consider first a vertex $v \in \VtxTaken^{\rcall'_y} \setminus \VtxTaken^{\rcall(x)}$ for some $y \in \mathcal{Y}$. 
Then, since $\VtxActive^{\rcall(y)}$ is nonadjacent to $\VtxFree^{\rcall(y)} \setminus \VtxActive^{\rcall(y)}$, we have that $N_G(v) \subseteq \VtxActive^{\rcall(y)} \cup \VtxTaken^{\rcall(x)} \cup \VtxDeleted^{\rcall(x)}$. We 
infer that since $v$ is not offending in~$\rcall'_y$, due to this subproblem being clean, it is also not offending
in $\rcall'$.

Consider now a vertex $v \in \VtxTaken^{\rcall(x)}$.
Let $u \in N_G(v) \cap \VtxTaken^{\rcall'}$ be such that $\degord^{\rcall'}(u) < \degord^{\rcall'}(v)$. 
Then either $u \in \VtxTaken^{\rcall(x)}$ and $\degord^{\rcall(x)}(u) < \degord^{\rcall(x)}(v)$,
or $u \in \VtxActive^{\rcall(x)}$ and $\poslimit^{\rcall(x)}(u) < \degord^{\rcall(x)}(v)$.
Since $v$ is not offending in $\rcall(x)$, the number of such vertices $u$
is bounded by $d$. Hence, $v$ is not offending in $\rcall'$. 
This completes the proof of the lemma.
\end{proof}

\paragraph{Branch node.} 
Branch node $x$ is additionally labelled with a \emph{pivot} $\pivot^x \in \VtxActive^{\rcall(x)}$.
Intuitively, in the branching we decide whether $\pivot^x$ belongs to the constructed solution or not.
In case we decide to include $\pivot^x$ in $\VtxTaken$, we also guess its left neighbors and their positions in our fixed degeneracy ordering, as well as several other useful pieces of information.

Formally, a branch node $x$ has a child $y^x$ such that $\rcall(y^x)$ is the result of deleting $\pivot^x$
from $\rcall(x)$.
We call the edge pair $xy^x$ of the subproblem tree the \emph{failure branch} at node $x$, 
while $y^x$ is the \emph{failure child} of $x$. 
Note that if $x$ is clean, then the failure child $y^x$ is also clean.

Additionally, for every tuple $\mathcal{D} = (D,\degord,(D_u)_{u \in D'})$,
where
\begin{itemize}
 \item 
$D \subseteq \VtxActive^{\rcall(x)} \cap N_G(\pivot^x)$ is
of size at most $\quota(\degord^{\rcall(x)}, \pivot^x,\poslimit^{\rcall(x)}(\pivot^x))$, 
\item
$\degord$ is a position guess for $\rcall$ and $D'$, where $D' \coloneqq D \cup \{\pivot^x\}$,
and 
\item $(D_u)_{u \in D'}$ is a left neighbor guess for $\rcall$ and $D'$
such that $D_{\pivot^x} = \emptyset$,
\end{itemize} 
the branch node $x$ has a child $z^x_{\mathcal{D}}$. This child is associated with the subproblem $\rcall(z^x_{\mathcal{D}})$
that is the result
of taking $D'$ at positions $\degord$ with left neighbors $(D_u)_{u \in D'}$
in the subproblem $\rcall$.
We call each edge $xz^x_{\mathcal{D}}$ of the subproblem tree a \emph{success branch} at node $x$,
and $z^x_{\mathcal{D}}$ is a \emph{success child} of $x$.

We remark that the children $z^x_{\mathcal{D}}$ may not be clean even if $x$ is clean.
Therefore, we require that every child $z^x_{\mathcal{D}}$ is a filter node. 
The child of $z^x_{\mathcal{D}}$, if present, is denoted as $s^x_{\mathcal{D}}$
and is called a \emph{success grandchild} of $x$. We require that all the success grandchildren of $x$ are free nodes.

Observe that for success children of a branch node, there are at most $n^d$ choices for $D$, at most $n^{d+1}$ choices for $\degord$, and at most $n^d$ choices for each $D_u$, $u \in D$. This
    gives at most $n^{d^2+2d+1} = n^{(d+1)^2}$ choices for the tuple~$\mathcal{D}$, and consequently this is an upper bound on the number of success children.

Finally, we make the following observation that follows immediately from \cref{lem:lucky-take}:
\begin{lemma}\label{lem:lucky-branch}
Assume that a branch node $x$ is lucky.
If $\pivot^x \notin S^\ast$, then the failure child $y^x$ is lucky.
Otherwise, if $\pivot^x \in S^\ast$, then for
\begin{align*}
D &= \{u \in N_G(\pivot^x) \cap \VtxActive^{\rcall(x)}\cap S^\ast~|~\degord^\ast(u) < \degord^\ast(\pivot^x)\},\\
\degord &= \degord^\ast|_{\VtxTaken^{\rcall(x)} \cup D'},\\
D_u &= \{w \in N_G(u) \cap \VtxActive^{\rcall(x)}\cap S^\ast \setminus D'~|~\degord^\ast(w) < \degord^\ast(u)\},
\end{align*}
we have $|D| \leq d$, $\degord$ is a position guess for $\rcall(x)$
 and $D'$, $(D_u)_{u \in D'}$ is a left neighbor guess for $\rcall(x)$, $D'$, and~$\degord$
 satisfying $D_{\pivot^x} = \emptyset$, and $z^x_{\mathcal{D}}$ is a lucky success child of $x$.
\end{lemma}

\paragraph{Free node.}
For free nodes, we put three restrictions:
\begin{itemize}
\item a free node is a success grandchild of a branch node or a child of a split node;
\item the children of a free node have the same level as the free node itself; and
\item a child of a free node is clean or is a filter node.
\end{itemize}
Recall also the requirements stated in the above sections:
\begin{itemize}
 \item Every child of a split node is a free node.
 \item Every success grandchild of a branch node is a free node.
\end{itemize}

Free nodes are essentially not used in the proof of \cref{thm:deg-branching}. Precisely, we use a trivial subroutine of handling them that just passes the same instance to the child, which can be a node of a different type. 
Free nodes will play a role in further results, precisely in the proof of \cref{thm:main}, where they will be used as placeholders for additional branching steps, implemented using non-trivial subroutines for handling free nodes. The requirements stated above boil down to placing free nodes as children of split nodes and as success grandchildren of branch nodes, which means that the algorithm is allowed to perform the discussed additional guessing at these points. As we will see in the analysis, both passing through a split node and through a success branch correspond to some substantial
progress in the recursion, which gives space for those extra branching steps.

\paragraph{Final remarks.}
The requirements that the root node is clean, that success children of
a branch node are filter nodes, and that a child of a free node is either clean or a filter node, ensure the following property: every node of the subproblem tree is clean
unless it is a filter node. 

In a subproblem tree, the level of a child node equals the level of the parent,
unless the parent node is a split node, in which case the level of a child is one less
than the level of a parent. In particular, on a root-to-leaf path in a subproblem tree 
the levels do not increase.

Also, if $y$ is a child of $x$, then $\VtxActive^{\rcall(y)} \subseteq \VtxActive^{\rcall(x)}$. 
Furthermore, $\VtxActive^{\rcall(y)} = \VtxActive^{\rcall(x)}$
can happen only when $x$ is a split node, a filter node, or a free node. 
Taking into account the restrictions on the parents of filter and free nodes, 
       we infer the following:
\begin{lemma}\label{lem:subproblem-tree-depth}
The depth of a subproblem tree rooted in a node $x$
is at most $\Oh(|\VtxActive^{\rcall(x)}| + \rlevel^{\rcall(x)})$.
\end{lemma}

\subsection{Branching strategy}

A \emph{branching strategy} is a recursive algorithm that, given a 
subproblem $\rcall$, creates a node $x$ of a subproblem tree with $\rcall(x) = \rcall$,
decides the type of $x$, 
appropriately constructs subproblems for the children of $x$, and recurses on those children subproblems. The subproblem trees returned by the recursive calls are then attached by making their roots children of $x$.
A branching strategy can pass some additional information to the child subcalls; for instance, after creating a split node, the call should ask all the subcalls to create free nodes, while a branch node should tell all its success children to be filter nodes with free nodes as their children. 
We require that the subproblem tree constructed by a branching strategy is a subproblem tree, as defined in the previous section.

A few remarks are in place.
If the level of $\rcall$ is $0$, the branching strategy has no option but to create a leaf
node and stop (unless the parent or grandparent asks it to perform a filter or free node first).
If a branching strategy makes a decision to create a filter node or a split node, 
there are no more decisions to make: the filtering step works deterministically,
and for the split node we always create child subproblems using \cref{lem:splittable}. 

If a branching strategy decides to make a branch node, the only remaining decision
is to choose the pivot $\pivot^x$; after this, the failure branch and the success branches are defined deterministically.
The method of choosing the pivot is the cornerstone of our combinatorial analysis, and is presented in the remainder of this section.
(We will also use the freedom of sometimes \emph{not} choosing to create a split node, even if the
 current subproblem is splittable.)

Finally, the definition of a subproblem tree allows only a limited freedom of creating
free nodes: they have to be children of a split node or success grandchildren of a branch node. 
On the other hand, we would like to leave the freedom of how to handle free nodes to applications.
Thus, a branching strategy is parameterized by a subroutine that handles free nodes:
the subroutine is called by the branching strategy when handling a child of a split node
or a success grandchild of a branch node, and asked to create the family of subproblems for children.

For \cref{thm:deg-branching}, we do not use the possibility of creating free nodes. Formally, the subroutine handling free nodes, invoked at a subproblem $\rcall$, returns a one-element family $\{\rcall\}$, that is, creates a dummy free node $x$ with a single child $y$
such that $\rcall(x) = \rcall(y)$. 

An application of a branching strategy to a graph $G$ is a shorthand
for applying the branching strategy to a subproblem $\rcall$ defined by:
\begin{itemize}
 \item 
$\VtxTaken^\rcall = \VtxDeleted^\rcall = \emptyset$;
\item 
$\VtxActive^\rcall = V(G)$;
\item 
$\rlevel^\rcall = \lceil -\log_{0.99}(|V(G)|+1) \rceil$, and
\item $\poslimit^\rcall(v) = 1$ for every $v \in V(G)$.
\end{itemize}
Note that such a subproblem $\rcall$ is clean and lucky, regardless of the choice of $S^\ast$ and $\degord^\ast$.

\cref{sec:pivot} is devoted to branching strategies that lead to a quasipolynomial
running time bounds. Formally, we prove there the following:
\begin{lemma}\label{lem:branching-strategy}
For every fixed pair of integers $d$ and $t$, and every subroutine handling free nodes,
there exists a branching strategy that,
when applied to an $n$-vertex $C_{>t}$-free graph $G$, creates a recursion tree
that is a subproblem tree such that on every root-to-leaf path there
are $\Oh(\log n)$ split nodes and $\Oh(\log^3 n)$ success~branches.
If the input graph is $P_t$-free, the bound on the number of success~branches on a single root-to-leaf path improves to $\Oh(\log^2 n)$. 

Furthermore, the branching strategy takes polynomial time to decide on the type of the node
and on the choice of the branching pivot (in case the node is decided to be a branch node). 
\end{lemma}

Let us now show how \cref{lem:branching-strategy} implies \cref{thm:deg-branching}
\begin{proof}[of \cref{thm:deg-branching}.]
Fix a set $S^\ast \subseteq V(G)$ maximizing $\wei(S^\ast)$ subject to $G[S^\ast]$
being of degeneracy at most $d$. Fix also a degeneracy ordering $\degord^\ast$ of $G[S^\ast]$. 
This allows us to speak about lucky nodes.

Consider the branching strategy provided by \cref{lem:branching-strategy} with 
the discussed dummy subroutine handling free nodes.

By \cref{lem:subproblem-tree-depth}, the depth of the subproblem tree is $\Oh(n)$.
Furthermore, every node of the subproblem tree has at most $1 + n^{(d+1)^2}$ children, and the only nodes that may have more than one child are branch nodes and split nodes. Since each branch node has only one failure child, while on every root-to-leaf path there are $\Oh(\log n)$ split nodes and $\Oh(\log^3 n)$ success branches, it follows that the subproblem tree has $n^{\Oh(\log^3 n)}$ nodes. 

For every clean node $x$ in the generated subproblem tree, we would like
to compute a clean subproblem $\rcall'(x)$ that extends $\rcall(x)$ completely
and, subject to that, maximizes $\wei(\VtxTaken^{\rcall'(x)})$. 
Observe that for the answer to $\rcall'(r)$ at the root node $r$ 
we can return $\VtxTaken^{\rcall'(r)}$ as the sought subset $S$: the cleanness
of $\rcall'(r)$ implies that $G[\VtxTaken^{\rcall'(r)}]$ is $d$-degenerate,
while a subproblem of level $0$ with $\VtxTaken = S^\ast$,
$\VtxDeleted = V(G) \setminus S^\ast$, and $\degord = \degord^\ast$
extends completely $\rcall(r)$. 

For a leaf node $x$, the only option is to return $\rcall(x)$.
For a split node $x$ with children $\mathcal{Y}$, we apply
\cref{lem:split-compose} to $\{\rcall'(y)~|~y \in \mathcal{Y}\}$,
thus obtaining a clean subproblem $\rcall'(x)$. 
For a branch node $x$, we take $\rcall'(x)$ to be 
the subproblem with the maximum weight of $\VtxTaken$
among subproblem $\rcall'(y^x)$ for the failure child~$y^x$
and subproblems $\rcall'(s^x_{\mathcal{D}})$ for success grandchildren $s^x_{\mathcal{D}}$. 
Note that both $\rcall'(y^x)$ and every $\rcall'(s^x_{\mathcal{D}})$ are clean subproblems
that extending completely $\rcall(x)$. 
Finally, for a dummy free node $x$ with a child $y$, $\rcall'(x)$ equals $\rcall'(y)$.

To finish the proof, it suffices to argue that for every lucky node $x$, we have
\begin{equation}\label{eq:good-result}
\wei(\VtxTaken^{\rcall'(x)}) \geq \wei(S^\ast \cap (\VtxTaken^{\rcall(x)} \cup \VtxActive^{\rcall(x)})).
\end{equation}
We prove this fact by a bottom-up induction on the subproblem tree.

For a lucky leaf node the claim is straightforward, as being lucky implies that
$\VtxTaken^{\rcall(x)} \subseteq S^\ast$ while being a leaf node implies
$\VtxActive^{\rcall(x)} = \emptyset$. 
For a lucky split node, note that all children are also lucky;
the claim is immediate by the inductive assumption for the children and the construction
of \cref{lem:split-compose}. 
For a lucky branch node $x$, \cref{lem:lucky-branch} implies that $x$ has a lucky
failure child or a lucky success grandchild; the claim follows from the inductive assumption
of the said lucky (grand)child.
Finally, for a lucky dummy free node, the claim is immediate by the inductive assumption on its child.

This finishes the proof of \cref{thm:deg-branching}.
\end{proof}

%% file: branching-analysis.tex
This section is devoted to the proof of \cref{lem:branching-strategy}.
Recall that a branching strategy is essentially responsible for choosing whether the current node of the tree is a split node or a branch node and, in the latter case, choosing the branching pivot. 

The following observation, immediate from the definition of a success branch, 
will be the basic source of progress.

\begin{lemma}\label{lem:success-quota}
Let $x$ be a branch node and $y \coloneqq s^x_{\mathcal{D}}$ be a successful
grandchild of $x$, where $\mathcal{D} = (D, \degord, (D_u)_{u \in D'})$ and $D' = D \cup \{\pivot^x\}$. 
Then, for every $u \in N_G(\pivot^x)$, 
we have either $u \notin \VtxActive^{\rcall(y)}$ or
$\poslimit^{\rcall(y)}(u) > \degord^{\rcall(y)}(\pivot^x)$.
Consequently, 
for every $u \in N_G(\pivot^x) \cap \VtxActive^{\rcall(y)}$, it holds that
$$\quota(\degord^{\rcall(y)}, u, \poslimit^{\rcall(y)}(u)) \leq \quota(\degord^{\rcall(x)}, u, \poslimit^{\rcall(x)}(u)) - 1.$$
\end{lemma}
\begin{proof}
If $u \in \VtxActive^{\rcall(y)}$, then in particular $y \notin D$. 
Then, in the definition of $\poslimit^{\rcall(z^x_{\mathcal{D}})}(u)$, 
one term over which the maximum is taken is equal to $1 + \degord(\pivot^x)$
and $\degord(\pivot^x) = \degord^{\rcall(y)}(\pivot^x)$. 
Thus, $\pivot^x$, which is in $\VtxActive^{\rcall(x)} \cap \VtxTaken^{\rcall(y)}$, is taken into account in
$\quota(\degord^{\rcall(y)}, u, \poslimit^{\rcall(y)}(u))$,
but does not contribute to
$\quota(\degord^{\rcall(x)}, u, \poslimit^{\rcall(x)}(u))$ due to being not included in $\VtxTaken^{\rcall(x)}$.
\end{proof}

\subsection{Quasipolynomial branching strategy in $P_t$-free graphs}

To obtain the promised bounds for $P_t$-free graphs,
we closely follow the arguments for \textsc{MWIS} from~\cite{PPR20SOSA}.

Let us fix $d$, $t$, and a subroutine handling free nodes. 

The branching strategy chooses a leaf node whenever possible (the level of the 
subproblem is zero) and a split node whenever possible (the subproblem is splittable). 
Given a subproblem of positive level that is not splittable, the branching strategy
creates a branch node.

For a $P_t$-free graph $G$ and a pair $\{u,v\} \in \binom{V(G)}{2}$, the \emph{bucket}
of $\{u,v\}$ is a set $\Bc^G_{u,v}$ that consists of all induced paths with endpoints $u$
and $v$. Two remarks are in place. First, $\Bc^G_{u,v} \neq \emptyset$ if and only if 
$u$ and $v$ are in the same connected component of $G$. Second, an $n$-vertex
$P_t$-free graph has at most $n^{t-1}$ induced paths; hence, all buckets can be enumerated
in polynomial time.

For $\varepsilon > 0$, a vertex $x \in V(G)$ is \emph{$\varepsilon$-heavy} if
the neighborhood $N[x]$ intersects more that $\varepsilon$ fraction of paths
from more than $\varepsilon$ fraction of buckets, that is, 
\[
  \left| \left\{\{u,v\} \in \binom{V(G)}{2}~\Big|~
\left|\{P \in \Bc^G_{u,v}~|~N[x] \cap V(P) \neq \emptyset\}\right| > \varepsilon |\Bc^G_{u,v}|\right\}\right| > \varepsilon \binom{|V(G)|}{2}.
\]
Note that empty buckets cannot contribute towards the left hand side of the inequality above.

We make use of the following lemma.
\begin{lemma}[\cite{PPR20SOSA}]\label{lem:Pt-free-heavy}
A connected $P_t$-free graph contains a $\frac{1}{2t}$-heavy vertex.
\end{lemma}
In our case, a subproblem $\rcall$ at a branch node may not have
$G[\VtxActive^\rcall]$ connected, but $G[\VtxActive^\rcall]$ contains
a large connected component if $\rcall$ is not splittable.
\begin{corollary}\label{cor:Pt-free-heavy}
If $\rcall$ is a not splittable subproblem of positive level, then 
$G[\VtxActive^\rcall]$ contains a $\frac{1}{3t}$-heavy vertex.
\end{corollary}
\begin{proof}
Recall that $n^\rcall := |\VtxActive^\rcall| < 0.99^{-\rlevel^\rcall}$.
Since $\rcall$ is of positive level but not splittable, there is a connected component $D$
of $G[\VtxActive^{\rcall}|$ of size at least $0.99^{-\rlevel^\rcall+1}$. 
That is, $|D| \geq 0.99n^{\rcall}$. 

\cref{lem:Pt-free-heavy} asserts there is a $\frac{1}{2t}$-heavy vertex in $G[D]$. 
We have
$$\binom{|D|}{2} \geq \binom{\lceil 0.99n^{\rcall} \rceil }{2} \geq 0.9 \binom{n^\rcall}{2}.$$
Thus, $x$ is $\frac{0.9}{2t}$-heavy in $G$ and $\frac{0.9}{2t} > \frac{1}{3t}$. 
\end{proof}
The branching strategy chooses a $\frac{1}{3t}$-heavy vertex of $G[\VtxActive^{\rcall(x)}]$ as a branching pivot $\pivot^x$ at a branch node $x$.
Recall that at a branch node $x$ the subproblem $\rcall(x)$ is of positive level and not splittable. Hence, \cref{cor:Pt-free-heavy} asserts the existence of such a pivot. As already discussed, all buckets can be enumerated in polynomial time and thus such a heavy vertex can be identified.

The bound on the number of split nodes on any root-to-leaf path in the subproblem tree is immediate
from the fact that there are $\Oh(\log n)$ levels. 
It remains to argue that any root-to-leaf path has $\Oh(\log^2 n)$ success branches. 

To this end, let $Q$ be a maximal upward path in the subproblem tree that consists of nodes of the same level $\rlevel$. 
As there are $\Oh(\log n)$ possible levels, it suffices to prove that $Q$ contains $\Oh(\log n)$ success branches.

Consider then a success branch on $Q$: a branch node $x$ and its success grandchild 
$y \coloneqq s^x_{\mathcal{D}}$ for $\mathcal{D} = (D, \degord, (D_u)_{u \in D'})$, $D' = D \cup \{\pivot^x\}$. \cref{lem:success-quota} suggests the following potential at node $x$:

\[
\mu(x) \coloneqq \sum_{\{u,v\} \in \binom{\VtxActive^{\rcall(x)}}{2}} 
  \log_2\left[1+ \sum_{P \in \Bc_{u,v}^{G[\VtxActive^{\rcall(x)}]}} \sum_{u \in V(P)} (1+\quota(\degord^{\rcall(x)}, u, \poslimit^{\rcall(x)}(u)))\right].
\]

Since we branch on a $\frac{1}{3t}$-heavy vertex and the innermost sums in the definition above
are bounded by $\Oh(dt)$, we have that for some universal constant $c$:
\[
\mu(x) - \mu(y) \geq \frac{c}{dt^3} \binom{|\VtxActive^{\rcall(x)}|}{3}.
\]
Let $x_0$ be the topmost node of $Q$ and $n_Q = |\VtxActive^{\rcall(x_0)}|$.
Since all nodes of $x_0$ are of the same level, we have $|\VtxActive^{\rcall(x)}| \geq 0.99n_Q$
for every branch node $x$ on $Q$. 
We infer that 
\[
  \mu(x) - \mu(y) = \Omega(n_P^2).
\]
We have $\mu(x_0) = \Oh(n_P^2 \log n_P)$ while $\mu(x) \geq 0$ for any node $x$.
Consequently, $Q$ may contain $\Oh(\log n_P) = \Oh(\log n)$ success branches, as desired.

This finishes the proof of \cref{lem:branching-strategy} for $P_t$-free graphs.

\subsection{Quasipolynomial branching strategy in $C_{>t}$-free graphs}

We now prove \cref{lem:branching-strategy}
for $C_{>t}$-free graphs. 
Hence, let us fix $d$, $t$, and a subroutine handling free nodes. W.l.o.g. assume $t$ is even and $t\geq 6$. Recall that our goal is to design a branching strategy, which given a subproblem $\rcall$ should decide the type of the node created for $\rcall$ and, in case this type is the branch node, choose a suitable branching pivot.
We will measure the progress of our algorithm by keeping track of the number of some suitably defined objects in the graph.

A \emph{connector} is a graph with three designated vertices, called 
\emph{tips}, obtained in the following way.
Take three induced paths $Q_1,Q_2,Q_3$; 
here we allow degenerated, one-vertex paths.
The paths $Q_1,Q_2,Q_3$ will be called the \emph{legs} of the connector.
The endvertices of $Q_i$ are called $a_i$ and $b_i$.
Now, join these paths in one of the following ways:
\begin{enumerate}[label=\alph*)]
\item identify $a_1,a_2$, and $a_3$ into a single vertex, i.e., $a_1=a_2=a_3$, or
\item add edges $a_1a_2$, $a_2a_3$, and $a_1a_3$.
\end{enumerate}
Furthermore, if the endvertices are identified, then at most one leg may be degenerated; thus, $b_1,b_2,b_3$ are pairwise different after the joining.
There are no other edges between the legs of the connector. 
The vertices $b_1$, $b_2$, and $b_3$ are the \emph{tips} of the connector,
and the set $\{a_1,a_2,a_3\}$  is called the \emph{center} (this set can have either three or one element);
see~\cref{fig:tripod}.
If one of the paths forming the connector has only one vertex, and the endvertices were identified, then the connector is just an induced path with the tips being the endpoints of the path plus one internal vertex of the path. 
Note that, given a connector as a graph and its tips,
the legs and the center of the connector are defined uniquely. 

\begin{figure}[tb]
\begin{center}
\includegraphics{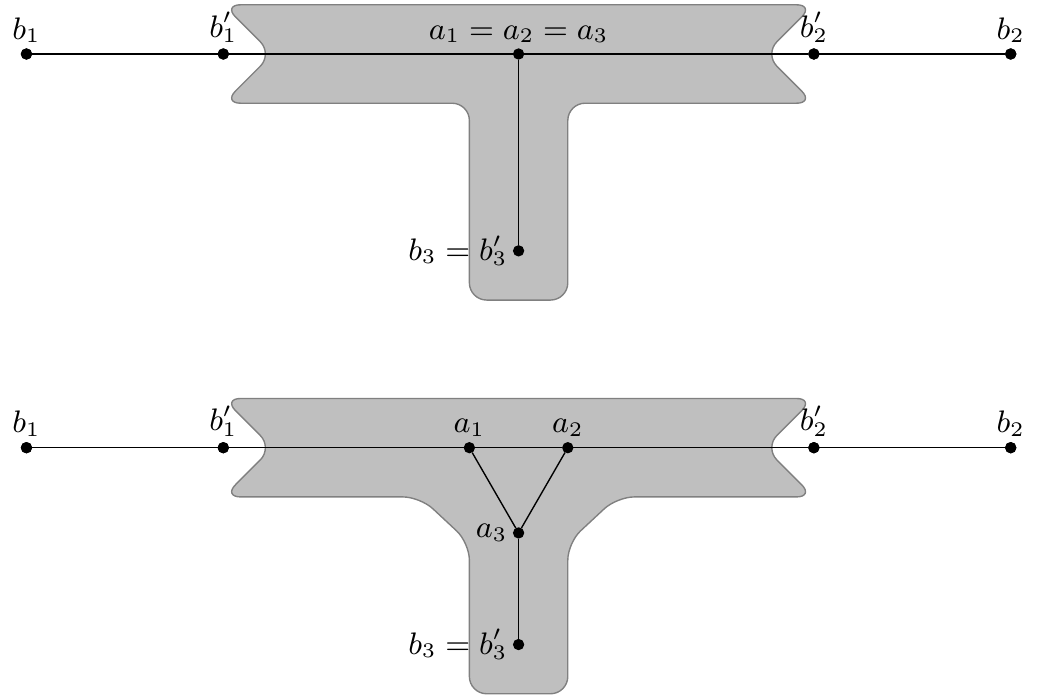}
\caption{Two connectors with two long legs and one short leg;
  one connector with $a_i$s identified and one with $a_i$s forming a triangle.
Vertices $b_i'$ are the tips of the core $T$ of the connector. The gray area
depicts the set $T^\ast$.}\label{fig:tripod}
\end{center}
\end{figure}

We will need the following folklore observation.
\begin{lemma}\label{lem:connector}
Let $G$ be a graph, $A \subseteq V(G)$ be a set consisting of exactly three vertices in the same connected component of $G$, and 
let $A \subseteq B \subseteq V(G)$ be an inclusion-wise minimal set such that $G[B]$ is connected. 
Then the graph $G[B]$ with the set $A$ as tips is a connector.
\end{lemma}
\begin{proof}
Let $A = \{u,v,w\}$, let $P_{uv}$ be a shortest path from $u$ to $v$ in $G[B]$ and let $P_w$ be a shortest path from $w$ to $V(P_{uv})$ in $G[B]$. 
By minimality of $B$, we have $B = V(P_{uv}) \cup V(P_w)$.

If $w \in V(P_{uv})$ (equivalently, $|V(P_w)| = 1$), then $G[B]$ is a path and we are done. Otherwise, 
let $q \in V(P_w) \cap V(P_{uv})$ be the endpoint of $P_w$ distinct than $w$ and let $p$ be the unique neighbor of $q$ on $P_w$. 
By the minimality of $P_w$, $p$ and $q$ are the only vertices of $P_w$ that may have neighbors on $P_{uv}$. 
If $p$ has two neighbors $x,y \in V(P_{uv})$ that are not consecutive on $P_{uv}$, then $G[B]$ remains connected after the deletion from $B$ of all vertices
on $P_{uv}$ between $x$ and $y$ (exclusive), a contradiction to the choice of $B$. 
Thus, $N(p) \cap V(P_{uv})$ consists of $q$ and possibly one neighbor of $q$ on $P_{uv}$. We infer that $G[B]$ is a connector with tips $u$, $v$, and $w$, as desired.
\end{proof}

A  \emph{tripod} is a connector where each of the paths $Q_1,Q_2,Q_3$ has at most $t/2+1$ vertices. 
A leg of a tripod is \emph{long} if it contains exactly $t/2+1$ vertices, and \emph{short} otherwise.
The \emph{core} of a connector $C$ with legs $Q_1,Q_2,Q_3$, is the tripod consisting of the first $t/2+1$ (or all of them, if the corresponding path $Q_i$ is shorter)
vertices of each path $Q_i$, starting from $a_i$.  
A tripod {\em{in $G$}} is a tripod that is an induced subgraph of~$G$.
Note that each tripod has at most $3t/2+3$ vertices, hence given $G$, we can enumerate all tripods in $G$ in time~$n^{\Oh(t)}$.

Let $T$ be a tripod in $G$ with legs $Q_1,Q_2,Q_3$. Let $L(T)$ denote the tips of $T$ which are endvertices of the long legs. We denote
$$T^\ast\coloneqq N_G[V(T) \setminus L(T)] \setminus L(T).$$ In other words, $T^\ast$ is the closed neighborhood of the tripod $T$ in $G$, except that we do not necessarily include the neighbors of tips of long legs and we exclude those tips themselves. Note that tips of short legs together with their neighborhoods are included in $T^\ast$.

We have the following simple observation.
\begin{lemma}\label{lem:tripod1}
For every tripod $T$ in $G$ and for every connected component $C$ of $G-T^\ast$, the component $C$ contains at most one tip of $L(T)$ and no tip of a short leg.
\end{lemma}
\begin{proof}
First, note that tips of short legs are contained in $T^\ast$, hence they are not contained in $G-T^\ast$.

For contradiction, without loss of generality assume that $b_1,b_2 \in C \cap L(T)$. Then, $Q_1$, $Q_2$, and a shortest path from $b_1$ to $b_2$ in $C$ yield an induced cycle on more than $t$ vertices in $G$,
a contradiction.
\end{proof}

Suppose $T$ is a tripod in $G$.
With each tip $b_i$ of $T$ we associate a \emph{bag} $B_i$ defined as follows: 
\begin{itemize}
 \item if $b_i$ is the endpoint of a long leg, then $B_i$ is the vertex
set of the connected component of $G-T^\ast$ that contains $b_i$; and
\item otherwise, $B_i = \{b_i\}$.
\end{itemize}
Note that \cref{lem:tripod1} implies that the bags $B_1,B_2,B_3$ are pairwise disjoint and nonadjacent in $G$, except for the corner case of two adjacent tips of short legs.

We now define \emph{buckets} that group tripods.
Each bucket will be indexed by an unordered triple of distinct vertices of $G$. 
Every tripod $T$ in $G$ with bags $B_1,B_2,B_3$ belongs to the bucket $\Bc_{u,v,w}^G$ for all triples $u,v,w$ such that $u \in B_1$, $v \in B_2$, and $w \in B_3$.
Note that thus, the buckets do not have to be pairwise disjoint.
The superscript $G$ can be omitted if the graph is clear from the context.

Observe the following.
\begin{lemma}\label{lem:tripod2}
For every $\{u,v,w\} \in \binom{V(G)}{3}$ and every $T \in \Bc_{u,v,w}$, there exists a connector $T'$ with tips $u,v,w$
whose core equals $T$.
Consequently, the bucket $\Bc_{u,v,w}$ is nonempty if and only if $u,v,w$ lie in the same connected component of $G$.
\end{lemma}
\begin{proof}
Let $Q_u, Q_v,Q_w$ be the legs of $T$ with tips $b_u$, $b_v$, $b_w$ and bags $B_u$, $B_v$, and $B_w$, respectively,
such that $u \in B_u$, $v \in B_v$, $w \in B_w$. Let $Q_u'$ be the concatenation of $Q_u$ and a shortest path from $b_u$ to $u$
in~$G[B_u]$. Similarly define $Q_v'$ and $Q_w'$. 

Recall that the bags $B_u$, $B_v$, and $B_w$ are pairwise
distinct and nonadjacent (except for the case of two adjacent tips of short legs). Hence, $Q_u'$, $Q_v'$, and $Q_w'$ form a connector $T'$ with tips $u$, $v$, and $w$.
Since $B_u \neq \{b_u\}$ only if the leg $Q_u$ is long, $T$ is the core of $T'$.
\end{proof}

The next combinatorial observation is critical to the complexity analysis.
\begin{lemma}\label{lem:tripod3}
Let $G$ be a connected $C_{> t}$-free graph, and let $u,v,w$ be three distinct vertices of $G$.
Let $X\subseteq V(G)$ be such that $G[X]$ is connected and no two of $u,v,w$ are in the same connected component of $G - N[X]$.
Then $N[X]$ intersects all tripods in $\Bc_{u,v,w}$.
\end{lemma}
\begin{proof}
Let $T \in \Bc_{u,v,w}$.
Let $Q_u, Q_v,Q_w$ be the legs of $T$ with tips $b_u$, $b_v$, $b_w$ and bags $B_u$, $B_v$, and $B_w$, respectively.
Let $T'$ be the connector for $T$ given by \cref{lem:tripod2} with legs $Q_u'$, $Q_v'$, and $Q_w'$. 
Since no two of $u,v,w$ lie in the same connected component of $G-N[X]$, the set $N[X]$ intersects at least two legs of $Q_u'$, $Q_v'$, and $Q_w'$.
Without loss of generality, assume that $N[X]$ intersects $Q_u'$ and $Q_v'$. 
Let $u' \in N[X] \cap V(Q_u')$ be the vertex of $N[X] \cap V(Q_u')$ that is farthest from $u$ on $Q_u'$
and similarly define $v' \in N[X] \cap V(Q_v')$. 
Then, the subpath of $Q_u'$ from $u'$ to the center of $T'$, the subpath of $Q_v'$ from $v'$ to the center of $T'$, and a shortest path from 
$u'$ to $v'$ with all internal vertices in $X$ form an induced cycle in $G$. This cycle has more than $t$ vertices
unless $u' \in V(Q_u)$ or $v' \in V(Q_v)$. Hence, $V(T) \cap N[X] \neq \emptyset$, as desired.
\end{proof}

For $\epsilon>0$, 
we will say that a vertex $v$ is \emph{$\epsilon$}-heavy (or just \emph{heavy}, if $\epsilon$ is clear from the context) if $N[v]$ intersects strictly more than $\epsilon$ fraction of tripods in at least $\epsilon$ fraction of buckets.
Note that the ``strictly more'' part makes empty buckets not count toward the $\epsilon$ fraction
of the buckets hit.

Intuitively, our branching strategy chooses as the
pivot a $(10^{-8}/t)$-heavy vertex of $G[\VtxActive^\rcall]$.
\cref{lem:success-quota} 
captures the source of the gain in a success~branch: the quotas of the neighbors of the pivot get reduced.

Unfortunately, it may happen that there is no $(10^{-8}/t)$-heavy vertex.
For an example, consider the case when $H \coloneqq G[\VtxActive^\rcall]$ is a long path. Then $\Bc_{u,v,w}$ consists of subpaths of $H$ of length at most $t$ containing the middle vertex of $\{u,v,w\}$, and
an arbitrary neighborhood $N[x]$ intersects tripods in roughly $(t/|V(H)|)$ fraction of all buckets. 

However, there is one particular scenario where a heavy vertex is guaranteed to~exist.
In a connected graph $G$, a set $X \subseteq V(G)$ is a \emph{connected three-way balanced separator} (C3WBS) if 
\begin{itemize}
\item $G[X]$ is connected;
\item the family of connected components of $G-N[X]$ can be partitioned into three sets $\mathcal{C}_1,\mathcal{C}_2,\mathcal{C}_3$
such that for every $i=1,2,3$ we have $|N[X] \cup \bigcup \mathcal{C}_i| \geq 0.1|V(G)|$. 
\end{itemize}
We observe that a small connected three-way balanced separator in a connected $C_{>t}$-free graph guarantees the existence of a heavy vertex.
\begin{lemma}\label{lem:tripodheavy}
Let $G$ be a connected $C_{> t}$-free graph and let $X\subseteq V(G)$ be a C3WBS. 
Then there exists a $(10^{-6}/|X|)$-heavy vertex in $G$.
\end{lemma}
\begin{proof}
Let $n = |V(G)|$.
If $n \leq 2$, then any vertex of $G$ is heavy, as there are no buckets.
If $2 < n \leq 100$, then any vertex $v \in V(G)$ with at least two neighbors, say $u_1,u_2$, is $10^{-6}$-heavy
as $N[v]$ hits all tripods in the bucket $\Bc_{u_1,u_2,v}$ by \cref{lem:separator} and there are fewer than $10^{6}$ buckets. 
Hence, we can assume $n > 100$. 

Let $\mathcal{C}_1,\mathcal{C}_2,\mathcal{C}_3$ be the partition of the connected components of $G-N[X]$ promised by the definition of the C3WBS $X$. 
Let $C_i = \bigcup \mathcal{C}_i$ for $i=1,2,3$.
Let $\mathcal{X}$ be the family of triples $\{u_1,u_2,u_3\} \in \binom{V(G)}{3}$ with $u_i \in C_i \cup N[X]$ for $i=1,2,3$.
We can estimate the size of $\mathcal{X}$ as follows:
\begin{align*}
|\mathcal{X}| = & |C_1| \cdot |C_2| \cdot |C_3| + |N[X]| \cdot \left(|C_1| \cdot |C_2| + |C_2| \cdot |C_3| + |C_3| \cdot |C_1|\right) \\
& + \binom{|N[X]|}{2} \cdot \left(|C_1| + |C_2| + |C_3|\right) + \binom{|N[X]|}{3} \\
\geq &\frac{1}{6} \; (|C_1| + |N[X]|) \cdot (|C_2|+|N[X]|-1) \cdot (|C_3|+|N[X]|-2) \\
\geq &\frac{1}{6} \; 0.1n \cdot (0.1n - 1) \cdot (0.1n - 2) \\
\geq &\frac{10^{-3}}{6} \; n \cdot (n - 10) \cdot (n - 20) \\
\geq &10^{-4} \cdot \binom{n}{3}.
\end{align*}
In the last inequality we have used the assumption $n > 100$. 

By \cref{lem:tripod3}, the set $N[X]$ intersects all tripods in at least $10^{-4}$ fraction of the
buckets. Hence, there exists $w \in X$ such that $N[w]$ intersects
at least $1/|X|$ fraction of tripods in at least $10^{-4}/|X|$ fraction of the buckets.
\end{proof}
Unfortunately, \cref{lem:separator} for $A = V(G)$ does not give us a C3WBS, but only a connected set $X$
of size at most $t$ such that every component of $G-N[X]$ has at most $|V(G)|/2$ vertices.
The example of a long path shows that the fraction $1/2$ cannot be improved while keeping
$X$ both connected and of constant size. The next lemma describes the scenario when \cref{lem:separator} does not return 
a C3WBS.
\begin{lemma}\label{lem:sep-c3wbs}
Let $G$ be a connected graph and $X \subseteq V(G)$ be such that $G[X]$ is connected and every connected component of $G-N[X]$
has at most $|V(G)|/2$ vertices. 
If $X$ is not a C3WBS, then there exist exactly two connected components of $G-N[X]$, each containing at least $0.4|V(G)|$ vertices.
\end{lemma}
\begin{proof}
Clearly, there is not enough vertices in $G$ for three such components. Assume then there is at most one such component; we show that $X$ is a C3WBS.
Let $C_1,C_2,\ldots,C_k$ be the connected components of $G-N[X]$ in the nonincreasing order of their sizes, that is, $|C_1| \geq |C_2| \geq \ldots \geq |C_k|$. 

Let $i_1 \geq 1$ be the minimum index such that
\[ \left| N[X] \cup \bigcup_{j=1}^{i_1} C_j \right| \geq 0.1|V(G)|. \]
If $i_1 > 1$, then $|C_{i_1}| \leq |C_1| < 0.1|V(G)|$. By the minimality of $i_1$,
\[ \left| N[X] \cup \bigcup_{j=1}^{i_1-1} C_j \right| < 0.1|V(G)|. \]
Hence, $|\bigcup_{j=1}^{i_1} C_j| \leq 0.2|V(G)|$. 
If $i_1 = 1$, then $|\bigcup_{j=1}^{i_1} C_j| = |C_1| \leq 0.5|V(G)|$. Hence, in both cases,
\[ |\bigcup_{j=1}^{i_1} C_j| \leq 0.5|V(G)|. \]

This allows us to define $i_2 > i_1$ to be the minimum index such that
\[ \left| N[X] \cup \bigcup_{j=i_1+1}^{i_2} C_j \right| \geq 0.1|V(G)|. \]
Similarly as before, if $i_2 > i_1+1$, then $|\bigcup_{j=i_1+1}^{i_2} C_j| \leq 0.2|V(G)|$.
If $i_2 = i_1+1$, then, since only $C_1$ is allowed to be of size at least $0.4|V(G)|$, we have
$|\bigcup_{j=i_1+1}^{i_2} C_j| = |C_{i_1+1}| \leq 0.4|V(G)|$.
Thus, in both cases,
\[ |\bigcup_{j=i_1+1}^{i_2} C_j| \leq 0.4|V(G)|. \]
We infer that:
\[ |N[X] \cup \bigcup_{j=i_2+1}^{k} C_j| \geq 0.1|V(G)|. \]
Hence the partition
\[ \mathcal{C}_1 = \{C_j~|~1 \leq j \leq i_1\}, \quad \mathcal{C}_2 = \{C_j~|~i_1+1 \leq j \leq i_2\}, \quad \mathcal{C}_3 = \{C_j~|~i_2+1 \leq j \leq k\} \]
proves that $X$ is a C3WBS, as desired.
\end{proof}

In the absence of a heavy vertex, we shift to a secondary branching strategy.
The secondary branching strategy:
\begin{enumerate}[label=(\arabic*)]
\item is initiated with a subproblem $\rcall$
and a set $X \subseteq \VtxActive^\rcall$ of size at most $t$ such that 
\begin{itemize}
\item $\rcall$ is not splittable;
\item $G[\VtxActive^\rcall]$ does not admit a $(10^{-8}/t)$-heavy vertex,
\item $G[X]$ is connected, and
\item every connected component of $G[\VtxActive^\rcall \setminus N[X]]$ has at most $|C_0|/2$ vertices, where $C_0$ is the largest connected component
of $G[\VtxActive^\rcall]$;
\end{itemize}
\item terminates (i.e., falls back to the primary branching strategy) if and only if when called at a subproblem $\rcall'$ that is splittable;\label{s:sec-terminate}
\item on every root-to-leaf path in the subproblem tree created by the recursion
there are $\Oh(\log^2 n^\rcall)$ success branches. (This will be shown at the end of section \ref{sec:pivot})
\end{enumerate}
Note that, in particular, 
the secondary branching strategy never forms a split node; at every node it either terminates
or selects a branching pivot, makes a branch node and successive filter and free nodes
for success branches. Hence, all subproblems in a subproblem tree created by the secondary branching strategy are of the same level.

We postpone the description of the secondary branching strategy to \cref{sec:secondary}.
Now, using it as a blackbox, we describe our primary strategy.

For a subproblem $\rcall$, the primary branching strategy makes the following decisions:
\begin{enumerate}
\item If $\rcall$ is of level $0$, make a leaf node and terminate.
\item If $\rcall$ is splittable, make a split node using \cref{lem:splittable}
and recurse on the constructed children.
\item If $G[\VtxActive^{\rcall}]$ contains a $(10^{-8}/t)$-heavy vertex $w$, 
  create a branch node $x$ and choose $w$ as the branching pivot $\pivot^x$.\label{step:b1}
\item Otherwise, let $C_0$ be the largest connected component of $G[\VtxActive^\rcall]$
(as $\rcall$ is not splittable, there is such connected component with at least $0.99|\VtxActive^\rcall|$ vertices),
  construct a set $X$ from \cref{lem:separator} for the graph $G[C_0]$
  and invoke the secondary branching strategy on $\rcall$ and $X$.\label{step:b2}
\end{enumerate}

We proceed with the analysis.
Consider the subproblem tree of the algorithm applied to the graph $G$.
The claim that every root-to-leaf path contains $\Oh(\log n)$ split nodes
is straightforward, because the root node has level $\lceil -\log_{0.99}(n+1) \rceil = \Oh(\log n)$,
the level of a child is never higher than the level of the parent, 
and the level of the split node is one higher than the level of its children.
To show the more difficult claim that every root-to-leaf path contains $\Oh(\log^3 n)$
success branches, it suffices to show that any upward path in the subproblem tree
consisting of nodes of the same level $\rlevel$ contains $\Oh(\log^2 n)$ success branches.

Consider such a maximal upward path $P$ with all nodes of level $\rlevel$. 
If the path $P$ contains a node $x$ where the branching strategy invoked the secondary branching strategy, 
then, since the secondary branching strategy terminates at a splittable node, the entire subpath $P_2$ of $P$ from $x$ downwards is contained
in the subtree corresponding to the call to the secondary branching strategy. Hence, $P_2$ contains $\Oh(\log^2 n)$ success branches. 
Let $P_1$ be the subpath of $P$ from $x$ upwards, or $P_1 = P$ if $P$ does not contain a node where the secondary branching strategy is invoked. 

It suffices to show that $P_1$ contains $\Oh(\log^2 n)$ success branches; note that these success branches correspond to branch nodes
of the primary branching strategy. We will actually show a stronger bound of $\Oh(\log n)$ success branches.

Let $n_P = |\VtxActive^{\rcall(r)}|$ where $r$ is the topmost node of $P_1$.
By the threshold at which we use split node in the primary strategy,
for every $x$ on $P_1$ except for possibly the bottom-most one, we have
\begin{equation}\label{eq:branching:large}
|\VtxActive^{\rcall(x)}| \geq 0.99n_P.
\end{equation}

Motivated by \cref{lem:success-quota}, we measure the progress using the following potential at a node $x$ on $P_1$:
\[
\mu(x) \coloneqq \sum_{\{u,v,w\} \in \binom{\VtxActive^{\rcall(x)}}{3}} 
  \log_2\left[1+ \sum_{T \in \Bc_{u,v,w}^{G[\VtxActive^{\rcall(x)}]}} \sum_{u \in V(T)} \left(1+\quota(\degord^{\rcall(x)}, u, \poslimit^{\rcall(x)}(u))\right)\right].
\]
Observe that when $y$ is a child of $x$ on $P$, we have $\mu(x)\geq \mu(y)$.

Consider a branch node $x$ with a success grandchild $y$ in $P_1$. Recall that $x$ is a branch node of the primary branching strategy and
corresponds to branching on a $(10^{-8}/t)$-heavy pivot $\pivot^x$. 
Since each of the innermost sums in the definition of $\mu$ is upper bounded by $\Oh(dt)$, by \cref{lem:success-quota} we may infer that 
$$\mu(x)-\mu(y)\geq \frac{c}{dt^3}\binom{|\VtxActive^{\rcall(x)}|}{3},$$
for some universal constant $c>0$.

On the other hand, since every bucket is of size $n^{\Oh(t)}$, we have that
$\mu(x) = \Oh(n_P^3 \log n_P)$ for every $x$ on $P_1$. 
Hence, there are $\Oh(\log n_P)$ success branches on $P_1$, as desired.

\subsubsection{Secondary branching strategy}\label{sec:secondary}

We now move to the explanation of the secondary branching strategy.

For a graph $H$, a set $C\subseteq V(H)$ such that $H[C]$ is connected, and distinct vertices $u,v \in N_H(C)$, 
a \emph{$C$-link} between $u$ and $v$ is a path $P$ in $H$ with the following properties:
\begin{itemize}
 \item $P$ has endpoints $u$ and $v$ and length at least $2$;
 \item all internal vertices of $P$ belong to $C$; and
 \item $P$ is an induced path in $H-E(H[N_H(C)])$ (i.e. $P$ is an induced path in $H$, except that we allow the existence of the edge $uv$).
 \end{itemize}
We now make a combinatorial observation that is critical to the analysis:
\begin{lemma}\label{lem:tripod4}
Let $H$ be a $C_{>t}$-free graph,
    let $X\subseteq V(H)$ be such that $H[X]$ is connected,
    and let $C$ be a connected component of $H-N[X]$.
Then every $C$-link has at most $t$ vertices.
\end{lemma}
\begin{proof}
Let $P$ be the $C$-link in question, let $u,v$ be its endpoints, and let $Q$ be a shortest path with endpoints $u,v$ and all internal vertices in $X$; such $Q$ exists because $u,v\in N[X]$ and $H[X]$ is connected.
Then $P \cup Q$ is an induced cycle in $H$. Thus, both $P$ and $Q$ have at most $t$ vertices.
\end{proof} 

Recall that the setting of the secondary branching strategy is as follows: 
we have
a subproblem $\rcall$ and a set $X \subseteq \VtxActive^\rcall$
such that $\rcall$ is not splittable (in particular, there is a connected component $C_1$ of $G[\VtxActive^\rcall]$ of size at least $0.99|\VtxActive^\rcall|$),
$G[\VtxActive^\rcall]$ contains no $(10^{-8}/t)$-heavy vertex,
$X \subseteq C_1$, $G[X]$ is connected, $|X| \leq t$, and every connected component of $G[C_1] \setminus N[X]$ is of size at most $|C_1|/2$.
Let $n^\rcall = |\VtxActive^\rcall|$ and $K = N_{G[\VtxActive^\rcall]}[X]$. 

As $G[\VtxActive^\rcall]$ contains no $(10^{-8}/t)$-heavy vertex and $|C_1| \geq 0.99n^\rcall$, $G[C_1]$ contains no $(10^{-7}/t)$-heavy vertex. 

If $G[C_1] \setminus K$ contains no connected component with at least $0.4 \cdot |C_1|$ vertices, then by \cref{lem:sep-c3wbs} $X$ is a C3WBS of $G[C_1]$ and \cref{lem:tripodheavy}
implies that $G_1$ contains a $(10^{-6}/t)$-heavy vertex, a contradiction.
Hence, there exists a component $C_2$ of $G[C_1] \setminus K$ with at least $0.4 \cdot |C_1|$ vertices. 

Let $Y$ be the result of the application of \cref{lem:separator} to $G[C_2]$; that is, $|Y| \leq t$,  $G[Y]$ is connected, and every connected component
of $G[C_2] \setminus N[Y]$ is of size at most $0.5|C_2|$. Let $L := N_{G[C_2]}(Y)$.
We make the following two observations:
\begin{lemma}\label{lem:secondary-XYfar}
The distance, in $G[C_1]$, between $X$ and $Y$, is more than $8t$. 
\end{lemma}
\begin{proof}
Let $P$ be a shortest path in $G[C_1]$ between a vertex of $K$ and a vertex of $L$. 
Then, $X' := X \cup Y \cup V(P)$ is connected in $G[C_1]$.
Since every connected component of $G[C_1] \setminus K$ is of size at most $0.5 |C_1|$ but $X$ is not a C3WBS of $G[C_1]$, 
\cref{lem:sep-c3wbs} implies that there are exactly two connected components of $G[C_1] \setminus K$ of size at least $0.4|C_1|$, one of which is $C_2$.
Since every connected component of $G[C_2] \setminus N[Y]$ is of size at most $0.5|C_2| \leq 0.25|C_1|$, there is at most one connected component 
of $G[C_1] \setminus N[X']$ that is of size at least $0.4|C_1|$. \cref{lem:sep-c3wbs} implies that $X'$ is a C3WBS of $G[C_1]$. 
As $G[C_1]$ does not admit a $(10^{-7}/t)$-heavy vertex, \cref{lem:tripodheavy} implies that $|X'| \geq 10t$. 
Consequently, $|V(P)| \geq 8t$, as desired.
\end{proof}

\begin{lemma}\label{lem:secondary-comps}
Let $\mathcal{D}$ be the family of connected components of $G[C_2] \setminus L$. 
For every $D \in \mathcal{D}$, we have $N[D] \cap L \neq \emptyset$
and there exists exactly one $D_0 \in \mathcal{D}$ with $N[D_0] \cap K \neq \emptyset$.
\end{lemma}
\begin{proof}
The first claim follows from the connectivity of $G[C_2]$. 
The existence of at least one component $D_0 \in \mathcal{D}$ with $N[D_0] \cap K \neq \emptyset$ follows from the connectivity of $G[C_1]$. 

Assume now there are two components $D_0,D_1 \in \mathcal{D}$ with $N[D_i] \cap K \neq \emptyset$ for $i=0,1$. 
Let $P_X$ be a shortest path between $D_0$ and $D_1$ with internal vertices in $K$ and let $P_Y$ be a shortest path between $D_0$ and $D_1$ in $L$. 
Then, connecting the endpoints of $P_X$ and $P_Y$ via $D_0$ and $D_1$ creates an induced cycle in $G[C_1]$ and \cref{lem:secondary-XYfar} implies that this cycle is longer than $t$. 
This is the desired contradiction.
\end{proof}

\begin{figure}[tb]
\begin{center}
\includegraphics{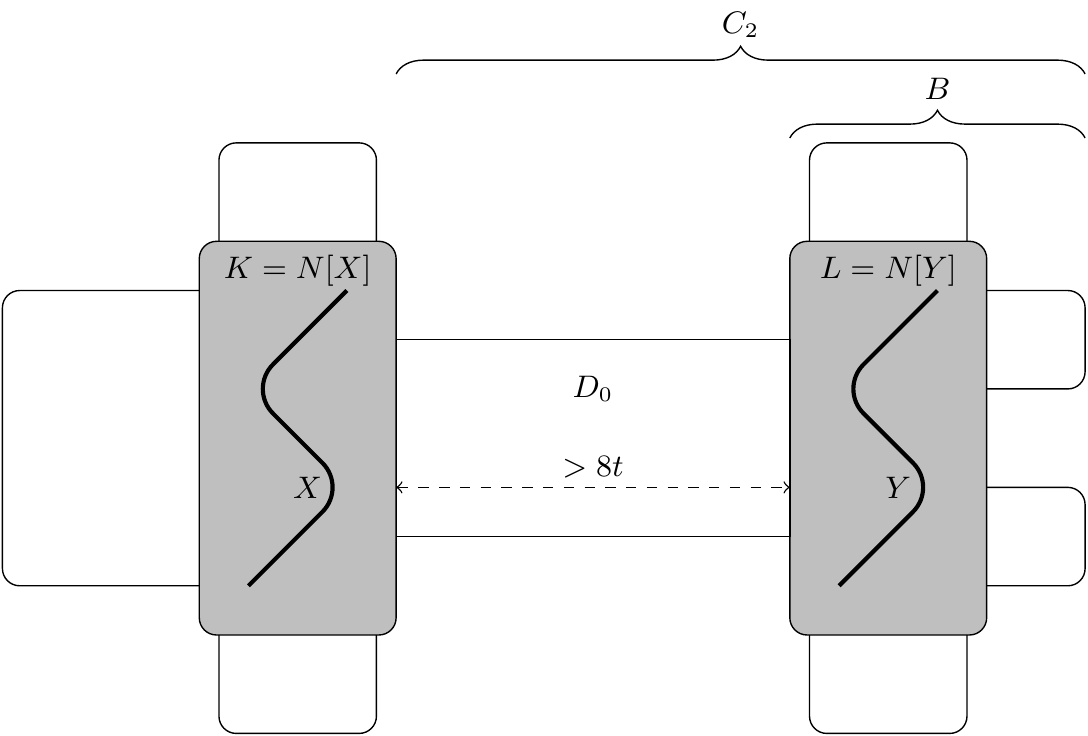}
\caption{The component $C_1$ in the scenario when the secondary branching is invoked.}\label{fig:secondary}
\end{center}
\end{figure}

Let 
$$B = L \cup \bigcup \left(\mathcal{D} \setminus \{D_0\}\right) = C_2 \setminus D_0,$$
where $\mathcal{D}$ and $D_0$ come from \cref{lem:secondary-comps}. See Figure~\ref{fig:secondary} for an illustration.
Recall that $|D_0| \leq 0.25|C_1|$ and $|C_2| \geq 0.4|C_1|$, which implies that $|D_0| \leq 0.5|C_2|$ and thus
we have $|B| \geq 0.2|C_1|$. 
Furthermore, $B$ is connected, by \cref{lem:secondary-comps} $N_{G[C_1]}(C_1 \setminus B) \subseteq L$
and by \cref{lem:secondary-XYfar}, the distance in $G[C_1]$ between $K$ and $B$ is more than $8t-2 > 7t$. 

Using the same argumentation as in the proof of \cref{lem:secondary-comps}, we have the following observation:
\begin{lemma}\label{lem:secondary-oneB}
For every $C' \subseteq C_2$, there exists at most one connected component of $G[C']$ that contains both a vertex of $B$ and a vertex adjacent to a vertex of $K$ in $G[C_1]$.
\end{lemma}
\begin{proof}
If there are two such components $D_0$ and $D_1$, then going via $D_0$, $K$, $D_1$, and $B$, one obtains an induced cycle that is longer than $t$ due to \cref{lem:secondary-XYfar}.
\end{proof}

The secondary branching strategy is allowed only to make 
branch nodes (and subsequent filter and free nodes at success branches)
and it always terminates whenever the current subproblem is splittable. 
The crux is to describe the choice of the branching pivot
when the current subproblem is not splittable. 

For an induced subgraph $H$ of $G[\VtxActive^\rcall]$, 
the {\emph{chip}} in $H$ is the vertex set $C'$ of a connected component of $G[V(H) \cap C_2]$ that contains both a vertex of $B$
and a vertex adjacent to a vertex of $K$. \cref{lem:secondary-oneB} implies that $H$ contains at most one chip. 

Assume that we are now considering a subproblem $\rcall'$ extending $\rcall$. 
Observe that if $H \coloneqq G[\VtxActive^{\rcall'}]$ has no chip, then $\rcall'$ is splittable as $|C_1| \geq 0.99n^\rcall$
and $|B| \geq 0.2|C_1|$. Hence, the secondary branching strategy terminates at $\rcall'$.

If $H$ contains a chip $C'$ with $|N_H(C')| = 1$, we choose
the unique element of $N_H(C')$ as the branching pivot. 
Note that after branching on such a pivot, for every child $z$ of the current branch node,
both in the failure and in the success branches,
(the remainder of) the chip $C'$ in $G[\VtxActive^{\rcall(z)}]$ 
is in a different connected component than any vertex of $K$. 
Hence, the failure child and the success grandchildren of $x$ are splittable. 

We are left with the following case: in $H$, 
there is exactly one chip $C'$ and $C'$ satisfies~$|N_H(C')| \geq 2$. 
We define (secondary branching) buckets as follows. 
The buckets are indexed by an unordered pair $\{u,v\} \in \binom{N_H(C')}{2}$. For such a choice of $u,v$,
the bucket $\Lc_{u,v}$ contains all $C'$-links with endpoints $u$ and $v$.
\cref{lem:tripod4} ensures that every link in a bucket has at most $t$
vertices and, consequently, every bucket has size $\Oh(n^t)$ and can be enumerated
in polynomial time.
Note that $\Lc_{u,v}$ is nonempty for every $\{u,v\} \in \binom{N_H(C')}{2}$.

For $\epsilon > 0$, a vertex $x \in V(H)$ is \emph{$\epsilon$-heavy}
if $N_H[x]$ intersects strictly more than an $\epsilon$ fraction of $C'$-links in at least
an $\epsilon$ fraction of buckets. 
We prove the following.

\begin{lemma}\label{lem:tripod6}
There exists a $\frac{1}{2t}$-heavy vertex.
\end{lemma}
\begin{proof}
Let $H' \coloneqq H[N_H[C']]$; note that $C' \subseteq N_H[C'] \subseteq C' \cup K$. 
Apply \cref{lem:separator} to $H'$ with $A = N_H(C')$, obtaining a set $Y'$ of size at most $t$ such that
every connected component of $H'-N[Y']$ contains at most $|N_H(C')|/2$ vertices of $N_H(C')$.
Consequently, $N_{H'}[Y']$ intersects all links in at least half of the buckets $\Lc_{u,v}$.
We infer that there is $y \in Y'$ such that $N_{H'}[y]$ intersects at least a $\frac{1}{t}$ fraction
of links in at least $\frac{1}{2t}$ fraction of all nonempty buckets.
This completes the proof.
\end{proof}

\cref{lem:tripod6} allows us to choose a $\frac{1}{2t}$-heavy vertex as the branching pivot. 

It remains to show that with this choice of the branching pivot, the subproblem tree
generated by the secondary branching strategy has $\Oh(\log^2 n)$ success branches
on any root-to-leaf path.

As argued, for a branch node $x$, if there exists 
a chip $C'$ in $H\coloneqq G[\VtxActive^{\rcall(x)}]$ with $|N_H(C')| = 1$,
then the pivot is the unique element
of $N_H(C')$ and the secondary branching strategy terminates
both in the failure child and in all success grandchildren.
Thus, on any root-to-leaf path there is at most one such branch node.

For every other branch node $x$, the pivot is a $\frac{1}{2t}$-heavy~vertex. Let $H(x) := G[\VtxActive^{\rcall(x)}]$ and let $C'(x)$ be the chip of $H(x)$.
The \emph{secondary level} of a node $x$ is defined as
\[
\lambda(x) \coloneqq \left\lfloor \log_2\left(1 + \binom{|N_{H(x)}(C'(x))|}{2}\right)\right\rfloor.
\]
Note that the secondary level is positive
if and only if there is a nonempty bucket. 
Since at the root node there are $\Oh((n^\rcall)^2)$ nonempty buckets, there
are $\Oh(\log n^\rcall)$ possible secondary levels. 
Furthermore, during the recursion the buckets can only shrink, so the secondary level of an ancestor
is never lower than the secondary level of a descendant. 
Hence, it suffices to prove that any upward path in the subproblem tree
on which all branching nodes have the same secondary level 
contains $\Oh(\log n^\rcall)$ success branches.

To this end, consider the following potential for a node $x$.
\[
\wh{\mu}(x) \coloneqq \sum_{\{u,v\} \in \binom{N_{H(x)}(C'(x))}{2} }
\log_2\left[
\sum_{Q \in \Lc_{u,v}^{G[\VtxActive^{\rcall(x)}]}} \sum_{q \in Q} \left(1 + \quota(\degord^{\rcall(x)}, q, \poslimit^{\rcall(x)}(q))\right)\right].
\]
For every node $x$ of secondary level $\lambda$, 
there are between $2^{\lambda}-1$ and $2^{\lambda+1}-2$ nonempty buckets of $G[\VtxActive^{\rcall(x)}]$. 
Hence, for node $x$ of secondary level $\lambda$, 
$$\wh{\mu}(x) \leq (2^{\lambda+1}-2) \cdot \Oh(\log n^\rcall) = 2(2^\lambda-1) \cdot \Oh(\log n^\rcall).$$
Let $x$ be a branch node and $y$ its successful grandchild, both with the same secondary level.
Noting that the innermost sums in the definition of $\wh{\mu}$ are lower-bounded by $1$ and upper-bounded by $(d+1)t$, from the definition of a $\frac{1}{2t}$-heavy vertex and \cref{lem:success-quota} we infer that
$$\wh{\mu}(x) - \wh{\mu}(y) \geq \frac{c}{dt^3} \cdot (2^{\lambda}-1),$$
for some universal constant $c>0$.
Since the potential $\mu$ never becomes negative, it can decrease only $\Oh(\log n^\rcall)$ times
at a successful branch when the secondary level $\lambda$ is fixed. Since there are $\Oh(\log n^\rcall)$ secondary levels in total, we conclude that every root-to-leaf path contains $\Oh(\log^2 n^\rcall)$ success branches.

This completes the proof of the properties of the
secondary branching strategy and thus of \cref{lem:branching-strategy}.

%% file: deg2tw.tex
It is well known that if a graph $G$ has treewidth $k$, then its degeneracy it at most $k$.
However, these parameters can be arbitrarily far away from each other: for instance, 3-regular expanders have degeneracy 3 and treewidth linear in the number of vertices~\cite{DBLP:journals/jct/GroheM09}.
In this section we prove that if we restrict our attention to $C_{>t}$-free graphs, the treewidth is bounded by a function of degeneracy.
In particular, we show \cref{thm:deg2tw}.

\degistw*

Before we proceed to the proof of  \cref{thm:deg2tw}, let us recall the notion of brambles.
Recall that two sets $A,b$ are adjacent if either $A \cap B \neq \emptyset$ or there is an edge with one endpoint in $A$ and the other in $B$.
For brevity, we say that a set $A$ is adjacent to a vertex $v$ if $A$ is adjacent to $\{v\}$, i.e.,  either $v \in A$ or $v$ is adjacent to some vertex of $A$.
A \emph{bramble} of size $p$ in a graph $G$ is a collection  $\Bb = (B_1,B_2,\ldots,B_p)$ of nonempty vertex subsets such that
each $B_i$ induces a connected graph and all $B_i$s are pairwise adjacent. The sets $B_i$ are called \emph{branch sets}.
The \emph{order} of a bramble $\Bb$ is the size of a smallest set of vertices that hits all branch sets.
Observe that the size of a bramble is always at least its order.
We will use the following result of Hatzel \emph{et al.}~\cite{DBLP:journals/corr/abs-2008-02133}, which in a graph of large treewidth constructs a bramble of large order in which no vertex participates in more than two branch sets.

\begin{theorem}[Hatzel \emph{et al.}~\cite{DBLP:journals/corr/abs-2008-02133}]\label{lem:bramble}
There exists a polynomial $\mathsf{p}( \cdot)$ such that for every positive integer $k$,
every graph $G$ of treewidth at least $k$ contains a bramble $\Bb$ of order at least $\sqrt{k}/\mathsf{p}(\log k)$
such that each vertex of $G$ is in at most two branch sets of $\Bb$.
\end{theorem}

We now proceed to the proof of \cref{thm:deg2tw}.
Without loss of generality we may assume that $t$ is even, $t\geq 4$, and $d \geq 2$.
For contradiction, suppose that $G$ is a $C_{>t}$-free graph with degeneracy at most $d$ and treewidth larger that 
$$k \coloneqq \left(500~000 \cdot d^2 t^5\right)^{4t+4} \cdot \left[\mathsf{p}\left(\log \left((500~000 \cdot d^2 t^5)^{4t+4} \right)\right)\right]^4,$$ where $\mathsf{p}( \cdot )$ is the polynomial provided by \cref{lem:bramble}.
Thus, by applying \cref{lem:bramble} to $G$
we obtain a bramble $\Bc = (B_1,B_2,\ldots,B_p)$ of order
\[p > \frac{\sqrt{k}}{\mathsf{p}(\log k)} \geq \left(500~000 \cdot d^2 t^5\right)^{2t+2}.\]

Note that we can assume that each branch set of $\Bc$ is inclusion-wise minimal (subject to $\Bc$ being a bramble),
as otherwise we can remove some vertices from branch sets.
Therefore, for each branch set $B_i$ and each vertex $v$ of $B_i$,
either there is some branch set $B_j$ which is adjacent to $v$ but nonadjacent to $B_i \setminus \{v\}$, or
$v$ is a cutvertex in $G[B_i]$ and its role is to keep the branch set connected.

\begin{claim} \label{cl:diameter}
For each $i \in [p]$, and all $u,v \in B_i$, the distance between $u$ and $v$ in $G[B_i]$ is at most $t$.
\end{claim}
\begin{claimproof}
For contradiction, suppose that there is $B_i$ violating the claim.
Let $u,v$ be the vertices at maximum distance in $G[B_i]$, by assumption this distance is at least $t+1$.
As $u$ and $v$ are the ends of a maximal path in $G[B_i]$, none of them is a cutvertex in $G[B_i]$.
Thus there is a branch set $B_u$ which is adjacent only to $u$ in $B_i$, and another branch set which is adjacent only to $v$ in $B_i$.
Recall that $B_u \cup B_v$ is connected and nonadjacent to $B_i \setminus \{u,v\}$.
So by concatenating a shortest $u$-$v$-path in $B_i$ and a shortest $u$-$v$-path in $B_u \cup B_v$, we obtain an induced cycle with at least $t+1$ vertices, a contradiction.
\end{claimproof}

Let $G'$ be the lexicographic product $G \bullet K_2$: the graph obtained from $G$ by introducing, for each $x \in V(G)$, a copy $x'$ of $x$ and making it adjacent to $x$, all neighbors of $x$, and all their copies. Note that in $G'$, the copy $x'$ is a true twin of $x$.
Observe also that the degeneracy of $G'$ is at most $2d+1$: we can modify a $d$-degeneracy ordering of $G$ into a $(2d+1)$-degeneracy ordering of $G'$ by inserting each vertex $x'$ immediately after $x$.

\begin{claim}\label{cl:minor}
The graph $G'$ contains $K_p$ as a depth-$t$ minor.
\end{claim}
\begin{claimproof}
We construct a family $\Bc' = (B'_1,B'_2,\ldots,B'_p)$ as follows.
We start with $B'_i \coloneqq B_i$ for all $i\in [p]$ and we iteratively inspect every vertex $x$ of $G$. If $x$ belongs to more than one of the sets $\{B_1,\ldots,B_p\}$, then, by the properties given by \cref{lem:bramble}, $x$ must belong to exactly two of them, say $x\in B_i\cap B_j$ for some $i\neq j$. Then replace $x$ with $x'$ in $B_j'$, thus making $B_i'$ and $B_j'$ not overlap on $x$. 

It is clear that once this operation is applied to each vertex of $G$, the resulting sets of $\Bc'$ are pairwise disjoint and pairwise adjacent.
Further, for each $i \in [p]$ the graph $G'[B'_i]$ is isomorphic to $G[B_i]$, as we only replaced some vertices by their true twins, so in particular $G'[B'_i]$ is connected.
Therefore, $\Bc'$ is a minor model of a clique of order $p$ in $G'$. By \cref{cl:diameter}, the radius of each graph $G'[B'_i]$ is at most~$t$, hence this model has depth at most $t$.
\end{claimproof}

The next result binds the maximum size of a bounded-depth clique minor and the maximum size of a bounded-depth topological clique minor that can be found in a graph.
It is a fairly standard fact used in the sparsity theory; for the proof, see e.g.~\cite[Lemma 2.19 and Corollary 2.20]{sparsity-notes}.

\begin{proposition} \label{prop:topominor}
Let $G$ be a graph and let $t,p,p'$ be integers such that $p \geq 1+(p'+1)^{2t+2}$.
If $G$ contains $K_p$ as a depth-$t$ minor, then $G$ contains $K_{p'}$ as a depth-$(3t+1)$ topological minor.
\end{proposition}

By combining \cref{cl:minor} and \cref{prop:topominor}, we conclude that $G''$ contains $K_{p'}$ as a topological depth-$(3t+1)$ minor, where
\[p' \coloneqq \left \lfloor \frac{p^{\frac{1}{2t+2}}}{4} \right \rfloor \geq  100~000 \cdot d^2\cdot t^5.\]

Fix some topological depth-$(3t+1)$ minor model of $K_{p'}$ in $G'$.
Let $R$ be the set of roots of the minor model and consider the graph $G'[R]$.
It has $p'$ vertices and, as a subgraph of $G'$, is $(2d+1)$-degenerate.
Therefore, there is an independent set $R'$ in $G'[R]$ of size at least 
\[
p'' \coloneqq \left \lceil \frac{p'}{2d+2} \right \rceil \geq \frac{100~000 \cdot d^2 \cdot t^5}{2d+2} \geq 20~000 \cdot d \cdot t^5.
\]
Observe that restricting our minor model only to the roots that are in $R'$ and paths incident to them gives us a topological depth-$(3t+1)$ minor model of $K_{p''}$ with the additional property that the roots are pairwise nonadjacent.

Let $H$ be the subgraph of $G'$ induced by the vertices used by the topological minor model obtained in the previous step.
Let $X$ be the set of vertices of $H$ with degree larger than $200 \cdot d \cdot t^2$, which are not roots.
Since $H$ is $(2d+1)$-degenerate, we observe that 
\[|X| \leq \frac{(2d+1)|V(H)|}{100 \cdot dt^2} \leq \frac{(2d+1)(6t+3) \binom{p''}{2}}{100 \cdot dt^2} \leq \frac{20}{100t} \binom{p''}{2} = \epsilon \cdot \binom{p''}{2}, \text{ where } \epsilon \coloneqq \frac{1}{5t}.\]

Let $H'$ be obtained from $H$ by removing all vertices in $X$, along with all paths from the topological minor model which contain a vertex from $X$.
Note that thus, we have removed at most $\epsilon \cdot \binom{p''}{2}$ paths.

Observe that $H'$ still contains a depth-$(3t+1)$ topological minor model of some graph $Z$ with $p''$ vertices and at least 
\[\binom{p''}{2} - |X| \geq \binom{p''}{2} - \epsilon \binom{p''}{2} = (1-\epsilon)\binom{p''}{2} \]
edges. Thus, the average degree of a vertex in $Z$ is at least $(1-\epsilon)(p''-1)$.

Let $\Wc = (v_0, v_1, \ldots, v_{t/2})$ be a sequence of vertices of $Z$, chosen independently and uniformly at random. 
In what follows, all arithmetic operations on the indices of the vertices $v_i$ are computed modulo $t/2+1$, in particular $v_{t/2+1} = v_0$. 

We prove that with positive probability, $\Wc$ has the following four properties:
\begin{enumerate}[label=(P\arabic*),ref=(P\arabic*),leftmargin=*]
\item The vertices $v_i$ are pairwise distinct.\label{prop:distinct}
\item For every $0 \leq i \leq t/2$, $v_iv_{i+1}$ is an edge of $Z$; let $P_i$ be the corresponding path in $H'$. \label{prop:closed}
\item For every $0 \leq i \leq t/2$ and $0 \leq j \leq t/2$ such that $j \notin \{i, {i+1}\}$, 
 the internal vertices on the path $P_i$ are anti-adjacent to $v_j$. \label{prop:anti1}
\item For all $0 \leq i < j \leq t/2$, the internal vertices of $P_i$ are anti-adjacent to the internal vertices of $P_j$.\label{prop:anti2}
\end{enumerate}
Observe that these four properties imply that the concatenation of all paths $P_i$ is a hole of length more than $t$ in $G'$ (recall here that the roots of the minor model are independent in $H'$). The assumption that $G$ is $C_{>t}$-free implies that $G'$ is $C_{>t}$-free as well, hence this will be a contradiction.

For~\ref{prop:distinct}, since $p'' \geq 20~000 \cdot d \cdot t^5$, by the union bound the probability that $v_i = v_j$ for some $i \neq j$ is at most $\binom{t}{2}/p'' <0.1$. 


For~\ref{prop:closed}, since $v_i$ and $v_{i+1}$ are independently chosen vertices, and $Z$ has at least $(1-\epsilon)\binom{p''}{2}$ edges, 
the probability that $v_iv_{i+1}$ is not an edge of $Z$ is bounded by $\epsilon = \frac{1}{5t}$.
By the union bound, the probability that \ref{prop:closed} does not hold is bounded by $\epsilon \cdot (t/2+1) \leq 0.2$.

For~\ref{prop:anti1}, fix $0 \leq i \leq t/2$ and assume $v_iv_{i+1} \in E(Z)$ so that $P_i$ is defined. 
Then, the total number of neighbors of the internal vertices of $P_i$ is bounded by $(6t+3) \cdot 200 \cdot d \cdot t^2 \leq 2000 \cdot d \cdot t^3$. 
Since $v_j$ is a vertex of $V(Z)$ chosen at random independently of the choice of $v_i$ and $v_{i+1}$, the probability
that $v_j$ is among these neighbors is bounded by $2000 \cdot dt^3 / p'' \leq 0.1/t^2$. 
By the union bound, \ref{prop:closed} holds but \ref{prop:anti1} does not hold with probability at most $t(t-2) \cdot \frac{0.1}{t^2} \leq 0.1$.

For~\ref{prop:anti2}, fix $0 \leq i < j \leq t/2$.
Note that it may be possible that $i+1 = j$ or $j+1 = i$ (cyclically modulo $t/2+1$), but not both. 
Hence, by symmetry between $i$ and $j$, assume that the choice of $v_{j+1}$ is independent of the choices of $v_i$, $v_{i+1}$, and $v_j$. 
Assume that $v_iv_{i+1} \in E(Z)$ so that $P_i$ is defined. As in the previous paragraph,
there are at most $2000dt^3$ neighbors in $H'$  of the internal vertices of $P_i$.
There are $p'' = |V(Z)|$ choices for $v_{j+1}$, all of them leading to either $v_jv_{j+1} \notin E(Z)$ or to vertex-disjoint (except for $v_j$) choices of the path $P_j$.
Hence, for at most $2000dt^3$ of these choices, we have $v_jv_{j+1} \in E(Z)$ but there is an edge between an internal vertex of $P_j$ and an internal vertex of $P_i$. 
By the union bound, \ref{prop:closed} holds but \ref{prop:anti2} does not hold with probability less than $\binom{t/2+1}{2} \cdot \frac{2000dt^3}{p''} \leq \binom{t/2+1}{2} \cdot \frac{2000dt^3}{20~000 dt^5} \leq 0.1$. 

By the union bound over all the above cases, $\Wc$ satisfies all properties \ref{prop:distinct}--\ref{prop:anti2} with probability at least $1 - 0.1 -  0.2 - 0.1 - 0.1 = 0.5$. 
This gives the desired contradiction and completes the proof.

%% file: longhole.tex
In this section we prove \cref{thm:main}
using the branching strategy described in \cref{sec:branching}.

\subsection{Extending subproblems}

Fix integers $k \geq 1$ and $t \geq 3$, \cmsotwo sentence $\phi$, and a $C_{>t}$-free graph $G$.
Let $q$ be the maximum of the quantifier rank of $\phi$, the
quantifier rank of $\phi_{\tw < k}$, where $\phi_{\tw < k}$ comes from \cref{lem:mso-tw-formula}, and $4$ (so that we can use \cref{lem:msotwo-connected}).
Let $p$ be the largest modulus used in $\phi$, or $0$ if $\phi$ does not contain any modular atomic expression.
We also denote $k' \coloneqq 6k$. In what follows, we will be using \cmsotwo types $\msotypes^{k',p,q}$, as given by \cref{lem:mso-type}. Hence, by ``\cmsotwo types'' we mean elements of $\msotypes^{k',p,q}$, and we drop the super- or subscript $k',p,q$ when it is clear from the context. 

Recall that every graph of treewidth less than $k$ is $(k-1)$-degenerate.
Hence, we will use the branching strategy provided \cref{lem:branching-strategy} for $d \coloneqq k-1$ and $t$. More precisely, the algorithm executes the branching, i.e. decides on the types of nodes, chooses branching pivots, etc., exactly as prescribed by the strategy given by \cref{lem:branching-strategy}. However, we enrich the subproblems with some additional piece of information, which intuitively encodes a skeleton of a decomposition of the subgraph induced by the constructed solution.
From the recursive subcalls, we expect returning quite an elaborate result: intuitively, optimum solutions to the subproblem of every possible \cmsotwo type.
As we will check at the end, the properties asserted by \cref{lem:branching-strategy} ensure a quasipolynomial running time bound. 

Let us start by formally augmenting the notion of a subproblem
with extra information. 
At every leaf, split, and branch node $x$, the subproblem $\rcall \coloneqq \rcall(x)$ contains additionally the following:
\begin{enumerate}
\item a tree decomposition $(T^\rcall,\beta^\rcall)$ of $G[\VtxTaken^\rcall]$
with maximum bag size at most $k' = 6k$ and the root node having an empty bag;
\item a labelling $\iota^\rcall \colon \VtxTaken^\rcall \to [k']$ such that
for every $a \in V(T^\rcall)$, $\iota^\rcall$ restricted to $\beta^\rcall(a)$ is injective.
\end{enumerate}
We call such a subproblem an \emph{extended subproblem}. Leaf, split, and branch nodes are called \emph{extendable~nodes}.

For such an extended subproblem $\rcall$, 
a set $S \subseteq \VtxActive^\rcall$ is a \emph{feasible solution}
if for every connected component $C$ of $G[S]$,
we have $|N_G(C) \cap \VtxTaken^\rcall| \leq 4k$ and there exists a 
node $a \in V(T^\rcall)$ such that $N_G(C) \cap \VtxTaken^\rcall \subseteq \beta^\rcall(a)$.
Observe that from the properties of a tree decomposition it follows
that the family of those nodes $a \in V(T^\rcall)$ for which
$N_G(C) \cap \VtxTaken^\rcall \subseteq \beta^\rcall(a)$ is a connected
subtree of $T$. By $a^\rcall(C)$ we denote the highest such node $a$.

For an extended subproblem $\rcall$, a
\emph{type assignment} is a function
$\typetree \colon V(T^\rcall) \to \msotypes$.
A feasible solution $S$ is \emph{of type $\typetree$} if
for every $a \in V(T^\rcall)$, 
the subgraph of $G$ induced by
$$ \extS^\rcall(a, S) \coloneqq \beta^\rcall(a) \cup \bigcup \{C \in \cc(G[S])~|~a^\rcall(C) = a\} $$
equipped with the labelling $\iota^\rcall|_{\beta^\rcall(a)}$ on boundary $\beta^\rcall(a)$
is of \cmsotwo type $\typetree(a)$. 

The branching strategy will return, at every extendable node $x$,
for every type assignment
$\typetree$ at~$x$, a feasible solution $\Stab[x,\typetree]$ of type $\typetree$.
We allow $\Stab[x,\typetree] = \bot$, indicating that no such feasible solution was found,
and we use the convention that the weight of $\bot$ is $-\infty$.
In the algorithm description the following operation will be useful when defining
$\Stab[x,\cdot]$ for a fixed node $x$: given a current state of the table $\Stab[x,\cdot]$
and a feasible solution $S$ of type $\typetree$, \emph{updating $\Stab[x,\typetree]$ with
$S$} is an operation that sets $\Stab[x,\typetree] \coloneqq S$ if the weight of $S$
is larger than the weight of the former value of $\Stab[x,\typetree]$.

In the root of the recursion $r$, the considered extension is the trivial one: the tree decomposition $(T^{\rcall(r)},\beta^{\rcall(r)})$ 
consists of a single node $a^r$ with an empty bag, and labelling $\iota^{\rcall(r)}$ is empty. 
After the computation if finished, we iterate over all type assignments $\typetree$ for $r$
for which $\typetree(a^r) \in \msotype[\phi] \cap \msotype[\phi_{\tw < k}]$
and return the set $\Stab[r,\typetree]$ of maximum weight found (ignoring values $\bot$). 
If no such set $\Stab[r,\typetree]$ is found (all values are $\bot$), the algorithm returns
that there is no such set $S$.

Note that, assuming the recursive strategy indeed maintains the invariant that
$\Stab[x,\typetree]$ is of type $\typetree$, for the returned set $S$ 
we have $G[S] \models \phi$ and $G[S]$ is of treewidth less than $k$. 

\subsection{Extended computation at nodes of the subproblem tree}

Recall that now, the setting is that each extendable node (lead, split, or branch node) is assigned an extended subproblem.
We need to describe (a) how at each node we handle the extended subproblem and what extended subproblems are passed down the subproblem tree; (b) what is the subroutine for handling free nodes; and (c) how the tables $\Stab[\cdot,\cdot]$ are computed along the recursion. As before, each type of a node is handled differently.

\subsubsection{Leaf nodes}

For a leaf node $x$, there is little choice the algorithm could do.

We have $\VtxActive^{\rcall(x)} = \emptyset$, so the only feasible solution is $S = \emptyset$.
There is exactly one type assignment $\typetree$ such that for every $a \in V(T^{\rcall(x)})$ we have that $G[\extS^{\rcall(x)}(a,\emptyset)]$ 
with the labelling $\iota^\rcall|_{\beta^\rcall(a)}$ is of type $\typetree(a)$.
For this labelling, we set $\Stab[x,\typetree] = \emptyset$
and for every other type assignment $\typetree'$ we set $\Stab[x,\typetree'] = \bot$.

\subsubsection{Split nodes}

Let $x$ be a split node. We do not use the power of free nodes below split nodes;
formally, they are dummy free nodes as in the proof of \cref{thm:deg-branching}
that only pass the subproblem to their  single child.
Therefore, in what follows we speak about grandchildren of a split node.

If $x$ has one grandchild $y$, then $\rcall(y)$ differs from $\rcall(x)$ only by its level.
Therefore, given an extension of $\rcall(x)$, we pass the same extension to the grandchild $y$ and for the return value at $x$,
we copy the result returned by the grandchild $y$.

The situation is more interesting if $x$ has two grandchildren $y_1$ and $y_2$. Recall that then $\VtxActive^{\rcall(y_1)}$ and $\VtxActive^{\rcall(y_2)}$ is a partition of $\VtxActive^{\rcall(x)}$, while $\VtxTaken^{\rcall(y_1)}=\VtxTaken^{\rcall(y_2)}=\VtxTaken^{\rcall(x)}$.
Given a subproblem extension of $\rcall(x)$, we pass it without modifications to both $y_1$ and $y_2$.
To compute $\Stab[x,\cdot]$ based on $\Stab[y_1,\cdot]$ and $\Stab[y_2,\cdot]$, we proceed as follows. 

First, we initiate $\Stab[x,\typetree] = \bot$ for every tree assignment $\typetree$ at $x$.
Then, we iterate over all pairs $\typetree_1$ and $\typetree_2$ of type assignments
for $\rcall(x)$ such that both $\Stab[y_1,\typetree_1] \neq \bot$ and $\Stab[y_2,\typetree_2] \neq \bot$. 
Let $\typetree$ be the type assignment for $\rcall(x)$ defined as
$$\typetree(a) = \typetree_1(a) \oplus \typetree_2(a)\qquad\textrm{for every }a \in T^{\rcall(x)}.$$ 
Then observe that as $\VtxActive^{\rcall(y_1)}$ is nonadjacent
to $\VtxActive^{\rcall(y_2)}$, $\Stab[y_1,\typetree_1] \cup \Stab[y_2,\typetree_2]$ is a 
feasible solution for $x$ of type $\typetree$.
We update $\Stab[x,\typetree]$ with $\Stab[y_1,\typetree_1] \cup \Stab[y_2,\typetree_2]$.

\subsubsection{Branch and subsequent free nodes}

Let $x$ be a branch node and recall that we are given a subproblem extension of $\rcall(x)$. For the failure child~$y^x$ of $x$, we pass this subproblem extension to $y^x$ without modifications.

Consider now a success grandchild $s\coloneqq s^x_{\mathcal{D}}$
for $\mathcal{D} = (D,\degord,(D_u)_{u \in D'})$, $D' = D \cup \{\pivot^x\}$. 
We now use the power of the free node $s^x_{\mathcal{D}}$ to guess how the extension at $x$ should be enhanced. 

Formally, we iterate over every possibility of:
\begin{itemize}
\item a partition $\branchpart$ of $D'$ into nonempty subsets;
\item a node $a_B \in V(T^{\rcall(x)})$ and a subset $N_B \subseteq \beta^{\rcall(x)}(a_B)$ of size at most $4k$ for every $B \in \branchpart$;
\item a set $X_B \subseteq \VtxActive^{\rcall(s)}$ of size at most $k$ for every $B \in \branchpart$ so that the sets $(X_B)_{B \in \branchpart}$
are pairwise vertex-disjoint; we denote $X_\branchpart \coloneqq \bigcup_{B \in \branchpart} X_B$;
\item a position guess $\degord^\branchpart$ for $X_\branchpart$ and a left neighbor guess $(D_u^\branchpart)_{u \in X_\branchpart}$ for $X_\branchpart$ and $\degord^\branchpart$.
\end{itemize}
For every choice $\mathcal{C} = (\branchpart, (a_B,N_B,X_B)_{B \in \branchpart},\degord^\branchpart, (D_u^\branchpart)_{u \in X_\branchpart})$ as above, we construct 
a child $\hat{s}^x_{\mathcal{D},\mathcal{C}}$ of $s^x_{\mathcal{D}}$
that is created from $s^x_{\mathcal{D}}$ by 
taking $X_\branchpart$ at positions $\degord^\branchpart$ with left neighbors $(D_u^\branchpart)_{u \in X_\branchpart}$.
Every created child $\hat{s}^x_{\mathcal{D},\mathcal{C}}$ is a filter node;
we denote the resulting grandchild (if present) as $\tilde{s}^x_{\mathcal{D},\mathcal{C}}$.

Denoting $\tilde{s} = \tilde{s}^x_{\mathcal{D},\mathcal{C}}$ for brevity,
the grandchild's subproblem is extended as follows:
\begin{itemize}
\item create $(T^{\rcall(\tilde{s})}, \beta^{\rcall(\tilde{s})})$ from $(T^{\rcall(x)}, \beta^{\rcall(x)})$ by adding, for every $B \in \branchpart$,
a new node $\tilde{a}_B$ with bag $N_B \cup B \cup X_B$ that is a child of $a_B$;
\item create $\iota^{\rcall(\tilde{s})}$ by extending $\iota^{\rcall(x)}$ to the new elements of $\VtxActive^{\rcall(\tilde{s})}$ in any manner
that is injective on the newly created bags. 
\end{itemize}
Observe that every newly created bag is of size bounded by
$4k + k + k = k'$.
Since $\branchpart$ is a partition of $D'$ and we assume that the sets $(X_B)_{B \in \branchpart}$ are pairwise disjoint, for every two newly created nodes
$a_{B_1}'$ and $a_{B_2}'$, we have $\beta^{\rcall(\tilde{s})}(a_{B_1}') \cap \beta^{\rcall(\tilde{s})}(a_{B_2}') \subseteq N_{B_1} \cap N_{B_2} \subseteq \VtxTaken^{\rcall(x)}$,
the pair $(T^{\rcall(\tilde{s})},\beta^{\rcall(\tilde{s})})$ is indeed a tree decomposition of $G[\VtxTaken^{\rcall(\tilde{s})}]$ with maximum bag size at most $k'$, and it is straightforward
to extend $\iota^{\rcall(x)}$ to obtain $\iota^{\rcall(\tilde{s})}$. 

To complete the description of the algorithm, it remains to show how to assemble the table $\Stab[x,\cdot]$ at the branch node $x$ from the tables computed for  
its great-great-grandchildren $\tilde{s}^x_{\mathcal{D},\mathcal{C}}$
and the failure child $y^x$. 

To this end, we initiate $\Stab[x,\typetree] = \Stab[y^x,\typetree]$ for every type assignment $\typetree$ at
$x$ (note that every feasible solution at $y^x$ is also a feasible solution at $x$ and is of the same type). 
Then, for every great-great-grandchild $\tilde{s} \coloneqq \tilde{s}^x_{\mathcal{D},\mathcal{C}}$ and type assignment $\widetilde{\typetree}$ at $\tilde{s}$
such that $\tilde{S} \coloneqq \Stab[\tilde{s}, \widetilde{\typetree}] \neq \bot$, we proceed as follows. 
Let 
$\mathcal{D} = (D,\degord,(D_u)_{u \in D'})$, $D' = D \cup \{\pivot^x\}$,
$\mathcal{C} = (\branchpart, (a_B, N_B, X_B)_{B \in \branchpart}, \degord^\branchpart, (D_u^\branchpart)_{u \in X_\branchpart})$,
  and $X_\branchpart = \bigcup_{B \in \branchpart} X_B$.
Observe that $\VtxTaken^{\rcall(\tilde{s})} \setminus \VtxTaken^{\rcall(x)} = D' \cup X_\branchpart$. 
Let $S \coloneqq \tilde{S} \cup D' \cup X_\branchpart$.

We say that the pair $(\tilde{s}, \tilde{S})$ is \emph{liftable}
if the set $S$ obtained as above is a feasible solution at $x$ and, furthermore, 
for every $B \in \branchpart$ there exists a connected component $\tilde{C}_B$ of $G[S]$ such that
$$\tilde{C}_B = B \cup X_B \cup \bigcup \{\tilde{C} \in \cc(G[\tilde{S}])~|~a^{\tilde{s}}(\tilde{C}) = \tilde{a}_B\},$$
the components $(\tilde{C}_B)_{B \in \branchpart}$ are pairwise distinct, and
$a^x(\tilde{C}_B) = a_B$ for every $B \in \branchpart$. 
For a liftable pair $(\tilde{s}, \tilde{S})$, the set $S$ is called the \emph{lift of $(\tilde{s}, \tilde{S})$}. 

If $(\tilde{s}, \tilde{S})$ is liftable, then we can use \cref{lem:mso-type} to compute 
the type $\typetree_S$ of the lift $S$ at $x$ as follows. 
First, for every $B \in \branchpart$, the \cmsotwo type of
$G[C_B \cup \beta^{\rcall(x)}(a_B)]$ with labelling $\iota^{\rcall(x)}|_{\beta^{\rcall(x)}}(a_B)$
can be computed from $\widetilde{\typetree}(\tilde{a}_B)$ by forgetting the labels
of $B \cup X_B$. 
Second, for each $a \in V(T^{\rcall(x)})$, the type $\typetree_S(a)$ is the composition
of $\widetilde{\typetree}(a)$ and the types of all graphs
$(G[C_B \cup \beta^{\rcall(x)}(a_B)], \iota^{\rcall(x)}|_{\beta^{\rcall(x)}(a_B)})$
for those $B$ for which $a_B = a$. 

For every liftable pair $(\tilde{s}, \tilde{S})$, we compute the lift $S$ and its type $\typetree_S$ at $x$ as above and we update $\Stab[x,\typetree_S]$ with $S$.
This finishes the description of the algorithm at branch nodes.

We conclude this section with an immediate, yet important corollary of the way how we compute the type $\typetree_S$.
\begin{lemma}\label{lem:two-liftable}
Let $x$ be a branch node, $\tilde{s}$ be its great-great-grandchild, and $\tilde{S}_1$ and $\tilde{S}_2$ be two feasible solutions
at $\tilde{s}$ of the same type such that both $(\tilde{s}, \tilde{S}_1)$ and $(\tilde{s},\tilde{S}_2)$ are liftable.
Let $S_i$ be the lift of $(\tilde{s}, \tilde{S}_i)$ for $i=1,2$.
Then $S_1$ and $S_2$ are of the same type at $x$. 
\end{lemma}
\begin{proof}
The aforementioned algorithm to compute the type of the lift of $(\tilde{s}, \tilde{S}_i)$ uses only the type 
of $\tilde{S}_i$ at $\tilde{s}$ and the subproblems $\rcall(x)$ and $\rcall(\tilde{s})$. The claim follows.
\end{proof}

\subsection{Correctness}

Fix a subset $S^\ast \subseteq V(G)$ such that $G[S^\ast]$ is of treewidth less than $k$ and $G[S^\ast] \models \phi$. 
We would like to show that the algorithm returns a set $S$ of weight at least the weight of $S^\ast$ (not necessarily $S^\ast$). 
Clearly, since $G[S^\ast]$ is $(k-1)$-degenerate, we can speak about lucky nodes of the subproblem tree, defined in the same manner as in \cref{sec:branching}. 

A lucky extendable node $x$ is called a \emph{gander} if 
$S^\ast \cap \VtxActive^{\rcall(x)}$ is a feasible solution for $\rcall(x)$. 
For a gander~$x$, define the type assignment $\typetree^x$ as the type of the feasible solution $S^\ast \cap \VtxActive^{\rcall(x)}$. 

Note that the root $r$ of the subproblem tree is a gander and 
$\typetree^r(a^r) \in \msotype[\phi] \cap \msotype[\phi_{\tw < k}]$.
Thus, it suffices to show the following:
\begin{lemma}\label{lem:main-correctness}
For every gander $x$, we have that $\Stab[x, \typetree^x] \neq \bot$
and the weight of $\Stab[x,\typetree^x]$ is at least the weight
of $S^\ast \cap \VtxActive^{\rcall(x)}$.
\end{lemma}
\begin{proof}
The proof proceeds by a bottom-up induction on the subproblem tree.

For a leaf gander $x$, since $x$ is lucky, from the definition of $\typetree^x$
the only type assignment
$\typetree$ for which $\Stab[x, \typetree] \neq \bot$ is exactly $\typetree^x$. 
Consequently, $\Stab[x,\typetree^x] = \emptyset$ and the claim is proven.

For a split gander $x$, the claim is straightforward if $x$ has one grandchild.
Assume then $x$ has two grandchildren $y_1$ and $y_2$. 
Since $\VtxActive^{\rcall(y_1)}$ is nonadjacent to $\VtxActive^{\rcall(y_2)}$, every connected component
of $G[S^\ast \cap \VtxActive^{\rcall(x)}]$ is contained either in $\VtxActive^{\rcall(y_1)}$ or in $\VtxActive^{\rcall(y_2)}$.
It follows that both $y_1$ and $y_2$ are ganders, too. 

By induction, for $i=1,2$ the weight of $\Stab[y_i,\typetree^{y_i}]$ is at least the weight of $S^\ast \cap \VtxActive^{\rcall(y_i)}$. 
Denote $$\typetree(a) \coloneqq \typetree^{y_1}(a) \oplus \typetree^{y_2}(a)\qquad\textrm{for }a \in V(T^{\rcall(x)}).$$ 
Then, on one hand from \cref{lem:mso-type} we have that $\typetree = \typetree^x$, and on the other hand
the handling of split nodes will update $\Stab[x,\typetree]$ with $\Stab[y_1,\typetree^{y_1}] \cup \Stab[y_2,\typetree^{y_2}]$. The claim for split ganders follows.

It remains to analyse branch ganders. Let $x$ be a branch gander. 
If $\pivot^x \notin S^\ast$, then $y^x$ is lucky and it is immediate that it is a gander, too. 
By the inductive assumption, $\wei(\Stab[y^x, \typetree^{y^x}]) \geq \wei(S^\ast \cap \VtxActive^{\rcall(y^x)})$.
Since in this case $\typetree^x = \typetree^{y^x}$ and we initiated $\Stab[x,\typetree^x]$ with $\Stab[y^x,\typetree^x]$, the inductive claim follows.

Assume then that $\pivot^x \in S^\ast$. Define $D$, $\degord$, and $(D_u)_{u \in D'}$ for $D' = D \cup \{\pivot^x\}$ as in \cref{lem:lucky-branch}, that is: 
\begin{align*}
D &= \{u \in N_G(\pivot^x) \cap \VtxActive^{\rcall(x)}\cap S^\ast~|~\degord^\ast(u) < \degord^\ast(\pivot^x)\},\\
\degord &= \degord^\ast|_{\VtxTaken^{\rcall(x)} \cup D'},\\
D_u &= \{w \in N_G(u) \cap \VtxActive^{\rcall(x)}\cap S^\ast \setminus D'~|~\degord^\ast(w) < \degord^\ast(u)\}.
\end{align*}
Then, by \cref{lem:lucky-branch}, the branch node $x$ creates
a child $z^x_{\mathcal{D}}$ for $\mathcal{D} = (D,\degord,(D_u)_{u \in D'})$ and
$z^x_{\mathcal{D}}$ is lucky. Hence, the grandchild $s^x_{\mathcal{D}}$ exists and is lucky as well.

We now want to look at some particular child of the free node $s^x_{\mathcal{D}}$. 
Define the partition $\branchpart$ of $D'$ as the partition induced by the connected components
of $G[S^\ast \cap \VtxActive^{\rcall(x)}]$ on $D'$, that is:
$$\branchpart = \{C \cap D'\colon C \in \cc(G[S^\ast \cap \VtxActive^{\rcall(x)}])\textrm{ such that } C \cap D' \neq \emptyset\}.$$
For every $B \in \branchpart$, let $C_B$ be the connected component of
$G[S^\ast \cap \VtxActive^{\rcall(x)}]$ that contains $B$. 
Define $a_B \coloneqq a^x(C_B)$ for $B \in \branchpart$ and $N_B \coloneqq N_G(C_B) \cap \VtxTaken^{\rcall(x)}$. 
Since $x$ is a gander, we have $|N_B| \leq 4k$ and $N_B \subseteq \beta^{\rcall(x)}(a_B)$. 
Let $X_B^\circ \subseteq C_B \cup N_B$ be a set of size at most $k$
promised by \cref{lem:tw-balanced-sep} for the graph $G[C_B \cup N_B]$
and $A = N_B$; note here that $G[C_B \cup N_B]$ is of treewidth less than $k$, because $C_B,N_B \subseteq S^\ast$.
Let $X_B \coloneqq X_B^\circ \cap C_B$. 
Finally, define $\degord^\branchpart$ and $(D_u^\branchpart)_{u \in X_{\branchpart}}$ 
similarly as for the grandchildren of $x$:
\begin{align*}
\degord^\branchpart &= \degord^\ast|_{\VtxTaken^{\rcall(x)} \cup D' \cup X_\branchpart},\\
D_u^\branchpart &= \{w \in N_G(u) \cap \VtxActive^{\rcall(x)}\cap S^\ast \setminus (D' \cup X_\branchpart)~|~\degord^\ast(w) < \degord^\ast(u)\}.
\end{align*}
Let $\mathcal{C} = (\branchpart, (a_B, N_B, X_B)_{B \in \branchpart}, \degord^\branchpart, (D_u^\branchpart)_{u \in X_\branchpart})$.
Observe that $\hat{s} \coloneqq \hat{s}^x_{\mathcal{D},\mathcal{C}}$ is lucky by 
\cref{lem:lucky-take}, as $X_\branchpart \subseteq S^\ast$. 
Hence, $\tilde{s} \coloneqq \tilde{s}^x_{\mathcal{D},\mathcal{C}}$ exists and is lucky as well.
We claim that $\tilde{s}$ is a gander.

To this end, consider a connected component $\tilde{C}$ of $G[S^\ast \cap \VtxActive^{\rcall(\tilde{s})}]$. 
We want to show that $|N_G(\tilde{C}) \cap S^\ast| \leq 4k$
and there exists a node $\tilde{a} \in V(T^{\rcall(\tilde{s})})$ such that 
$N_G(\tilde{C}) \cap S^\ast \subseteq \beta^{\rcall(\tilde{s})}(\tilde{a})$. 
The claim is straightforward if 
$\tilde{C}$ is also a connected component of $G[S^\ast \cap \VtxActive^{\rcall(x)}]$, because $x$ is a gander.
Otherwise, $\tilde{C} \subsetneq C$ for some connected component $C$ of $G[S^\ast \cap \VtxActive^{\rcall(x)}]$. 
Note that there exists $B \in \branchpart$ such that $B \subseteq C$, that is,
$C = C_B$ for some $B \in \branchpart$.

By the choice of $X_B^\circ$, the connected component of $G[C_B \cup N_B] \setminus X_B^\circ$
that contains $\tilde{C}$ contains at most $|N_B|/2 \leq 2k$ vertices of $N_B$. Consequently, 
$$|N_G(\tilde{C}) \cap S^\ast| \leq |X_B| + |B| + |N_B|/2 \leq k + k + 2k = 4k.$$
Furthermore, 
$$N_G(\tilde{C}) \cap S^\ast \subseteq X_B \cup B \cup N_B \subseteq \beta^{\rcall(\tilde{s})}(\tilde{a}_B).$$
We infer that $\tilde{s}$ is indeed a gander.

Since $S^\ast \cap \VtxActive^{\rcall(\tilde{s})}$ is a feasible solution at $\tilde{s}$, it follows immediately
from the definition of $\mathcal{C}$ that $(\tilde{s}, S^\ast \cap \VtxActive^{\rcall(\tilde{s})})$ is 
liftable with the lift $S^\ast \cap \VtxActive^{\rcall(x)}$. 

Let now $\tilde{S} = \Stab[\tilde{s}, \typetree^{\tilde{s}}]$
and $S = \tilde{S} \cup D' \cup X_\branchpart$; 
note that $\VtxTaken^{\rcall(\tilde{s})} \setminus \VtxTaken^{\rcall(x)} = D' \cup X_\branchpart$. 

We claim that $(\tilde{s},\tilde{S})$ is liftable. 
To this end, fix $B \in \branchpart$.
Observe that:
$$\extS^{\tilde{s}}(\tilde{a}_B, S^\ast \cap \VtxActive^{\rcall(\tilde{s})}) = \beta^{\rcall(\tilde{s})}(a) \cup C_B.$$
Consider the connected component $C_B$ of $G[S^\ast\cap C^{\rcall(x)}]$: it contains $X_B\cup B=\beta^{\rcall(\tilde{s})}(a_B)\setminus \beta^{\rcall(\tilde{s})}(\tilde{a}_B)$, which is nonempty due to $B\neq \emptyset$, and its neighborhood in $\VtxTaken^{\rcall(x)}$ is the set $N_B\subseteq \beta^{\rcall(x)}$. 
Observe that the graph induced by $\extS^{\tilde{s}}(\tilde{a}_B, \tilde{S})$
is of the same type as the graph induced by
$\extS^{\tilde{s}}(\tilde{a}_B, S^\ast \cap \VtxActive^{\rcall(\tilde{s})})$
(both with the boundary labelling $\iota^{\tilde{s}}|_{\beta^{\rcall(\tilde{s})}(\tilde{a}_B)}$). Therefore, due to $X_B\cup B$ being nonempty, by \cref{lem:msotwo-connected} we infer that there exists  a connected component $\tilde{C}_B$ of $G[S]$ such that
$$
\tilde{C}_B=X_B\cup B\cup \bigcup\{C' \in \cc(G[\tilde{S}])~|~a^{\tilde{s}}(C') = \tilde{a}_B\}$$
and
 \begin{equation}\label{eq:liftable} 
  N_G(\tilde{C}_B) \cap \VtxTaken^{\rcall(x)} = N_B=N_G(C_B) \cap \VtxTaken^{\rcall(x)}.
  \end{equation}
Consequently, every connected component $C$ of $G[S]$ that is disjoint with $D'$
is also a connected component of $G[\tilde{S}]$, and hence
$N_G(C) \cap \VtxTaken^{\rcall(\tilde{s})} = N_G(C) \cap \VtxTaken^{\rcall(x)}$.
Since $\tilde{S}$ is a feasible solution at $\tilde{s}$, we infer that 
$S$ is a feasible solution at $x$. 
Furthermore, from~\eqref{eq:liftable} if follows that
for every $B \in \branchpart$ we have $a^x(\tilde{C}_B) = a$. 
Hence, $(\tilde{s},\tilde{S})$ is liftable and $S$ is the lift.

By induction, the weight of
$\tilde{S}$ is not smaller than the weight of $S^\ast \cap \VtxActive^{\rcall(\tilde{s})}$. 
As $D' \cup X_\branchpart \subseteq S^\ast$,
the weight of $S$ is not smaller than the weight of $S^\ast \cap \VtxActive^{\rcall(x)}$.
Since $S$ is liftable, the algorithm updates $\Stab[x, \typetree_S]$ with $S$, where $\typetree_S$
is the type assignment of $S$ at node $x$. 

Since $\tilde{S}$ and $S^\ast \cap \VtxActive^{\rcall(\tilde{s})}$ are of the same type $\typetree^{\tilde{s}}$ at $\tilde{s}$,
it follows from \cref{lem:two-liftable} that $S$ and $S^\ast \cap \VtxActive^{\rcall(x)}$ are of the same type
at $x$, that is, $\typetree_S = \typetree^x$. 
This finishes the induction step for branch ganders and completes the proof of the lemma.
\end{proof}

\subsection{Complexity analysis}

We are left with arguing that the time complexity is as promised.
By \cref{lem:branching-strategy}, the subproblem tree generated by the recursion
has depth $\Oh(n)$, $\Oh(\log^2 n)$ or $\Oh(\log^3 n)$ success branches on any root-to-leaf path depending on whether we work in $P_t$-free or $C_{>t}$-free regime,
and $\Oh(\log n)$ split nodes on any root-to-leaf path.

Note that success branches are the only places where we add nodes to the tree decomposition
$(T^\rcall,\beta^\rcall)$. 
Furthermore, a success branch adds at most $k$ nodes to the tree decomposition, one for each
element of~$\branchpart$. Hence, at every node $x$
we have $|V(T^{\rcall(x)})| = \Oh(\log^3 n)$
and $|V(T^{\rcall(x)})| = \Oh(\log^2 n)$ if $G$ is $P_t$-free.
As $|\msotypes|=\Oh(1)$, there are $2^{\Oh(\log^3 n)}$ type assignments to consider
($2^{\Oh(\log^2 n)}$ if $G$ is $P_t$-free). 

At a free node that is a grandchild of a branch node, the sets $D'$, $\branchpart$,
$X_B$, $N_B$ are of constant size. Consequently, 
every free node has a number of children bounded polynomially in $n$. 

We infer that the whole subproblem tree has size $n^{\Oh(\log^3 n)}$. At every node,
the algorithm spends time $n^{\Oh(\log^3 n)}$ inspecting all type assignments
(or pairs of type assignments in the case of a split node). 
Both bounds improve to $n^{\Oh(\log^2 n)}$ if $G$ is $P_t$-free.
The running time bound follows, and hence the proof of \cref{thm:main} is complete.

\subsection{A generalization}\label{sec:generalization}

We now give a slight generalization of \cref{thm:main} that can be useful for expressing some problems that do not fall directly under its regime. The idea is that together with the solution $S$ we would like to distinguish a subset $M\subseteq S$ that satisfies some \cmsotwo-expressible predicate, and only the vertices of $M$ contribute to the weight of the solution. The proof is a simple gadget reduction to \cref{thm:main}.

\begin{theorem}\label{thm:main2}
Fix a pair of integers $d$ and $t$ and a \cmsotwo formula $\phi(X)$ with one free vertex subset variable. Then there exists an algorithm that, given a $C_{>t}$-free $n$-vertex graph $G$
and a weight function $\wei\colon V(G) \to \mathbb{N}$, in time $n^{\Oh(\log^3 n)}$ finds subsets of vertices~$M\subseteq S$ such that $G[S]$ is $d$-degenerate, $G[S]\models \phi(M)$, and, subject to the above, $\wei(M)$ is maximum~possible; the algorithm may also conclude that no such vertex subsets exist. The running time can be improved to $n^{\Oh(\log^2 n)}$ if $G$ is $P_t$-free.
\end{theorem}
\begin{proof}
For a graph $G$ and a set $M \subseteq V(G)$, 
    the \emph{forked version} of $(G,M)$ is the graph $\forked{G}^M$ created from 
$G$ by attaching three degree-one neighbors to every vertex of $V(G)$
and, additionally, a two-edge path to every vertex of $M$.
If $G$ is weighted, then we assign weights to the vertices of $\forked{G}^M$ so that all of them are zero, except
that for every $v \in M$ the other endpoint of the attached two-edge path
inherits the weight of~$v$. 

Note that the vertices of $V(\forked{G}^M) \setminus V(G)$ are exactly the vertices
of degree one or two in $\forked{G}^M$; 
the vertices of $V(G)$ are of degree at least three in $\forked{G}^M$.
This implies that if $H'$ is a forked version of some other 
graph $H$ and $M \subseteq V(H)$, then the pair $(H,M)$ is defined uniquely and is easy to decode:
\begin{itemize}
\item The vertices of $H$ are exactly the vertices of $H'$ that are of degree at least three, and $H$ is the induced subgraph of $H'$ induced by those vertices.
\item Every vertex of $H$ needs to be adjacent to exactly three vertices of degree~$1$
or to three vertices of degree $1$ and one vertex of degree $2$, which in turn has another neighbor
of degree $1$. The vertices of the latter category are exactly the vertices of $M$.
\end{itemize}
Note that if $G$ is $C_{>t}$-free for some $t \geq 3$, then $\forked{G}^M$ is $C_{>t}$-free as well. Moreover, $\tw(\forked{G}^M) \leq \max(\tw(G), 1)$.

Let $k$, $t$, and $\phi$ be as in \cref{thm:main}. 
Construct a \cmsotwo sentence $\forked{\phi}$ that for a graph $H'$ behaves as~follows:
\begin{itemize}
\item If $H'$ is a forked version of some pair $(H,M)$, then
 $H' \models \forked{\phi}$ if and only if $H \models \phi(M)$.
\item Otherwise, $H' \not\models \forked{\phi}$.
\end{itemize}
Writing $\forked{\phi}$ in \cmsotwo is straightforward: we distinguish vertices of $H$ and $M$ as described above, and then apply $\phi$ relativized to those subsets of vertices.

Let $\phi' \coloneqq \forked{\phi}$ and $k' = k$ if $k \geq 2$, and $\phi' \coloneqq \forked{\phi} \wedge \phi_{\tw < k}$ and $k' = 2$ if $k < 2$ (where the sentence $\phi_{\tw < k}$ comes from \cref{lem:mso-tw-formula}). 
We apply \cref{thm:main2} to $k'$, $t$, and $\phi'$, 
and, given a graph $G$ with weight function $\wei$,
  apply the asserted algorithm to $G' \coloneqq \forked{G}^{V(G)}$ (with weight function $\wei'$),
    obtaining a set $S'$. 
By the construction of $G'$ and $\phi'$, $G'[S']$ must be a forked version
of $(G[S],M)$ for some $M \subseteq S \subseteq V(G)$ such that $G[S] \models \phi(M)$ and $G[S]$ has treewidth less than $k$. Moreover, $\wei'(S') = \wei(M)$.

We return $(S,M)$. 
To see the correctness of this output, note that 
for every $M \subseteq S \subseteq V(G)$ such that $G[S]$
is of treewidth less than $k$ and $G[S] \models \phi(M)$, 
if we denote $H = G[S]$, then
$\forked{H}^M$ is an induced subgraph of $G'$ of treewidth less than $k'$, $\wei'(V(\forked{H}^M)) = \wei(M)$, 
and $\forked{H}^M \models \phi'$. 
\end{proof}

For an example application of \cref{thm:main2}, consider the {\sc{Maximum Induced Cycle Packing}} problem: given an (unweighted) graph $G$, find the largest (in terms of cardinality) collection of pairwise non-adjacent induced cycles in $G$. Consider the following property of a graph $G$ and a vertex subset $M\subseteq V(G)$: $G$ is a disjoint union of cycles and every connected component of $G$ contains exactly one vertex of $M$. It is straightforward to write a \cmsotwo formula $\varphi(M)$ such that $G\models \varphi(M)$ if and only if $G$ and $M$ have this property. Noting that disjoint unions of cycles are $2$-degenerate, we can apply \cref{thm:main2} for the formula $\phi(X)$ to conclude that the {\sc{Maximum Induced Cycle Packing}} problem admits a $n^{\Oh(\log^3 n)}$-time algorithm on $C_{>t}$-free graphs, for every fixed $t$. Here, we endow the input graph with a weight function that assigns a unit weight to every vertex.

%% file: packing.tex
In this final section we present a simple technique for turning polynomial-time and quasipolynomial-time algorithms for {\sc{MWIS}} on $P_t$-free and $C_{>t}$-free graphs into PTASes and QPTASes for more general problem, definable as looking for the largest induced subgraph that belongs to some weakly hyperfinite class. Let us stress that this technique works only for unweighted problems.

We define the {\em{blob graph}} of a graph $G$, denoted~$\Blob{G}$, as the graph defined as follows:
\begin{align*}
V(\Blob{G}) & \coloneqq \{ X \subseteq V(G) ~|~ G[X] \text{ is connected} \},\\
E(\Blob{G}) & \coloneqq \{ X_1X_2 ~|~ X_1 \text{ and } X_2 \text{ are adjacent} \}.
\end{align*}
The main combinatorial insight of this section is the following combinatorial property of $\Blob{G}$.
Let us point out that a similar result could be derived from the work of Cameron and Hell~\cite{DBLP:journals/mp/CameronH06}, although it is not stated there explicitly.

\begin{theorem}\label{thm:gstar}
Let $G$ be a graph. The following hold.
\begin{enumerate}[label=(S\arabic*),ref=(S\arabic*),leftmargin=*]
\item The length of a longest induced path in $\Blob{G}$ is equal to the length of a longest induced path in $G$. \label{it:path}
\item The length of a longest induced cycle in $\Blob{G}$ is equal to the length of a longest induced cycle in $G$, with the exception 
that if $G$ has no cycle at all ($G$ is a forest), then $\Blob{G}$ may contain triangles, but it has no induced cycles of length larger than $3$ (i.e. it is a chordal graph). \label{it:cycle}
\end{enumerate}
\end{theorem}
\begin{proof}
Note that since $G$ is an induced subgraph of $\Blob{G}$ (as witnessed by the mapping $u\mapsto \{u\}$), we only need to upper-bound the length of a longest induced path (resp., cycle) in $\Blob{G}$ by the length of a longest induced path (resp. cycle) in $G$.

Let $\Blob{P}=X_1,X_2,\ldots,X_t$ be an induced path in $\Blob{G}$.
We observe that the graph $G[\bigcup_{j=1}^t X_j]$ is connected and for each $j' \in [t-2]$ the sets $\bigcup_{j=1}^{j'} X_j$ and $\bigcup_{j=j'+2}^t X_j$ are nonadjacent.

Fix an induced path $P=v_1,v_2,\ldots,v_p$ in $G[\bigcup_{i=1}^t X_i]$. We define 
\[\reach(i) \coloneqq \max\{ j ~|~ \{v_1,v_2,\ldots,v_i\} \cap X_j \neq \emptyset\}.\]

\begin{claim}\label{cl:creep}
For all $i \in [p-1]$ it holds that $\reach(i+1) \in \{\reach(i), \reach(i)+1\}$.
\end{claim}
\begin{claimproof}
It is clear that $\reach(i+1) \geq \reach(i)$, so suppose $\reach(i+1) \geq \reach(i)+2$.
Since $v_iv_{i+1}$ is an edge of $G$, we conclude that there is an edge in $\Blob{G}$ between the sets $\{ X_j ~|~ j \leq \reach(i)\}$ and $\{X_j ~|~j \geq \reach(i)+2\}$, a contradiction with $\Blob{P}$ being induced.
\end{claimproof}

The following claim encapsulates the main idea of the proof.

\begin{claim}\label{cl:path}
Let $\Blob{P} = X_1,X_2,\ldots,X_t$ be an induced path in $\Blob{G}$ such that $X_1 \not\subseteq X_2$.
Let $X_1' \subseteq X_1 \setminus X_2$ and $X'_t \subseteq X_t$ be nonempty sets. 
Let $P=v_1,v_2,\ldots,v_p$ be a shortest path in $G[\bigcup_{j=1}^t X_j]$ such that $v_1 \in X_1'$ and $v_p \in X_t'$.
Then $P$ is induced, $p \geq t$, and $\{v_2,v_3,\ldots,v_{p-1}\} \cap (X'_1 \cup X'_t) = \emptyset$.
\end{claim}
\begin{claimproof}
The path $P$ is induced and $\{v_2,v_3,\ldots,v_{p-1}\} \cap (X'_1 \cup X'_t) = \emptyset$ by the minimality assumption.
Recall that $X_1$ must be disjoint with $\bigcup_{j=3}^t X_j$.
Thus $\reach(1) = 1$ and $\reach(p)=t$, so the claim follows from \cref{cl:creep}.
\end{claimproof}

Now we are ready to prove \ref{it:path}. Our goal is to prove that if $\Blob{G}$ contains an induced path on $t$ vertices, then so does $G$. If $t=1$, then the statement is trivial,
so assume that $t \geq 2$ and let $\Blob{P} = X_1,X_2,\ldots,X_t$ be an induced path in $\Blob{G}$.

If $X_1 \not\subseteq X_2$, then we are done by \cref{cl:path} applied to $\Blob{P}$ for $X'_1 = X_1 \setminus X_2$ and $X'_t = X_t$.
So assume that $X_1 \subseteq X_2$ and note that $X_2 \not\subseteq X_1$, for $X_1$ and $X_2$ are two different vertices of $\Blob{P}$.
If $t = 2$, then any edge from $X_1$ to $X_2 \setminus X_1$ is an induced path in $G$ with two vertices;
such an edge exists as $G[X_2]$ is connected. So from now on we may assume $t \geq 3$.

Let $X_2' \subseteq X_2 \setminus X_1$ be such that $G[X_2']$ is a connected component of $G[X_2 \setminus X_1]$
and $X_2'$ and $X_3$ are adjacent.
Such a set exists as $X_3$ is adjacent to $X_2$, but nonadjacent to $X_1$.
Note that $G[X_2]$ being connected implies that there exists a nonempty set $X''_2 \subseteq X'_2$, such that every vertex from $X''_2$ has a neighbor in~$X_1$. Furthermore, $X''_2 \cap X_3 = \emptyset$, as $X_1$ is nonadjacent to $X_3$.
Observe that ${\Blob{\wh{P}}}\coloneqq X_2',X_3,\ldots,X_t$ is an induced path in $\Blob{G}$ with at least $t-1 \geq 2$ vertices, such that $X_2' \not\subseteq X_3$.
Let $P'=v_2,v_3,\ldots,v_p$ be the induced path in $G$ with at least $t-1$ vertices obtained by \cref{cl:path} applied to ${\Blob{\wh{P}}}$, $X''_2$, and $X_t$.
Now recall that $v_2 \in X_2''$, so there is $v_1 \in X_1$ adjacent to $v_2$. Note that $v_1$ is nonadjacent to every $v_i$ for $i > 2$, because $v_i\notin X_2''$ for $i>2$. Thus $P\coloneqq v_1,v_2,\ldots,v_p$ is an induced path in $G$ with at least $t$ vertices.

Now let us prove \ref{it:cycle}. We proceed similarly to the proof of \ref{it:path}.
If $\Blob{G}$ is chordal (every induced cycle is of length $3$), then we are done by the 
exceptional case of the statement. Otherwise, 
let $\Blob{C}=X_1,X_2,\ldots,X_t$ be an induced cycle in $\Blob{G}$ for some $t \geq 4$; we want to find an induced cycle of length at least $t$ in $G$.
Note that $X_t \not\subseteq X_{t-1}$ and $X_t \not\subseteq X_{1}$, as otherwise $\Blob{C}$ is not induced.
We observe that there are nonempty sets $X^1_t \subseteq X_t$ and $X^{t-1}_t \subseteq X_t$,
such that every vertex from $X^1_t$ has a neighbor in $X_1$ and every vertex from $X^{t-1}_t$ has a neighbor in $X_{t-1}$.
Let $Q$ be a shortest path contained in $X_t$ whose one endvertex, say $x^1$ is in $X^1_t$ and the other endvertex,
say $x^{t-1}$ is in $X^{t-1}_t$. Note that it is possible that $x^1=x^{t-1}$.
The minimality of $Q$ implies that no vertex of $Q$, except for $x^1,x^{t-1}$, has a neighbor in $\bigcup_{j=1}^{t-1} X_j$.

Let $\Blob{P}$ be the induced path $X_1,X_2,\ldots,X_{t-1}$.
Denote $X'_1 \coloneqq N(x^1) \cap X_1$ and $X'_{t-1} \coloneqq N(x^{t-1}) \cap X_{t-1}$.
Recall that both these sets are nonempty and $X'_1 \cap X_2 = \emptyset$ and $X'_{t-1} \cap X_{t-2} = \emptyset$.
Let $P=v_1,v_2,\ldots,v_p$ be the induced path given by \cref{cl:path} for $\Blob{P}$, $X'_1$, and $X'_{t-1}$.
Recall that $p \geq t-1$.
Now let $C$ be the cycle obtained by concatenating $P$ and $Q$, and observe that the cycle $C$ is induced.
Furthermore, as $P$ has at least $t-1$ vertices and $Q$ has at least one vertex, $C$ has at least $t$ vertices, which completes the proof.
\end{proof}

Let us define an auxiliary problem called \MIP.
An instance of \MIP is a triple $(G,\Fc, \wei)$, where $G$ is a graph,
$\Fc$ is a family of connected induced subgraph of~$G$,
and $\wei \colon \Fc \to \mathbb{R}_+$ is a weight function.
A {\em{solution}} to $(G,\Fc,\wei)$ is a set $X \subseteq V(G)$, such that
\begin{itemize}
\item each connected component of $G[X]$ belongs to $\Fc$; and
\item $\sum_{C\colon \text{ component of } G[X]} \wei(C)$ is maximized.
\end{itemize}
We observe the following.

\begin{theorem}\label{thm:packing}
Let $(G,\Fc,\wei)$ be an instance of \MIP, where $|\Fc|=N$.
\begin{enumerate}
\item If $G$ is $P_t$-free for some integer $t$, then the instance $(G,\Fc,\wei)$ can be solved in time $N^{\Oh(\log^2 N)}$.
\item If $G$ is $C_{>t}$-free for some integer $t$, then the instance $(G,\Fc,\wei)$ can be solved in time $N^{\Oh(\log^3 N)}$.
\item If $G$ is $P_6$-free or $C_{>4}$-free, then the instance $(G,\Fc,\wei)$ can be solved in time $N^{\Oh(1)}$.
\end{enumerate}
\end{theorem}
\begin{proof}
Let $G'$ be the subgraph of $\Blob{G}$ induced by $\Fc$. Clearly, $G'$ has $N$ vertices.
We observe that solving the instance $(G,\Fc,\wei)$ of \MIP is equivalent to solving the instance $(G',\wei)$ of {\sc{MWIS}}.
Now the theorem follows from \cref{thm:gstar} and the fact that {\sc{MWIS}} can be solved in time $n^{\Oh(\log^2n)}$ in $n$-vertex $P_t$-free graphs~\cite{GL20,PPR20SOSA}, 
in time $n^{\Oh(\log^3n)}$ in $n$-vertex $C_{>t}$-free graphs, using \cref{thm:main} only for \textsc{MWIS}, 
and in polynomial time in $P_6$-free~\cite{GrzesikKPP19} or $C_{>4}$-free graphs~\cite{ACPRzS}.
\end{proof}

As an example of an application of \cref{thm:packing}, we obtain the following corollary.
\begin{corollary}
For every fixed $d$ and $t$, given an $n$-vertex $P_t$-free graph $G$, 
in time $n^{\Oh(\log^2 n)}$ we can find the largest induced subgraph of $G$ with maximum degree at most $d$.
\end{corollary}
\begin{proof}
Note that every connected $P_t$-free graph with maximum degree at most $d$ has at most $d^t$ vertices.
Thus, the family $\Fc$ of all connected induced subgraphs of $G$ with maximum degree at most $d$ has size at most~$N\coloneqq n^{d^t}$ and can be enumerated in polynomial time. For each $F \in \Fc$ set $\wei(F) \coloneqq |V(F)|$. We may now
apply \cref{thm:packing} to solve the instance $(G,\Fc,\wei)$ of \MIP in time $N^{\Oh(\log ^2 N)} = n^{\Oh(\log^2 n)}$.
\end{proof}

Note that the strategy we used to prove \cref{thm:packing} cannot be used to solve \MIF in quasipolynomial time,
as there can be arbitrarily larger $P_t$-free tree; consider, for instance, the family of stars.
However, it is sufficient to obtain a simple QPTAS for the unweighted version of the problem.

A class of graphs $\Cc$ is called \emph{weakly hyperfinite} if for every $\epsilon > 0$ there is $c(\epsilon) \in \mathbb{N}$,
such that in every graph $F \in \Cc$ there is a subset $X$ of at least $(1-\epsilon) |V(F)|$ vertices
such that every connected component of $F[X]$ has at most $c(\epsilon)$ vertices~\cite[Section 16.2]{sparsity}.
Weakly hyperfinite classes are also known under the name \emph{fragmentable}~\cite{EDWARDS200130}. Every class closed under edge and vertex deletion which has sublinear separators is weakly hyperfinite~\cite[Theorem 16.5]{sparsity}, hence well-known classes of sparse graphs, such as planar graphs, graphs of bounded genus, or in fact all proper minor-closed classes, are weakly hyperfinite.

For a class $\Cc$ of graphs, by \textsc{Largest Induced $\Cc$-Graph} we denote the following problem: given a graph $G$, find a largest induced subgraph of~$G$, which belongs to $\Cc$. To make the problem well defined, we will always assume that $K_1 \in \Cc$. We can now conclude the following.

\begin{theorem} \label{thm:qptas}
Let $\Cc$ be a nonempty, weakly hyperfinite class of graphs, which is closed under vertex deletion and disjoint union operations.
Then, the \textsc{Largest Induced $\Cc$-Graph} problem 
\begin{enumerate}
\item has a QPTAS in $C_{>t}$-free graphs, for every fixed $t$; and
\item has a PTAS in $P_6$-free graphs and in $C_{>4}$-free graphs.
\end{enumerate}
\end{theorem}
\begin{proof}
Let $n$ be the number of vertices of the given graph $G$ and let $\eps$ be the desired accuracy,
i.e., the goal is to find a solution whose size is at least a $(1-\eps)$ fraction of the optimum.
Let $c \coloneqq c(\epsilon)$.

Let $X^*$ be the vertex set of an optimum solution.
By the properties of $\Cc$, there exists $X' \subseteq X^*$ of size at least $(1-\epsilon)|X^*|$
such that each connected component of $G[X']$ has at most $c$ vertices.
Let $\Fc$ be the set of all connected induced subgraphs of $G$ that have at most $c$ vertices and belong to $\Cc$.
Clearly $|\Fc| \leq n^c$ and $\Fc$ can be enumerated in polynomial time. For each $F \in \Fc$, we set $\wei(F)\coloneqq|V(F)|$.

Apply the algorithm of \cref{thm:packing} to solve the  instance $(G,\Fc,\wei)$ of \MIP in time $n^{\Oh(\log^3 n^c)}=n^{\Oh(\log^3 n)}$ if $G$ is $C_{>t}$-free, or in polynomial time if $G$ is $P_6$-free or $C_{>4}$-free.
Let $X$ be the optimum solution found by the algorithm. As $\Cc$ is closed under the disjoint union operation, we observe that $G[X]$ is a feasible solution to \textsc{Largest Induced $\Cc$-Graph}.
Moreover we have $|X| \geq |X'| \geq (1-\epsilon)|X^*|$.
\end{proof}
%